\documentclass[a4paper,onecolumn,accepted=2023-03-31]{quantumarticle}
\pdfoutput=1

\usepackage[T1]{fontenc}
\usepackage[english]{babel}
\usepackage[utf8]{inputenc}
\usepackage[dvipsnames]{xcolor}
\usepackage[normalem]{ulem}
\usepackage[colorlinks=true,urlcolor=blue,citecolor=blue,linkcolor=blue,anchorcolor=blue]{hyperref}
\usepackage{graphics, graphicx, url, color, physics, cancel, multirow, bbm, soul, mathtools, amsfonts, amssymb, amsmath, amsthm, dsfont, tabularx, bm, verbatim, theoremref, tikz, mathtools, caption, subcaption, xcolor, soul, cite}

\newtheorem{theorem}{Theorem}
\newtheorem{assumption}{Assumption}

\newtheorem{lemma}{Lemma}
\newtheorem{prop}{Conjecture}
\newtheorem{claim}{Claim}
\newtheorem{remark}{Remark}

\def\Var{{\textrm{Var}}\,}
\def\E{{\textrm{E}}\,}
\def\CostXi{{\langle X_i\rangle}\,}
\def\vot{\mathbf{v_{13}}}
\def\votm{\mathbf{v_{13}^-}}
\def\vt{\mathbf{v_{2}}}

\fancyheadoffset{0cm}

\begin{document}

\title{Barren plateaus in quantum tensor network optimization}

\author{Enrique Cervero Mart\'{i}n}
\email{enrique.cervero@u.nus.edu}
\affiliation{Quantinuum, Partnership House, Carlisle Place, London SW1P 1BX, United Kingdom}
\affiliation{Centre for Quantum Technologies, National University of Singapore, 3 Science Drive 2, Singapore 117543}

\author{Kirill Plekhanov}
\affiliation{Quantinuum, Partnership House, Carlisle Place, London SW1P 1BX, United Kingdom}

\author{Michael Lubasch}
\email{michael.lubasch@quantinuum.com}
\affiliation{Quantinuum, Partnership House, Carlisle Place, London SW1P 1BX, United Kingdom}

\begin{abstract}
We analyze the barren plateau phenomenon in the variational optimization of quantum circuits inspired by matrix product states (qMPS), tree tensor networks (qTTN), and the multiscale entanglement renormalization ansatz (qMERA).
We consider as the cost function the expectation value of a Hamiltonian that is a sum of local terms.
For randomly chosen variational parameters we show that the variance of the cost function gradient decreases exponentially with the distance of a Hamiltonian term from the canonical centre in the quantum tensor network.
Therefore, as a function of qubit count, for qMPS most gradient variances decrease exponentially and for qTTN as well as qMERA they decrease polynomially.
We also show that the calculation of these gradients is exponentially more efficient on a classical computer than on a quantum computer.
\end{abstract}

\maketitle

\section{Introduction}

Noisy intermediate-scale quantum (NISQ) devices possess just a small number of imperfect qubits~\cite{JPreskill18} but offer unprecedented computational capabilities.
Whilst not powerful enough to run paradigm-shifting quantum algorithms with guaranteed quantum advantage, such as Shor's algorithm~\cite{Shor} or Grover search~\cite{Grover}, they can already outperform classical computers~\cite{Arute2019, YulinEtAl2021}.

\begin{figure}[h]
\centering
\scalebox{0.8}{\input{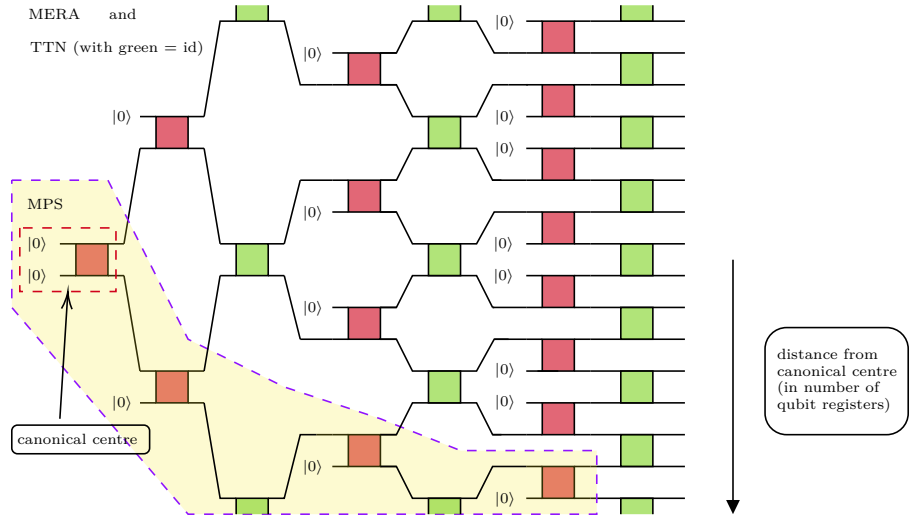}}
\caption{\label{fig:1}
Summary of the main results.
We consider the qMERA with periodic boundary conditions (all gates shown; top light green gates connect to bottom ones), the qTTN (dark red gates) and the qMPS (dark red gates in shaded area).
For most gates in these circuits the gradient variance with respect to randomly chosen parameters decreases exponentially with the distance of the cost function's observable from the canonical centre.
As a function of qubit count this distance can grow linearly for qMPS and it does grow logarithmically for both qTTN and qMERA so that the corresponding gradient variances decrease exponentially and polynomially, respectively.
}
\end{figure}

Variational quantum algorithms are a promising toolbox to work with NISQ devices and achieve a quantum advantage~\cite{CeEtAl21, BhEtAl22, TiEtAl21}.
The variational approach is characterized by an iterative feedback loop between a quantum and a classical computer during which a parameterized quantum circuit (PQC) is optimized to solve the problem of interest.
On the quantum device, the PQC is applied to some initial state to realize the variational wavefunction on which measurements are performed.
The measurement results are subsequently processed on the classical device which, e.g., evaluates the cost function, computes gradients and updates the PQC parameters.

Since the seminal articles proposing the variational quantum eigensolver (VQE)~\cite{APeruzzo14} and quantum approximate optimization algorithm (QAOA)~\cite{EFarhi14}, variational quantum algorithms have been designed and analyzed for numerous applications including combinatorial optimization~\cite{LZhou20, DAmaro22filtering, DAmaro22case}, machine learning~\cite{MBenedetti19, MSchuld20, BeEtAl21}, quantum chemistry~\cite{AKandala17, YCao19, McEtAl20, BBauer20}, finance~\cite{PlEtAl22, HeEtAl22}, partial differential equations~\cite{LuEtAl20, KyPaEl21} and Hamiltonian simulation~\cite{YuEtAl19, BeFiLu21}.

The variational optimization of a PQC, however, is hard~\cite{BittelKliesch2021}.
One of the difficulties that can be encountered during the optimization is related to the barren plateau phenomenon~\cite{McClean18} which manifests itself by a parameter landscape of the cost function that, in simple terms, is flat everywhere except for narrow gorges surrounding local minima.
These flat landscapes pose a problem for the optimization of a PQC as they imply that one needs to run the quantum computer and collect samples many times to accurately determine the gradients of the cost function with respect to the variational parameters.
The large sampling cost can rule out any quantum advantage one is aiming at with variational quantum algorithms.
The severity of the barren plateau problem depends on the cost function~\cite{MCerezo21} and the PQC architecture~\cite{McClean18, PeEtAl21, CZhao21}.
A plethora of proposals exist to avoid barren plateaus in certain cases~\cite{Grant19, MOCarrero20, JKim21, ASkolik21, MCerezo21, SSack22, HLiu21, JDborin21, CZhao21, LiuYuEtAl2021, Zhang22, MeEtAl22}.

In this article we study the trainability of quantum tensor networks using the approach~\cite{CZhao21} (see also~\cite{WaYe22}) which is based on the ZX-calculus~\cite{CoeckeEtAl2008, CoeckeDuncan2011}.
Tensor networks have proven to be a powerful variational ansatz for the simulation of quantum many-body systems on classical computers~\cite{Sc05, VeMuCi08, Sc11, Or14, Orus2019Tensor, CiEtAl21, Ba22}.
Quantum tensor networks have become popular recently since they can be realized on current NISQ devices~\cite{ChEtAl21, NiEtAl21, AnEtAl22, SmEtAl22, FoEtAl22, ZhEtAl22, WaEtAl22, GuEtAl22} and have advantages over their classical counterparts~\cite{HuEtAl19, JLiu19, FoEtAl21, ArSp22}.
We focus on PQC architectures inspired by matrix product states~\cite{FaNaWe92, Schon05, PeEtAl07, Os11} (qMPS), tree tensor networks~\cite{YShi06, Hackbusch09, Tagliacozzo09, HuEtAl19} (qTTN) and the multiscale entanglement renormalization ansatz~\cite{GVidal08, GEvenbly09} (qMERA).
An important concept in these tensor networks is the canonical centre which is the first quantum gate of the circuit.
We show that the barren plateau phenomenon is fundamentally connected to the distance between the observable of interest and the canonical centre.
Figure~\ref{fig:1} summarizes our results.

Our analysis is inspired by~\cite{CZhao21} and extends their results.
For the qMPS ansatz considered in~\cite{CZhao21} we study the barren plateau problem in more detail.
In~\cite{CZhao21} a discriminative qTTN is analyzed and here we explore the corresponding generative variant~\cite{HuEtAl19}, which represents the quantum counterpart to standard classical TTN~\cite{YShi06, Hackbusch09, Tagliacozzo09, HuEtAl19}.
Additionally we investigate a qMERA ansatz not considered in~\cite{CZhao21}.
It is worth noting that~\cite{CZhao21} studies the quantum convolutional neural network (qCNN) ansatz of~\cite{ICong19} which can be viewed as the discriminative variant of the qMERA.
In~\cite{CZhao21} it is shown that the discriminative qTTN and qCNN avoid barren plateaus, but their results are fundamentally different from the ones presented here:
This is because in the discriminative variants the distance between the observable and the canonical distance is always equal to the number of qubits, whereas in the generative variants this is not the case in general.
We also emphasize that the purpose of this work is not to relate to generative quantum machine learning but to address the application of classical tensor network techniques in quantum machine learning.

This article is structured as follows.
In Sec.~\ref{sec_background} we present the necessary background.
Section~\ref{sec_meth} contains the results.
Technical details including the proofs are provided in appendices.

\section{Background}
\label{sec_background}

We collect background information on VQE in Sec.~\ref{back_vqe}, the barren plateau phenomenon in Sec.~\ref{back_bp} and the ZX-calculus in Sec.~\ref{meth_zx_calc}

\subsection{Variational quantum eigensolver}
\label{back_vqe}

Originally introduced in~\cite{APeruzzo14} the variational quantum eigensolver (VQE) consists of a training loop that iterates between a quantum and a classical computer and makes use of the variational principle to solve the minimization problem $\langle H \rangle_{\boldsymbol{\theta^*}} = \underset{\boldsymbol{\theta}}{\min} \langle H \rangle_{\boldsymbol{\theta}}$ where
\begin{align}\label{cost_func}
\langle H \rangle_{\boldsymbol{\theta}} = \bra{\psi(\boldsymbol{\theta})} H \ket{\psi(\boldsymbol{\theta})},
\end{align}
for some Hermitian observable $H$, e.g.\ a Hamiltonian.
During each training iteration the quantum computer prepares the variational wavefunction $\ket{\psi(\boldsymbol{\theta})} = U(\boldsymbol{\theta})\ket{\mathbf{0}}$ via a PQC of the form
\begin{align}\label{pqc}
U(\boldsymbol{\theta}) = \prod_{j = 1}^{M} U_{j}(\theta_{j}),
\end{align}
where $U_{j}(\theta_{j}) = \exp(-i \theta_{j} V_{j} / 2) W_{j}$, $\theta_{j} \in [-\pi, \pi]$, $V_{j}^{2} = I$ and $W_{j}$ is an unparameterized unitary.
The quantum computer is also used to compute cost function gradients via the parameter-shift rule
\begin{align}
\partial_{{\theta_{j}}} \langle H \rangle_{\boldsymbol{\theta}} \equiv \partial_{j} \langle H \rangle_{\boldsymbol{\theta}} = \frac{1}{2}\big(\langle H \rangle_{\boldsymbol{\theta}+\frac{\pi}{2} \mathbf{e_{j}}} - \langle H \rangle_{\boldsymbol{\theta}-\frac{\pi}{2} \mathbf{e_{j}}} \big)
\end{align}
where $\mathbf{e_{j}}$ is the $j$-th unit vector~\cite{MiEtAl18, ScEtAl19}.
The classical computer subsequently updates the parameters $\boldsymbol{\theta}$ and then feeds them back to the quantum machine for the next training iteration. 
The parameters are updated e.g.\ using the gradient descent procedure:
\begin{align}
\boldsymbol{\theta} \rightarrow \boldsymbol{\theta} - \eta \nabla_{\boldsymbol{\theta}} \langle H \rangle_{\boldsymbol{\theta}},
\end{align}
where $\eta$ is the learning rate and $\nabla_{\boldsymbol{\theta}} \langle H \rangle_{\boldsymbol{\theta}}$ denotes the gradient vector.
An alternative gradient-based method that has become popular in the context of variational quantum algorithms is the Adam optimizer~\cite{KingmaBa2014}.
A comprehensive review article on VQE is~\cite{TiEtAl21}.

In this article we focus on $k$-local Hamiltonians, i.e.\ sums of observables which act on at most $k$ qubits.
One example of a $2$-local Hamiltonian is the transverse-field quantum Ising chain:
\begin{align}\label{ham_ising}
H_{\text{Ising}} = -J \sum_{\langle i, j \rangle} Z_{i} Z_{j} - h \sum_{i} X_{i},
\end{align}
where $J$ and $h$ are Hamiltonian parameters, $\langle i, j \rangle$ represents adjacent qubits and $X$ ($Z$) is the Pauli $X$ ($Z$) matrix.
Another example is the Heisenberg model:
\begin{align}\label{ham_heisenber}
H_{\text{Heis}} = \frac{1}{4}\sum_{\langle i, j\rangle} X_{i} X_{j} + Y_{i} Y_{j} + Z_{i} Z_{j}.
\end{align}

\subsection{Barren plateaus}
\label{back_bp}

The barren plateau phenomenon in the variational optimization of quantum circuits was first discussed in~\cite{McClean18} and characterized in the following way:
\begin{theorem}\thlabel{bp_theorem}
Let Eq.~\eqref{cost_func} be a cost function with an associated parameterized ansatz Eq.~\eqref{pqc} acting on $N$ qubits.
For some $1 \leq k \leq M$ define
\begin{align}
 U &= U_{\text{L}} U_{k} U_{\text{R}}
\end{align}
for $U_{\text{L}} = \prod_{j < k} U_{j}(\theta_{j})$ and $U_{\text{R}} = \prod_{j > k} U_{j}(\theta_{j})$.
Then
\begin{itemize}
 \item $\E[\partial_{k}\langle H \rangle_{\boldsymbol{\theta}}] = 0$ if $U_{\text{L}}$ and $U_{\text{R}}$ form random unitary 1-designs,
 \item $\Var[\partial_{k}\langle H \rangle_{\boldsymbol{\theta}}] \in O(c^{-N})$ for $c > 1$ if either $U_{\text{L}}$, $U_{\text{R}}$ or both form random unitary 2-designs
\end{itemize}
where $\E[\cdot]$ denotes the average value and $\Var[\cdot]$ the variance over randomly chosen parameters.
\end{theorem}
In simple terms \thref{bp_theorem} tells us that the unitary $2$-design condition establishes a cost landscape which is nearly flat everywhere (barren plateaus) except for exponentially small regions around minima (narrow gorges).
Using Chebyshev's inequality we see that for randomly chosen parameters the probability of obtaining a gradient of magnitude $\lvert \partial_{k}\langle H \rangle_{\boldsymbol{\theta}}\rvert > \kappa$ vanishes exponentially with qubit count\footnote{
We use the following notation: $f(N) \in O(g(N))$ if $f(N)$ is asymptotically bounded above by $c \cdot g(N)$ for some $c > 0$, $f(N) \in \Omega(g(N))$ if $ f(N)$ is asymptotically bounded below by $c \cdot g(N)$ for some $c > 0$, and $f(N) \in \Theta(g(N))$ if $f(N)$ is asymptotically bounded below by $c_{1} \cdot g(N)$ and above by $c_{2} \cdot g(N)$ for some $c_{1}, c_{2} > 0$.
}:
\begin{align}\label{cheby}
\Pr[\big\lvert \partial_{k}\langle H \rangle_{\boldsymbol{\theta}} - \E[\partial_{k}\langle H \rangle_{\boldsymbol{\theta}}] \big\rvert \geq \kappa] \leq \frac{\Var[\partial_{k}\langle H \rangle_{\boldsymbol{\theta}}]}{\kappa^{2}} \in O\Big(\frac{c^{-N}}{\kappa^{2}}\Big).
\end{align}
The barren plateau phenomenon is a problem for the trainability of PQCs since the computation of exponentially small gradients using standard techniques, such as the parameter-shift rule, requires exponentially many measurements on a quantum computer.
Because the computational cost of performing these calculations on a classical computer also scales exponentially with qubit count, a classical approach might be more efficient than a quantum one in which case there is no quantum advantage.

While in~\cite{McClean18} it is shown that the onset of the unitary $2$-design property is caused by large circuit depth, in~\cite{MCerezo21} the authors show that also the form of the cost function affects the depth at which barren plateaus emerge.
More specifically they show that PQC optimization with local cost functions is efficient for depths that scale logarithmically with qubit count and transitions into the barren plateau regime when depths scale as $O(\text{poly}(\log(N)))$.
PQC training based on global cost functions, however, is shown to always be subject to barren plateaus, even for shallow $O(1)$ depth circuits.

Focusing on local observables the analysis in~\cite{MCerezo21} suggests that the onset of barren plateaus is related to the entanglement in the causal cone of the observable\footnote{
The causal cone of an observable $X_i$ acting on qubit register $i$ is the sub-circuit composed of only the qubits and gates in the PQC which affect the measurement outcome at site $i$.
If a variational parameter is in the causal cone of an observable then we refer to them as causally connected.
}.
This is analysed in detail in~\cite{SSack22} where the authors show that sufficiently large amounts of entanglement in the quantum circuit are necessary for the emergence of unitary $2$-designs and claim that entanglement-induced barren plateaus \cite{MOCarrero20, JKim21, SWang21} and barren plateaus for local cost functions are equivalent.

Due to its importance for the field of variational quantum algorithms, the barren plateau problem has been studied in many articles.
Some articles have identified PQC architectures that avoid barren plateaus~\cite{APesah21, CZhao21} and others propose ways to mitigate the barren plateau problem, e.g.\ in~\cite{Grant19} the authors propose to initialize the circuit with shallow identity gates formed by unitaries and their adjoints, in~\cite{ASkolik21} they advertise a layer-wise learning strategy, in~\cite{HLiu21, MeEtAl22} they propose to initialize the PQC using previously trained PQCs, in~\cite{JDborin21} they propose to use a previously trained qMPS for the PQC initialization, and in~\cite{Zhang22} the authors claim that the barren plateau problem is solved by choosing the initial parameters from a particular Gaussian distribution.

\subsection{ZX-calculus for barren plateau analysis}
\label{meth_zx_calc}

In~\cite{CZhao21} Chen Zhao and Xiao-Shan Gao pioneer the use of the ZX-calculus~\cite{CoeckeEtAl2008, CoeckeDuncan2011} to analyse the barren plateau phenomenon.
They use the following assumption:
\begin{assumption}
\thlabel{zx_assumption}
The parameterized quantum ansatz in Eq.~\eqref{pqc} is such that
\begin{enumerate}
 \item each gate $U_{j}$ in $U$ is from $\{R_{X} = \exp(-\text{i} \theta_{j} X/2), R_{Z} = \exp(-\text{i} \theta_{j} Z/2), H, CNOT\}$ where $H$ is the Hadamard gate and $CNOT$ the controlled-$X$ gate,
 \item each parameter $\theta_{j}$ is uniformly sampled from $[-\pi, \pi]$.
\end{enumerate}
\end{assumption}
\noindent They show:
\begin{theorem}\thlabel{bp_theorem_zx}
Let Eq.~\eqref{cost_func} be a cost function with associated parameterized ansatz~\eqref{pqc} for $N$ qubits and under~\thref{zx_assumption}:
\begin{itemize}
 \item $\E[\partial_{j}\langle H \rangle_{\boldsymbol{\theta}}] = 0$, 
 \item $\Var[\partial_{j}\langle H \rangle_{\boldsymbol{\theta}}] = \frac{\abs{c}^{2}}{4^{N}} \sum_{a_{k} \in \{T_{1}, T_{2}, T_{3}\}, k \neq j} V_{U}^{a_{1},\ldots , a_{j-1}, T_{2}, a_{j+1},\ldots, a_{M}}$, where $c$ is a constant, $V_{U}^{a_{1},\ldots,a_{M}}$ is a ZX-diagram and $a_{1},\ldots,a_{M}$, $T_{1}$, $T_{2}$ and $T_{3}$ are labels defining the ZX-diagram~\cite{CZhao21}.
\end{itemize}
\end{theorem}
While \thref{bp_theorem_zx} does not immediately tell us whether a specific choice of PQC and cost function leads to barren plateaus, it provides us with a constructive procedure to compute the variance of gradients by evaluating ZX-diagrams.
This calculation can be further simplified by turning the ZX-diagram into tensor networks whose contraction directly produces the sought-after variance value.
In App.~\ref{app_zx_meth} we explain the ZX-calculus formalism that is relevant for this article and also give a simple example that illustrates step-by-step how one can use this formalism to obtain the tensor network for the gradient variance starting from a PQC and using ZX-diagrams.

\section{Results}
\label{sec_meth}

We present the results on qMPS in Sec.~\ref{meth_mps}, qTTN in Sec.~\ref{meth_ttn} and qMERA in Sec.~\ref{meth_mera}.
In Sec.~\ref{meth_quvscl} we compare the quantum and classical computational cost of calculating gradients.

\subsection{Quantum matrix product states}
\label{meth_mps}

We consider the qMPS ansatz
\begin{align}\label{mps_ans}
 U^{\text{qMPS}} := \prod_{j = N-1}^{1} U_{j}^{\text{qMPS}}
\end{align}
composed of two-qubit blocks of the form
\begin{align}
     \scalebox{1}{\tikzset{every picture/.style={line width=0.75pt}} %set default line width to 0.75pt        

\begin{tikzpicture}[x=0.75pt,y=0.75pt,yscale=-1,xscale=1]
%uncomment if require: \path (0,454); %set diagram left start at 0, and has height of 454

%Straight Lines [id:da20548037714034906] 
\draw    (105,120) -- (120,120) ;
%Straight Lines [id:da08278280356851098] 
\draw    (105,140) -- (120,140) ;
%Shape: Rectangle [id:dp522055916783714] 
\draw   (120,110) -- (160,110) -- (160,150) -- (120,150) -- cycle ;
%Straight Lines [id:da04752079154452771] 
\draw    (160,120) -- (175,120) ;
%Straight Lines [id:da9911561725308018] 
\draw    (160,140) -- (175,140) ;
%Straight Lines [id:da290588803606622] 
\draw    (205,120) -- (220,120) ;
%Straight Lines [id:da9874268670394466] 
\draw    (205,140) -- (220,140) ;
%Straight Lines [id:da5374832142312962] 
\draw    (239,120) -- (250,120) ;
%Straight Lines [id:da7198912571387237] 
\draw    (239,140) -- (250,140) ;
%Straight Lines [id:da0595707997713133] 
\draw    (269,120) -- (301,120) ;
%Straight Lines [id:da4093058326276362] 
\draw    (319,120) -- (330,120) ;
%Straight Lines [id:da3684633660618002] 
\draw    (349,120) -- (365,120) ;
%Straight Lines [id:da6758310769859821] 
\draw    (269,140) -- (365,140) ;
%Straight Lines [id:da29369936519101203] 
\draw    (285,120) -- (285,137.65) ;
\draw [shift={(285,140)}, rotate = 90] [color={rgb, 255:red, 0; green, 0; blue, 0 }  ][line width=0.75]      (0, 0) circle [x radius= 3.35, y radius= 3.35]   ;
%Straight Lines [id:da5085798309262421] 
\draw    (285,137) -- (285,143) ;

% Text Node
\draw (122,119.4) node [anchor=north west][inner sep=0.75pt]  [font=\footnotesize]  {$U_{j}^{\text{qMPS}}$};
% Text Node
\draw (179,120.4) node [anchor=north west][inner sep=0.75pt]    {$=$};
% Text Node
\draw    (220.85,113.09) -- (239.85,113.09) -- (239.85,128.09) -- (220.85,128.09) -- cycle  ;
\draw (230.35,120.59) node  [font=\tiny]  {$R_{X}$};
% Text Node
\draw    (220.85,132.09) -- (239.85,132.09) -- (239.85,147.09) -- (220.85,147.09) -- cycle  ;
\draw (230.35,139.59) node  [font=\tiny]  {$R_{X}$};
% Text Node
\draw    (251.22,113.09) -- (269.22,113.09) -- (269.22,128.09) -- (251.22,128.09) -- cycle  ;
\draw (260.22,120.59) node  [font=\tiny]  {$R_{Z}$};
% Text Node
\draw    (251.22,132.09) -- (269.22,132.09) -- (269.22,147.09) -- (251.22,147.09) -- cycle  ;
\draw (260.22,139.59) node  [font=\tiny]  {$R_{Z}$};
% Text Node
\draw    (300.85,113.09) -- (319.85,113.09) -- (319.85,128.09) -- (300.85,128.09) -- cycle  ;
\draw (310.35,120.59) node  [font=\tiny]  {$R_{X}$};
% Text Node
\draw    (331.22,113.09) -- (349.22,113.09) -- (349.22,128.09) -- (331.22,128.09) -- cycle  ;
\draw (340.22,120.59) node  [font=\tiny]  {$R_{Z}$};

\end{tikzpicture}}
\end{align}
acting on qubits $j$ and $j+1$ for $j < N-1$ and 
\begin{align}
     \scalebox{1}{\tikzset{every picture/.style={line width=0.75pt}} %set default line width to 0.75pt        

\begin{tikzpicture}[x=0.75pt,y=0.75pt,yscale=-1,xscale=1]
%uncomment if require: \path (0,454); %set diagram left start at 0, and has height of 454

%Straight Lines [id:da05278039451571481] 
\draw    (105,120) -- (120,120) ;
%Straight Lines [id:da29634508938933846] 
\draw    (105,140) -- (120,140) ;
%Shape: Rectangle [id:dp787342496131072] 
\draw   (120,110) -- (160,110) -- (160,150) -- (120,150) -- cycle ;
%Straight Lines [id:da2841239920491745] 
\draw    (160,120) -- (175,120) ;
%Straight Lines [id:da956434900835496] 
\draw    (160,140) -- (175,140) ;
%Straight Lines [id:da9655419938666219] 
\draw    (205,120) -- (220,120) ;
%Straight Lines [id:da19620349342458954] 
\draw    (205,140) -- (220,140) ;
%Straight Lines [id:da9110788312604126] 
\draw    (239,120) -- (250,120) ;
%Straight Lines [id:da8169316100612716] 
\draw    (239,140) -- (250,140) ;
%Straight Lines [id:da83666562204908] 
\draw    (269,120) -- (301,120) ;
%Straight Lines [id:da1884203214168143] 
\draw    (319,120) -- (330,120) ;
%Straight Lines [id:da6122519349139441] 
\draw    (349,120) -- (365,120) ;
%Straight Lines [id:da15664796906191447] 
\draw    (285,120) -- (285,137.65) ;
\draw [shift={(285,140)}, rotate = 90] [color={rgb, 255:red, 0; green, 0; blue, 0 }  ][line width=0.75]      (0, 0) circle [x radius= 3.35, y radius= 3.35]   ;
%Straight Lines [id:da6519079591710746] 
\draw    (285,137) -- (285,143) ;
%Straight Lines [id:da38130346752159494] 
\draw    (269,140) -- (301,140) ;
%Straight Lines [id:da7228695106627636] 
\draw    (319,140) -- (330,140) ;
%Straight Lines [id:da26073689487958585] 
\draw    (349,140) -- (365,140) ;

% Text Node
\draw (122,119.4) node [anchor=north west][inner sep=0.75pt]  [font=\footnotesize]  {$U_{j}^{\text{qMPS}}$};
% Text Node
\draw (179,120.4) node [anchor=north west][inner sep=0.75pt]    {$=$};
% Text Node
\draw    (220.85,113.09) -- (239.85,113.09) -- (239.85,128.09) -- (220.85,128.09) -- cycle  ;
\draw (230.35,120.59) node  [font=\tiny]  {$R_{X}$};
% Text Node
\draw    (220.85,132.09) -- (239.85,132.09) -- (239.85,147.09) -- (220.85,147.09) -- cycle  ;
\draw (230.35,139.59) node  [font=\tiny]  {$R_{X}$};
% Text Node
\draw    (251.22,113.09) -- (269.22,113.09) -- (269.22,128.09) -- (251.22,128.09) -- cycle  ;
\draw (260.22,120.59) node  [font=\tiny]  {$R_{Z}$};
% Text Node
\draw    (251.22,132.09) -- (269.22,132.09) -- (269.22,147.09) -- (251.22,147.09) -- cycle  ;
\draw (260.22,139.59) node  [font=\tiny]  {$R_{Z}$};
% Text Node
\draw    (300.85,113.09) -- (319.85,113.09) -- (319.85,128.09) -- (300.85,128.09) -- cycle  ;
\draw (310.35,120.59) node  [font=\tiny]  {$R_{X}$};
% Text Node
\draw    (331.22,113.09) -- (349.22,113.09) -- (349.22,128.09) -- (331.22,128.09) -- cycle  ;
\draw (340.22,120.59) node  [font=\tiny]  {$R_{Z}$};
% Text Node
\draw    (300.85,133.09) -- (319.85,133.09) -- (319.85,148.09) -- (300.85,148.09) -- cycle  ;
\draw (310.35,140.59) node  [font=\tiny]  {$R_{X}$};
% Text Node
\draw    (331.22,133.09) -- (349.22,133.09) -- (349.22,148.09) -- (331.22,148.09) -- cycle  ;
\draw (340.22,140.59) node  [font=\tiny]  {$R_{Z}$};

\end{tikzpicture}}
\end{align}
acting on qubits $N-1$ and $N$, cf.\ App.~\ref{app_mps} for a full circuit diagram.
Here $U^{\text{qMPS}}_{1}$ is the canonical centre of the qMPS.

\begin{theorem}\thlabel{th_mps}
Let $\langle X_{i} \rangle_{\text{qMPS}}$ be the cost function associated with the observable $X_{i}$ and consider the qMPS ansatz for $N$ qubits defined in Eq.~\eqref{mps_ans}, then:
\begin{align}
 \Var[\partial_{j,1}\langle X_{N} \rangle_{\text{qMPS}}] & = \begin{cases} \frac{1}{4} \cdot \big(\frac{3}{8}\big)^{N-1} & \text{if } j < N,\\
  \frac{1}{4} \Big(1+\big(\frac{3}{8}\big)^{N-1}\Big) & \text{if } j = N, \end{cases}\\
 \label{th_mps_eq}\Var[\partial_{j,1}\CostXi{}_{\text{qMPS}}] & = \begin{cases} 11 \cdot \big(\frac{1}{8}\big)^{2} \big(\frac{3}{8}\big)^{i-1} & \text{if } j < i \text{ or } j = i = 1,\\
  3 \cdot \big(\frac{1}{8}\big)^{2} \Big(1+\frac{11}{8} \cdot \big(\frac{3}{8}\big)^{i-2}\Big) & \text{if } j = i,\\
  3 \cdot \big(\frac{1}{8}\big)^{2} \Big(1+\big(\frac{3}{8}\big)^{i-1}\Big) & \text{if } j = i+1, \end{cases}
\end{align}
where $\partial_{j,1}\CostXi{}_{\text{qMPS}}$ refers to the gradient w.r.t.\ the $1$-st parameter in the $j$-th qubit register.
\end{theorem}
\begin{proof}
See App.~\ref{app_mps}, \thref{th_mps2} and \thref{th_mps_varjk}.
\end{proof}

\thref{th_mps} tells us that the gradient variance with respect to parameter $(j,1)$ for $j < i$ is independent of $j$ and depends only on $i$, i.e.\ the distance between the observable at site $i$ and the canonical centre.
We also learn from \thref{th_mps} that for $j = i, i+1$ the gradient variance has a constant contribution.
Note that for $j > i+1$ we have $\Var[\partial_{j, k}\CostXi{}_{\text{qMPS}}] = 0$ since the variational parameter indexed by $(j,k)$ is outside the causal cone of the observable $X_{i}$, see e.g.\ Fig.~\ref{fig:mps_causal_cone} in App.~\ref{app_mps}. 
We show in App.~\ref{app_mps} that $\Var[\partial_{j,k}\CostXi{}_{\text{qMPS}}] \geq \Var[\partial_{1,1}\CostXi{}_{\text{qMPS}}]$ for all $j, k$ for which $\Var[\partial_{j,k}\CostXi{}_{\text{qMPS}}] \neq 0$.
In other words the variance w.r.t.\ the top-left parameter is a lower bound to all other non-zero variances in the qMPS ansatz.

Note that \thref{th_mps} implies that the qMPS ansatz avoids the barren plateau problem for a Hamiltonian that is a sum of local terms acting on all qubits, e.g.\ the Hamiltonian $H = \sum_{i=1}^{N} X_{i}$, Ising and Heisenberg models.
Focusing on the Hamiltonian $H = \sum_{i=1}^{N} X_{i}$, this is because \thref{th_mps} shows that each term $X_{i}$ in $H$ leads to non-vanishing gradient variances for parameters in registers $i$ and $i+1$.
Hence, every parameter in the qMPS will have a contribution to the gradient variance which is non-vanishing.
However, this is not the case for arbitrary Hamiltonians.
If we consider a Hamiltonian acting on a single site, for example $H = X_{N}$, then \thref{th_mps} shows that the gradient variances for all parameters in registers $i < N$ vanish exponentially.

Additionally we show:
\begin{theorem}\thlabel{th_mpsk2}
Let $\langle X_{i} X_{i+1} \rangle_{\text{qMPS}}$ be the cost function associated with the observable $X_{i} X_{i+1}$ and consider the qMPS ansatz of Eq.~\eqref{mps_ans}, then:
\begin{align}
 \Var[\partial_{1,1}\langle X_{i} X_{i-1} \rangle_{\text{qMPS}}] = c_{i} \Big(\frac{3}{8}\Big)^{i},
\end{align}
where
\begin{align}
 c_{i} = \begin{cases} \frac{1}{4} \cdot \Big(\big(\frac{3}{8}\big)^{2}+\frac{13}{16}\Big) & \text{if } i = 1,\\
 \frac{1}{4} \cdot \Big(\frac{37}{2\cdot8^{2}}+\frac{3}{16})\Big) & \text{if } 1 < i < N,\\
 \frac{37}{3 \cdot 8^{2}} & \text{if } i = N-1, \end{cases}
\end{align}
where $\partial_{1,1}\CostXi{}_{\text{qMPS}}$ refers to the gradient w.r.t.\ the $1$-st parameter in the $1$-st qubit register.
\end{theorem}
\begin{proof}
See App.~\ref{app_mps}, \thref{th_mps_xx}.
\end{proof}

We generalize the results to $k$-local observables and propose:
\begin{prop}\thlabel{prop_mps}
If $k \ll N$ then the $k$-local operators $X_{I}$ acting on qubits $I = \{i_{1}, \ldots, i_{k}\}$ with $i_{1} < \ldots < i_{k}$ satisfy $\Var[\partial_{1,1}\langle X_{I}\rangle_{\text{qMPS}}] \in \Omega(c^{-i_{k}})$ for $c > 1$.
\end{prop}
The cases $k = 1$ and $k = 2$ are already shown in \thref{th_mps} and \thref{th_mpsk2} and we discuss $k > 2$ in App.~\ref{app_mps}.

\subsection{Quantum tree tensor networks}
\label{meth_ttn}

We consider a qTTN ansatz for $N = 2^{n}$ qubits of the following form for $n = 1$:
\begin{align}
 \scalebox{1}{\tikzset{every picture/.style={line width=0.75pt}} %set default line width to 0.75pt        

\begin{tikzpicture}[x=0.75pt,y=0.75pt,yscale=-1,xscale=1]
%uncomment if require: \path (0,454); %set diagram left start at 0, and has height of 454

%Straight Lines [id:da8363021890332014] 
\draw    (105,120) -- (120,120) ;
%Straight Lines [id:da08054273635318987] 
\draw    (105,140) -- (120,140) ;
%Shape: Rectangle [id:dp9190988735133878] 
\draw   (120,110) -- (160,110) -- (160,150) -- (120,150) -- cycle ;
%Straight Lines [id:da10374830912612265] 
\draw    (160,120) -- (175,120) ;
%Straight Lines [id:da5104203761725301] 
\draw    (160,140) -- (175,140) ;
%Straight Lines [id:da7529848035883877] 
\draw    (225,120) -- (240,120) ;
%Straight Lines [id:da29499114022086736] 
\draw    (225,140) -- (240,140) ;
%Straight Lines [id:da7499751098094796] 
\draw    (259,120) -- (270,120) ;
%Straight Lines [id:da23416409619770162] 
\draw    (259,140) -- (270,140) ;
%Straight Lines [id:da20625337248421682] 
\draw    (289,120) -- (321,120) ;
%Straight Lines [id:da8662931709446637] 
\draw    (339,120) -- (350,120) ;
%Straight Lines [id:da1342980616778524] 
\draw    (369,120) -- (385,120) ;
%Straight Lines [id:da08491634302723283] 
\draw    (305,120) -- (305,137.65) ;
\draw [shift={(305,140)}, rotate = 90] [color={rgb, 255:red, 0; green, 0; blue, 0 }  ][line width=0.75]      (0, 0) circle [x radius= 3.35, y radius= 3.35]   ;
%Straight Lines [id:da018078231945010126] 
\draw    (305,137) -- (305,143) ;
%Straight Lines [id:da6920325309999276] 
\draw    (289,140) -- (321,140) ;
%Straight Lines [id:da27614007776612715] 
\draw    (339,140) -- (350,140) ;
%Straight Lines [id:da2310017725907716] 
\draw    (369,140) -- (385,140) ;

% Text Node
\draw (141.58,128.94) node  [font=\footnotesize]  {$U_{2^{1}}^{\text{qTTN}}$};
% Text Node
\draw (179,120.4) node [anchor=north west][inner sep=0.75pt]    {$=$};
% Text Node
\draw    (240.85,113.09) -- (259.85,113.09) -- (259.85,128.09) -- (240.85,128.09) -- cycle  ;
\draw (250.35,120.59) node  [font=\tiny]  {$R_{X}$};
% Text Node
\draw    (240.85,132.09) -- (259.85,132.09) -- (259.85,147.09) -- (240.85,147.09) -- cycle  ;
\draw (250.35,139.59) node  [font=\tiny]  {$R_{X}$};
% Text Node
\draw    (271.22,113.09) -- (289.22,113.09) -- (289.22,128.09) -- (271.22,128.09) -- cycle  ;
\draw (280.22,120.59) node  [font=\tiny]  {$R_{Z}$};
% Text Node
\draw    (271.22,132.09) -- (289.22,132.09) -- (289.22,147.09) -- (271.22,147.09) -- cycle  ;
\draw (280.22,139.59) node  [font=\tiny]  {$R_{Z}$};
% Text Node
\draw    (320.85,113.09) -- (339.85,113.09) -- (339.85,128.09) -- (320.85,128.09) -- cycle  ;
\draw (330.35,120.59) node  [font=\tiny]  {$R_{X}$};
% Text Node
\draw    (351.22,113.09) -- (369.22,113.09) -- (369.22,128.09) -- (351.22,128.09) -- cycle  ;
\draw (360.22,120.59) node  [font=\tiny]  {$R_{Z}$};
% Text Node
\draw    (320.85,133.09) -- (339.85,133.09) -- (339.85,148.09) -- (320.85,148.09) -- cycle  ;
\draw (330.35,140.59) node  [font=\tiny]  {$R_{X}$};
% Text Node
\draw    (351.22,133.09) -- (369.22,133.09) -- (369.22,148.09) -- (351.22,148.09) -- cycle  ;
\draw (360.22,140.59) node  [font=\tiny]  {$R_{Z}$};
% Text Node
\draw (92.89,119.2) node  [font=\scriptsize]  {$\ket{0}$};
% Text Node
\draw (92.5,139) node  [font=\scriptsize]  {$\ket{0}$};
% Text Node
\draw (213,119) node  [font=\scriptsize]  {$\ket{0}$};
% Text Node
\draw (212.61,138.8) node  [font=\scriptsize]  {$\ket{0}$};

\end{tikzpicture}}
\end{align}
and for $n > 1$:
\begin{align}\label{ttn_ans}
 \scalebox{1}{\tikzset{every picture/.style={line width=0.75pt}} %set default line width to 0.75pt        

\begin{tikzpicture}[x=0.75pt,y=0.75pt,yscale=-1,xscale=1]
%uncomment if require: \path (0,454); %set diagram left start at 0, and has height of 454

%Straight Lines [id:da397396765178782] 
\draw    (105,120) -- (120,120) ;
%Straight Lines [id:da6731822152347378] 
\draw    (105,140) -- (120,140) ;
%Shape: Rectangle [id:dp5897407425856811] 
\draw   (120,110) -- (160,110) -- (160,150) -- (120,150) -- cycle ;

%Straight Lines [id:da3723451592924989] 
\draw    (160,120) -- (175,120) ;
%Straight Lines [id:da36848526028692397] 
\draw    (160,140) -- (175,140) ;
%Straight Lines [id:da37868003638530245] 
\draw    (250,100) -- (266,100) ;
%Straight Lines [id:da446165835753759] 
\draw    (250,140.91) -- (266,140.91) ;
%Straight Lines [id:da8550041653930172] 
\draw    (284,100) -- (295,100) ;
%Straight Lines [id:da9647068223588626] 
\draw    (284,140.91) -- (295,140.91) ;
%Straight Lines [id:da23309017333179716] 
\draw    (314,100) -- (346,100) ;
%Straight Lines [id:da1012677150679342] 
\draw    (386,100) -- (402,100) ;
%Straight Lines [id:da766233066856818] 
\draw    (330,100) -- (330,138.56) ;
\draw [shift={(330,140.91)}, rotate = 90] [color={rgb, 255:red, 0; green, 0; blue, 0 }  ][line width=0.75]      (0, 0) circle [x radius= 3.35, y radius= 3.35]   ;
%Straight Lines [id:da9411075325469187] 
\draw    (330,137.91) -- (330,143.91) ;
%Straight Lines [id:da8506351621046409] 
\draw    (314,140.91) -- (346,140.91) ;
%Straight Lines [id:da5828813132241617] 
\draw    (386,140.91) -- (402,140.91) ;
%Straight Lines [id:da7931319547315963] 
\draw    (250,120) -- (346,120) ;
%Straight Lines [id:da5003088774802491] 
\draw    (386,120) -- (402,120) ;
%Straight Lines [id:da7034239093357366] 
\draw    (386,160.91) -- (402,160.91) ;
%Straight Lines [id:da5452565392533832] 
\draw    (250,160) -- (346,160) ;
%Shape: Rectangle [id:dp6296630646679537] 
\draw   (345.92,95) -- (385.92,95) -- (385.92,125) -- (345.92,125) -- cycle ;

%Shape: Rectangle [id:dp33850342701911584] 
\draw   (346,135) -- (386,135) -- (386,165) -- (346,165) -- cycle ;

% Text Node
\draw (179,120.4) node [anchor=north west][inner sep=0.75pt]    {$=$};
% Text Node
\draw    (266.35,93.09) -- (284.35,93.09) -- (284.35,108.09) -- (266.35,108.09) -- cycle  ;
\draw (275.35,100.59) node  [font=\tiny]  {$R_{X}$};
% Text Node
\draw    (266.35,133) -- (284.35,133) -- (284.35,148) -- (266.35,148) -- cycle  ;
\draw (275.35,140.5) node  [font=\tiny]  {$R_{X}$};
% Text Node
\draw    (296.22,93.09) -- (314.22,93.09) -- (314.22,108.09) -- (296.22,108.09) -- cycle  ;
\draw (305.22,100.59) node  [font=\tiny]  {$R_{Z}$};
% Text Node
\draw    (296.22,133) -- (314.22,133) -- (314.22,148) -- (296.22,148) -- cycle  ;
\draw (305.22,140.5) node  [font=\tiny]  {$R_{Z}$};
% Text Node
\draw (82.89,119.2) node  [font=\scriptsize]  {$\ket{0}^{\otimes 2^{n-1}}$};
% Text Node
\draw (83,139) node  [font=\scriptsize]  {$\ket{0}^{\otimes 2^{n-1}}$};
% Text Node
\draw (141.58,128.94) node  [font=\footnotesize]  {$U_{2^{n}}^{\text{qTTN}}$};
% Text Node
\draw (367.5,109.2) node  [font=\footnotesize]  {$U_{2^{n-1}}^{\text{qTTN}}$};
% Text Node
\draw (367.58,149.2) node  [font=\footnotesize]  {$U_{2^{n-1}}^{\text{qTTN}}$};
% Text Node
\draw (226,158) node  [font=\scriptsize]  {$\ket{0}^{\otimes 2^{n-1} -1}$};
% Text Node
\draw (226,118) node  [font=\scriptsize]  {$\ket{0}^{\otimes 2^{n-1} -1}$};
% Text Node
\draw (235,99) node  [font=\scriptsize]  {$\ket{0}$};
% Text Node
\draw (235,139) node  [font=\scriptsize]  {$\ket{0}$};

\end{tikzpicture}}
\end{align}
Appendix~\ref{app_ttn} contains an example of a full circuit diagram.
The top recursion level in Eq.~\eqref{ttn_ans} is the canonical centre of the network.

Each qubit in the qTTN ansatz is causally connected to $n = \log N$ qubits, which allows us to show:
\begin{theorem}\thlabel{th_ttn}
Let $\langle X_{i}\rangle_{\text{qTTN}}$ be the cost function associated with the observable $X_{i}$ and consider the qTTN ansatz defined in Eq.~\eqref{ttn_ans}, then:
\begin{enumerate}
\item $\Var[\partial_{1,1}\langle X_{1}\rangle_{\text{qTTN}}] \geq \Var[\partial_{1,1}\CostXi{}_{\text{qTTN}}] \geq \Var[\partial_{1,1}\langle X_{N}\rangle_{\text{qTTN}}]$ for all $i = 1, \ldots, N$,
\item $\Var[\partial_{1,1}\langle X_{N}\rangle_{\text{qTTN}}] = \frac{1}{4} \cdot \big(\frac{3}{8}\big)^{n}$,
\item $\Var[\partial_{1,1}\langle X_{1}\rangle_{\text{qTTN}}] \in \Omega\Big(\big(\frac{\lambda_{2}}{4}\big)^{n}\Big)$ where $\lambda_{2} \approx 2.3187$.
\end{enumerate}
\end{theorem}
\begin{proof}
See App.~\ref{app_ttn}, \thref{th_ttn_app} and \thref{lemma_ttn_app}.
\end{proof}

In summary \thref{th_ttn} tells us that $\Var[\partial_{1,1}\CostXi{}] \in \Theta(c^{-\log N})$ for all $i$ and for some $c > 1$.
We show in App.~\ref{app_ttn} that $\Var[\partial_{j,k}\langle X_{N} \rangle_{\text{qTTN}}] \geq \Var[\partial_{1,1}\langle X_{N} \rangle_{\text{qTTN}}]$ for all pairs of indices $(j,k)$ provided the former variance is not $0$.
The variance is $0$ in the qTTN ansatz when the variational paramater indexed by $(j,k)$ is outside the causal cone of the observable.
In contrast to the qMPS ansatz, for qTTN the variance decreases polynomially and independently of the site $i$ being considered since the distance between the qubit that the observable acts on and the canonical centre is always $\log N$.
We conclude that the qTTN ansatz avoids the barren plateau problem.

We extend the results to $k$-local observables for $k \ll N$.
In this case the observable is causally connected to $O(k \log N)$ qubits.
We propose:
\begin{prop}\thlabel{prop_ttn}
If $k \ll N$ then the $k$-local operators $X_{I}$ acting on qubits $I=\{i_{1}, \ldots, i_{k}\}$ satisfy $\Var[\partial_{1,1}\langle X_{I}\rangle_{\text{qTTN}}] \in \Omega(c^{-k \log N})$.
\end{prop}
The case $k = 1$ is covered by \thref{th_ttn} and we discuss the general case in App.~\ref{app_ttn}.

\subsection{Quantum multiscale entanglement renormalization ansatz}
\label{meth_mera}

We define the qMERA ansatz for $N = 2^{n}$ qubits as a product of $n$ layers each of which is composed of a disentangling (Dis) and a coarse-graining (CG) layer:
\begin{align}\label{mera_ans1}
 \scalebox{0.8}{\input{tikz_files/mera_circ2}}
\end{align}
where the two qubit gates are given by
\begin{align}
\label{mera_BM_circ}
 \scalebox{1}{\tikzset{every picture/.style={line width=0.75pt}} %set default line width to 0.75pt        

\begin{tikzpicture}[x=0.75pt,y=0.75pt,yscale=-1,xscale=1]
%uncomment if require: \path (0,454); %set diagram left start at 0, and has height of 454

%Straight Lines [id:da10715283218172877] 
\draw    (135,120) -- (145,120) ;
%Straight Lines [id:da6708361884742309] 
\draw    (135,140) -- (145,140) ;
%Shape: Rectangle [id:dp4613409968144164] 
\draw   (145,110) -- (165,110) -- (165,150) -- (145,150) -- cycle ;
%Straight Lines [id:da5745002397807644] 
\draw    (165,120) -- (175,120) ;
%Straight Lines [id:da025902377285038636] 
\draw    (165,140) -- (175,140) ;
%Straight Lines [id:da5165232501630017] 
\draw    (205,120) -- (220,120) ;
%Straight Lines [id:da8411174219953503] 
\draw    (205,140) -- (220,140) ;
%Straight Lines [id:da6804563986878303] 
\draw    (239,120) -- (250,120) ;
%Straight Lines [id:da5502946327928564] 
\draw    (239,140) -- (250,140) ;
%Straight Lines [id:da18582034728664532] 
\draw    (269,120) -- (301,120) ;
%Straight Lines [id:da6568772924347643] 
\draw    (285,120) -- (285,137.65) ;
\draw [shift={(285,140)}, rotate = 90] [color={rgb, 255:red, 0; green, 0; blue, 0 }  ][line width=0.75]      (0, 0) circle [x radius= 3.35, y radius= 3.35]   ;
%Straight Lines [id:da13084096858479843] 
\draw    (285,137) -- (285,143) ;
%Straight Lines [id:da25395390779257965] 
\draw    (269,140) -- (301,140) ;

% Text Node
\draw (149,124.4) node [anchor=north west][inner sep=0.75pt]  [font=\footnotesize]  {$U$};
% Text Node
\draw (179,120.4) node [anchor=north west][inner sep=0.75pt]    {$=$};
% Text Node
\draw    (221.35,113.09) -- (239.35,113.09) -- (239.35,128.09) -- (221.35,128.09) -- cycle  ;
\draw (230.35,120.59) node  [font=\tiny]  {$R_{X}$};
% Text Node
\draw    (220.85,132.09) -- (239.85,132.09) -- (239.85,147.09) -- (220.85,147.09) -- cycle  ;
\draw (230.35,139.59) node  [font=\tiny]  {$R_{X}$};
% Text Node
\draw    (251.22,113.09) -- (269.22,113.09) -- (269.22,128.09) -- (251.22,128.09) -- cycle  ;
\draw (260.22,120.59) node  [font=\tiny]  {$R_{Z}$};
% Text Node
\draw    (251.22,132.09) -- (269.22,132.09) -- (269.22,147.09) -- (251.22,147.09) -- cycle  ;
\draw (260.22,139.59) node  [font=\tiny]  {$R_{Z}$};

\end{tikzpicture}}
\end{align}
and in the last layer, prior to the measurements, there is an additional $R_{X} R_{Z}$ operation on each qubit register.
The canonical centre of the qMERA is in the first CG layer.
Each qubit is connected to at most $2 \log N$ qubits via the CG and Dis layers.
This quantum tensor network is motivated by the MERA in~\cite{GEvenbly09}.

\begin{theorem}
\thlabel{th_mera}
Let $\langle X_{i}\rangle_{\text{qMERA}}$ be the cost function associated with the observable $X_{i}$ and consider the qMERA ansatz defined in Eq.~\eqref{mera_ans1}, then:
\begin{enumerate}
\item $\Var[\partial_{1,1}\CostXi{}_{\text{qMERA}}] \geq \Var[\partial_{1,1}\langle X_{N}\rangle_{\text{qMERA}}]$,
\item $\Var[\partial_{1,1}\langle X_{N}\rangle_{\text{qMERA}}] \geq \frac{1}{4} \cdot \big(\frac{3}{8}\big)^{2n}$.
\end{enumerate}
\end{theorem}
\begin{proof}
See App.~\ref{app_mera}.
\end{proof}

\thref{th_mera} tells us that the qMERA avoids barren plateaus for $1$-local observables.
In contrast to qMPS and qTTN, here the lower bound is not tight.
In App.~\ref{app_mera} we present a numerical method to calculate the exact variances.
Numerically we find that the upper bound scales as $O(N^{-1.2})$ and the lower bound as $\Omega(N^{-2.7})$.

We extend these results to $k$-local observables.
In this case the observable is causally connected to $O(2 k \log N)$ qubits. 
\begin{prop}\thlabel{prop_mera}
If $k \ll N$ then the $k$-local operators $X_{I}$ acting on qubits $I = \{i_{1}, \ldots, i_{k}\}$ satisfy $\Var[\partial_{1,1}\langle X_{I}\rangle_{\text{qMERA}}] \in \Omega(c^{-2 k \log N})$.
\end{prop}

\subsection{Quantum versus classical computational cost of computing gradients}
\label{meth_quvscl}

On a quantum computer we assume that gradients are computed via sampling which has an error scaling as $O(1/\sqrt{M})$ in terms of the sample count $M$~\cite{BhEtAl22}.
Therefore, to resolve gradients decreasing exponentially with the distance from the canonical centre, $M$ needs to scale exponentially with that distance.

On a classical computer the computational cost of basic arithmetic operations (addition, subtraction, multiplication and division) scales polynomially with $\log(1/\epsilon)$ for error $\epsilon$~\cite{Kn97}.
In other words, in classical computing it is efficient to exponentially decrease the error of basic arithmetic operations.
For the quantum tensor networks and local observables considered here, gradients can be evaluated on a classical computer via tensor network contraction techniques (see~\cite{Or14} for MPS, \cite{YShi06} for TTN and~\cite{GEvenbly09} for MERA).
Their computational cost, i.e.\ the total number of arithmetic operations, scales polynomially with the distance of the observable from the canonical centre and, therefore, the total classical computational cost scales polynomially with that distance.

\section{Discussion}
\label{s_conc_disc}

In the context of randomly initialized quantum tensor networks we have shown that qMPS suffer from exponentially vanishing gradients whilst qTTN and qMERA avoid this barren plateau problem.
Therefore qTTN and qMERA are recommended over qMPS.

Interestingly any MPS of bond dimension $\chi$ can be equivalently represented by a TTN of bond dimension $\chi^{2}$~\cite{VeMuCi08, Or14, Orus2019Tensor, CiEtAl21, Ba22}.
Figure~\ref{fig:MPSToTTN} illustrates a constructive procedure for transforming a MPS into a TTN (a) and for transforming a qMPS into a qTTN (b) for $N = 8$.
The same procedure can be used for larger values of $N$ and, for the qMPS considered in this article, leads to a qTTN composed of four-qubit quantum gates.
Since the qTTN circuit depth is logarithmic in the number of qubits the resulting qTTN avoids the barren plateau problem~\cite{MCerezo21}.

\begin{figure}[h]
\centering
\includegraphics[width=0.8\textwidth]{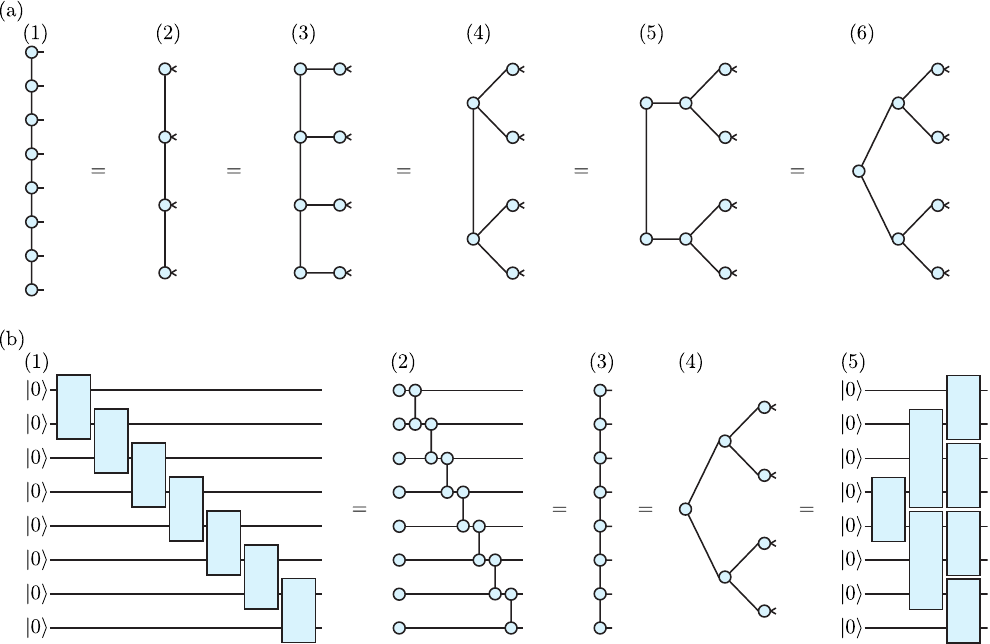}
\caption{\label{fig:MPSToTTN}
(a) We transform a MPS into a TTN by iterating two steps:
1.\ We multiply pairs of adjacent tensors, in (1) $\to$ (2), (3) $\to$ (4), (5) $\to$ (6).
2.\ We perform tensor factorizations, e.g.\ based on the polar decomposition, in (2) $\to$ (3), (4) $\to$ (5).
(b) We transform a qMPS into a qTTN in several steps.
(1) $\to$ (2): Tensor factorizations transform the quantum circuit into a tensor network.
(2) $\to$ (3): We multiply adjacent tensors to obtain the MPS.
(3) $\to$ (4): The MPS is turned into a TTN using the procedure in (a).
(4) $\to$ (5): We construct the qTTN from the TTN.
}
\end{figure}

From the perspective of the barren plateau phenomenon, therefore, generalized versions of qTTN and qMERA with larger unitary gates are recommended over qMPS because they can contain qMPS and their depth scales logarithmically with qubit count.
We conjecture, however, that the classical computation of gradients for these quantum tensor networks can still be more efficient than their quantum computation, cf.\ Sec.~\ref{meth_quvscl}.
Our results show that exhausting the possibilities of classical optimization in the context of variational quantum algorithms can have significant advantages, similar to what was also found in other contexts of quantum computation, e.g.\ Hamiltonian simulation~\cite{KeLu22, MaEtAl23}.

\newpage

\bibliographystyle{quantum}
\bibliography{refs}

\appendix

\newpage

\section{ZX-calculus}
\label{app_zx_meth}

For the sake of completeness, here we summarise the techniques of~\cite{CZhao21} that are relevant for our work.

Let $U(\boldsymbol{\theta})$ be a PQC satisfying the constraints of \thref{zx_assumption} with $\boldsymbol{\theta} \in [-\pi, \pi]^M$.
Then $U(\boldsymbol{\theta}) = c \cdot G_U(\boldsymbol{\theta})$ where $G_U(\boldsymbol{\theta})$ is a graph-like ZX-diagram\footnote{
A graph-like ZX-diagram is composed entirely of so-called $Z$ spiders connected via non-parallel Hadamard edges, without self-loops, in which every input or output is connected to a $Z$ spider and in which every $Z$ spider is connected to at most one input or output (see~\cite{CZhao21} for a nice introduction to the ZX-calculus).
}
representing the circuit $U(\boldsymbol{\theta})$ and $c$ is the constant obtained in the process of turning $U(\boldsymbol{\theta})$ into $G_U(\boldsymbol{\theta})$. 

For example the graph-like ZX-diagram for the $3$-qubit qMPS in Eq.~\eqref{mps_ans} is
\begin{align*}\label{mps_zx_diagram}
 \scalebox{0.8}{\input{tikz_files/mps_zx_diagram}}
\end{align*}
where the intermediate step $\hat{G}_U(\boldsymbol{\theta})$ corresponds to the usual ZX-diagram with the parameters of the $Z$ and $X$ spiders given implicitly and the blue dashed lines in the last step are Hadamard edges.

Thus for a general PQC and $G_U(\boldsymbol{\theta})$ the quantity $\langle H\rangle_{\boldsymbol{\theta}} := \bra{\mathbf{0}}U^\dagger(\boldsymbol{\theta}) H U(\boldsymbol{\theta})\ket{\mathbf{0}}$ is given by the ZX-diagram
\begin{align}
 \scalebox{0.95}{\input{tikz_files/exp_zx_diagram}}
\end{align}
where the prefactor $\frac{1}{2^N}$ comes from the identity $\sqrt{2}\ket{0}=\scalebox{0.8}{\tikzset{every picture/.style={line width=0.75pt}} %set default line width to 0.75pt        

\begin{tikzpicture}[x=0.75pt,y=0.75pt,yscale=-1,xscale=1]
%uncomment if require: \path (0,300); %set diagram left start at 0, and has height of 300

%Straight Lines [id:da21995225937522478] 
\draw    (80,75) -- (100,75) ;
%Shape: Circle [id:dp3016828625020809] 
\draw  [color={rgb, 255:red, 208; green, 2; blue, 27 }  ,draw opacity=1 ][fill={rgb, 255:red, 208; green, 2; blue, 27 }  ,fill opacity=1 ] (70,75) .. controls (70,72.24) and (72.24,70) .. (75,70) .. controls (77.76,70) and (80,72.24) .. (80,75) .. controls (80,77.76) and (77.76,80) .. (75,80) .. controls (72.24,80) and (70,77.76) .. (70,75) -- cycle ;

\end{tikzpicture}}$.

If we initialise the parameters in the quantum circuit uniformly at random $[-\pi, \pi]^M \leftarrow \boldsymbol{\theta}$, then the variance of the gradient with respect to parameter $j$ is
\begin{align}
 \Var[\partial_j\langle H\rangle_{\boldsymbol{\theta}}] = \frac{1}{(2\pi)^M}\int_{\theta_1}\cdots\int_{\theta_M} \big\lvert \partial_j\langle H\rangle\big\rvert^2 d\theta_1 \ldots d\theta_M
\end{align}
where the integrand is given by
\begin{align}
 \scalebox{1}{\input{tikz_files/var_zx_diagram}}
\end{align}

It is shown in~\cite{CZhao21} that the variance integral becomes
\begin{align}\label{var_sum}
\Var[\partial_j\langle H\rangle_{\boldsymbol{\theta}}] = \frac{\abs{c}^2}{4^n} \sum_{a_k\in\{T_1,T_2,T_3\}, k\neq j} V_U^{a_1, \ldots, a_{j-1}, T_2, a_{j+1}, \ldots, a_M}
\end{align}
where 
\begin{align}
 \scalebox{1}{\tikzset{every picture/.style={line width=0.75pt}} %set default line width to 0.75pt        

\begin{tikzpicture}[x=0.75pt,y=0.75pt,yscale=-1,xscale=1]
%uncomment if require: \path (0,300); %set diagram left start at 0, and has height of 300

%Straight Lines [id:da615424709725795] 
\draw    (79,91) -- (90,80) ;
%Shape: Circle [id:dp5002771900276124] 
\draw  [color={rgb, 255:red, 208; green, 2; blue, 27 }  ,draw opacity=1 ][fill={rgb, 255:red, 208; green, 2; blue, 27 }  ,fill opacity=0.2 ] (220,95) .. controls (220,92.24) and (222.24,90) .. (225,90) .. controls (227.76,90) and (230,92.24) .. (230,95) .. controls (230,97.76) and (227.76,100) .. (225,100) .. controls (222.24,100) and (220,97.76) .. (220,95) -- cycle ;
%Shape: Circle [id:dp2073050000425496] 
\draw  [color={rgb, 255:red, 65; green, 117; blue, 5 }  ,draw opacity=1 ][fill={rgb, 255:red, 65; green, 117; blue, 5 }  ,fill opacity=1 ] (70,95) .. controls (70,92.24) and (72.24,90) .. (75,90) .. controls (77.76,90) and (80,92.24) .. (80,95) .. controls (80,97.76) and (77.76,100) .. (75,100) .. controls (72.24,100) and (70,97.76) .. (70,95) -- cycle ;
%Straight Lines [id:da37519211328935675] 
\draw    (60,110) -- (71,99) ;
%Straight Lines [id:da12940878140746004] 
\draw    (79,99) -- (90,110) ;
%Straight Lines [id:da185139940221966] 
\draw    (60,80) -- (71,91) ;
%Shape: Circle [id:dp19557418720946296] 
\draw  [color={rgb, 255:red, 65; green, 117; blue, 5 }  ,draw opacity=1 ][fill={rgb, 255:red, 65; green, 117; blue, 5 }  ,fill opacity=1 ] (200,95) .. controls (200,92.24) and (202.24,90) .. (205,90) .. controls (207.76,90) and (210,92.24) .. (210,95) .. controls (210,97.76) and (207.76,100) .. (205,100) .. controls (202.24,100) and (200,97.76) .. (200,95) -- cycle ;
%Straight Lines [id:da4379980469675284] 
\draw    (190,110) -- (201,99) ;
%Straight Lines [id:da1873618875530838] 
\draw    (190,80) -- (201,91) ;
%Shape: Circle [id:dp7571857314559189] 
\draw  [color={rgb, 255:red, 65; green, 117; blue, 5 }  ,draw opacity=1 ][fill={rgb, 255:red, 65; green, 117; blue, 5 }  ,fill opacity=1 ] (250,95) .. controls (250,97.76) and (247.76,100) .. (245,100) .. controls (242.24,100) and (240,97.76) .. (240,95) .. controls (240,92.24) and (242.24,90) .. (245,90) .. controls (247.76,90) and (250,92.24) .. (250,95) -- cycle ;
%Straight Lines [id:da2695136535679774] 
\draw    (260,80) -- (249,91) ;
%Straight Lines [id:da41405140461429135] 
\draw    (260,110) -- (249,99) ;

%Straight Lines [id:da4061483551391383] 
\draw    (210,95) -- (220,95) ;
%Straight Lines [id:da8141857363829907] 
\draw    (230,95) -- (240,95) ;
%Shape: Circle [id:dp3263907204934682] 
\draw  [color={rgb, 255:red, 208; green, 2; blue, 27 }  ,draw opacity=1 ][fill={rgb, 255:red, 208; green, 2; blue, 27 }  ,fill opacity=0.2 ] (375,90) .. controls (377.76,90) and (380,92.24) .. (380,95) .. controls (380,97.76) and (377.76,100) .. (375,100) .. controls (372.24,100) and (370,97.76) .. (370,95) .. controls (370,92.24) and (372.24,90) .. (375,90) -- cycle ;
%Shape: Circle [id:dp8004113274289975] 
\draw  [color={rgb, 255:red, 65; green, 117; blue, 5 }  ,draw opacity=1 ][fill={rgb, 255:red, 65; green, 117; blue, 5 }  ,fill opacity=1 ] (375,70) .. controls (377.76,70) and (380,72.24) .. (380,75) .. controls (380,77.76) and (377.76,80) .. (375,80) .. controls (372.24,80) and (370,77.76) .. (370,75) .. controls (370,72.24) and (372.24,70) .. (375,70) -- cycle ;
%Straight Lines [id:da904477056527123] 
\draw    (360,60) -- (371,71) ;
%Straight Lines [id:da04300130923769885] 
\draw    (390,60) -- (379,71) ;
%Shape: Circle [id:dp6528193688069592] 
\draw  [color={rgb, 255:red, 65; green, 117; blue, 5 }  ,draw opacity=1 ][fill={rgb, 255:red, 65; green, 117; blue, 5 }  ,fill opacity=1 ] (375,120) .. controls (372.24,120) and (370,117.76) .. (370,115) .. controls (370,112.24) and (372.24,110) .. (375,110) .. controls (377.76,110) and (380,112.24) .. (380,115) .. controls (380,117.76) and (377.76,120) .. (375,120) -- cycle ;
%Straight Lines [id:da48952029187458024] 
\draw    (390,130) -- (379,119) ;
%Straight Lines [id:da8056519126717006] 
\draw    (360,130) -- (371,119) ;

%Straight Lines [id:da11397286363561987] 
\draw    (375,80) -- (375,90) ;
%Straight Lines [id:da9412855315078632] 
\draw    (375,100) -- (375,110) ;

% Text Node
\draw (221,91.4) node [anchor=north west][inner sep=0.75pt]  [font=\tiny]  {$\pi $};
% Text Node
\draw (371,91.4) node [anchor=north west][inner sep=0.75pt]  [font=\tiny]  {$\pi $};
% Text Node
\draw (21,88.4) node [anchor=north west][inner sep=0.75pt]  [font=\footnotesize]  {$T_{1} =$};
% Text Node
\draw (151,88.4) node [anchor=north west][inner sep=0.75pt]  [font=\footnotesize]  {$T_{2} =$};
% Text Node
\draw (321,88.4) node [anchor=north west][inner sep=0.75pt]  [font=\footnotesize]  {$T_{3} =$};

\end{tikzpicture}}
\end{align}
and for $a_i \in \{T_1, T_2, T_3\}$:
\begin{align}
 \scalebox{1}{\input{tikz_files/vU_zx}}
\end{align}

Analytically computing the sum of $3^{M-1}$ terms in Eq.~\eqref{var_sum} is inefficient and does not clarify whether the parameterised circuit $U(\boldsymbol{\theta})$ suffers from barren plateaus.
Thus the authors in~\cite{CZhao21} devise a way to simplify the expression~\eqref{var_sum} directly from $G_U(\boldsymbol{\theta})$.
Given a graph-like PQC
\begin{align}\label{graph_like_pqc}
 \scalebox{1}{\input{tikz_files/GU_zx_diagram}}
\end{align}
it is shown in~\cite{CZhao21} that 
\begin{align}\label{var_as_tensor_net}
 \Var[\partial_j\langle H\rangle_{\boldsymbol{\theta}}] = \frac{\abs{c}^2}{4^n} \sum V_U^{a_1,\ldots,a_{k},b_1,\ldots,b_l,\ldots,T_2,\ldots,c_{1},\ldots,c_M} \\
 \scalebox{1}{\input{tikz_files/var_sum_tn_diagram}}
\end{align}
where the parameters $\theta_{k \neq j}$ in the graph-like PQC in~\eqref{graph_like_pqc} are replaced by the copy tensor
\begin{align}
 \scalebox{0.8}{\tikzset{every picture/.style={line width=0.75pt}} %set default line width to 0.75pt        

\begin{tikzpicture}[x=0.75pt,y=0.75pt,yscale=-1,xscale=1]
%uncomment if require: \path (0,341); %set diagram left start at 0, and has height of 341

%Straight Lines [id:da7473270556661318] 
\draw    (92,164) -- (110,180) ;
\draw [shift={(110,180)}, rotate = 41.63] [color={rgb, 255:red, 0; green, 0; blue, 0 }  ][fill={rgb, 255:red, 0; green, 0; blue, 0 }  ][line width=0.75]      (0, 0) circle [x radius= 3.35, y radius= 3.35]   ;
%Straight Lines [id:da08466913376222074] 
\draw    (110,180) -- (128,164) ;
%Straight Lines [id:da10958641751602527] 
\draw    (92,196) -- (110,180) ;
%Straight Lines [id:da4734570580585813] 
\draw    (110,180) -- (128,196) ;
%Shape: Brace [id:dp3807291421458434] 
\draw   (80.5,159.5) .. controls (75.83,159.56) and (73.53,161.92) .. (73.59,166.59) -- (73.63,169.84) .. controls (73.71,176.51) and (71.42,179.87) .. (66.75,179.92) .. controls (71.42,179.87) and (73.79,183.17) .. (73.88,189.84)(73.84,186.84) -- (73.92,193.09) .. controls (73.97,197.76) and (76.33,200.06) .. (81,200) ;
%Shape: Brace [id:dp18366990654084336] 
\draw   (140.5,200.5) .. controls (145.17,200.5) and (147.5,198.17) .. (147.5,193.5) -- (147.5,190.5) .. controls (147.5,183.83) and (149.83,180.5) .. (154.5,180.5) .. controls (149.83,180.5) and (147.5,177.17) .. (147.5,170.5)(147.5,173.5) -- (147.5,167.5) .. controls (147.5,162.83) and (145.17,160.5) .. (140.5,160.5) ;

% Text Node
\draw (123,165) node [anchor=north west][inner sep=0.75pt]  [font=\tiny]  {$\vdots $};
% Text Node
\draw (89,165) node [anchor=north west][inner sep=0.75pt]  [font=\tiny]  {$\vdots $};
% Text Node
\draw (53,177) node [anchor=north west][inner sep=0.75pt]  [font=\tiny] [align=left] {$\displaystyle in$};
% Text Node
\draw (159,177) node [anchor=north west][inner sep=0.75pt]  [font=\tiny] [align=left] {$\displaystyle out$};

\end{tikzpicture}} = \sum_{i=0}^2\ket{i}^{\otimes \text{in}} \bra{i}^{\otimes \text{out}},
\end{align}
the parameter $\theta_j$ is replaced by the projection onto the second dimension
\begin{align}
 \scalebox{0.8}{\tikzset{every picture/.style={line width=0.75pt}} %set default line width to 0.75pt        

\begin{tikzpicture}[x=0.75pt,y=0.75pt,yscale=-1,xscale=1]
%uncomment if require: \path (0,341); %set diagram left start at 0, and has height of 341

%Straight Lines [id:da8410377729541016] 
\draw    (85,165) -- (101,176) ;
%Straight Lines [id:da5641822461178763] 
\draw    (119,176) -- (134,166) ;
%Straight Lines [id:da614161235327565] 
\draw    (86,194) -- (101,182) ;
%Straight Lines [id:da7327111618458788] 
\draw    (119,182) -- (134,192) ;
%Shape: Brace [id:dp8915888977142439] 
\draw   (80.5,159.5) .. controls (75.83,159.56) and (73.53,161.92) .. (73.59,166.59) -- (73.63,169.84) .. controls (73.71,176.51) and (71.42,179.87) .. (66.75,179.92) .. controls (71.42,179.87) and (73.79,183.17) .. (73.88,189.84)(73.84,186.84) -- (73.92,193.09) .. controls (73.97,197.76) and (76.33,200.06) .. (81,200) ;
%Shape: Brace [id:dp013882918341818984] 
\draw   (140.5,200.5) .. controls (145.17,200.5) and (147.5,198.17) .. (147.5,193.5) -- (147.5,190.5) .. controls (147.5,183.83) and (149.83,180.5) .. (154.5,180.5) .. controls (149.83,180.5) and (147.5,177.17) .. (147.5,170.5)(147.5,173.5) -- (147.5,167.5) .. controls (147.5,162.83) and (145.17,160.5) .. (140.5,160.5) ;

% Text Node
\draw (127,165) node [anchor=north west][inner sep=0.75pt]  [font=\tiny]  {$\vdots $};
% Text Node
\draw (84,165) node [anchor=north west][inner sep=0.75pt]  [font=\tiny]  {$\vdots $};
% Text Node
\draw (53,177) node [anchor=north west][inner sep=0.75pt]  [font=\tiny] [align=left] {$\displaystyle in$};
% Text Node
\draw (159,177) node [anchor=north west][inner sep=0.75pt]  [font=\tiny] [align=left] {$\displaystyle out$};
% Text Node
\draw    (101,172) -- (119,172) -- (119,187) -- (101,187) -- cycle  ;
\draw (103,174) node [anchor=north west][inner sep=0.75pt]  [font=\tiny]  {$P_{2}$};

\end{tikzpicture}} = \sum_{i=0}^2 \ket{i}^{\otimes \text{in}} \bra{1}^{\otimes \text{out}}
\end{align}
and each Hadamard edge is replaced by the $3 \times 3$ matrix
\begin{align}\label{m_matrix}
 \scalebox{0.8}{\tikzset{every picture/.style={line width=0.75pt}} %set default line width to 0.75pt        

\begin{tikzpicture}[x=0.75pt,y=0.75pt,yscale=-1,xscale=1]
%uncomment if require: \path (0,341); %set diagram left start at 0, and has height of 341

%Straight Lines [id:da6061905982014946] 
\draw    (310,270) -- (327.65,270) ;
\draw [shift={(330,270)}, rotate = 0] [color={rgb, 255:red, 0; green, 0; blue, 0 }  ][line width=0.75]      (0, 0) circle [x radius= 3.35, y radius= 3.35]   ;
%Straight Lines [id:da08891103861243499] 
\draw    (333,270) -- (350,270) ;

\end{tikzpicture}} = \frac{1}{4}
  \begin{bmatrix}
   1 & 1 & 1\\
   1 & 1 & -1\\
   1 & -1 & 1
  \end{bmatrix} =: M.
\end{align}
The tensors $\Tilde{I}_{a_1,\ldots,a_j}$ and $\Tilde{H}_{c_1,\ldots,c_M}$ are related to the initial state and the observable $H$, respectively.
In this article the initial state is $\ket{0}^{\otimes N}$ and so $\Tilde{I}$ is
\begin{align}
 \scalebox{1}{\tikzset{every picture/.style={line width=0.75pt}} %set default line width to 0.75pt        

\begin{tikzpicture}[x=0.75pt,y=0.75pt,yscale=-1,xscale=1]
%uncomment if require: \path (0,341); %set diagram left start at 0, and has height of 341

%Straight Lines [id:da30322808568933235] 
\draw    (130,160) -- (110,160) ;
\draw [shift={(110,160)}, rotate = 180] [color={rgb, 255:red, 0; green, 0; blue, 0 }  ][fill={rgb, 255:red, 0; green, 0; blue, 0 }  ][line width=0.75]      (0, 0) circle [x radius= 3.35, y radius= 3.35]   ;
%Shape: Brace [id:dp7660045401948781] 
\draw   (81,150) .. controls (76.33,149.97) and (73.99,152.29) .. (73.96,156.96) -- (73.81,180.46) .. controls (73.77,187.13) and (71.42,190.44) .. (66.75,190.41) .. controls (71.42,190.44) and (73.73,193.79) .. (73.69,200.46)(73.71,197.46) -- (73.55,223.96) .. controls (73.52,228.63) and (75.84,230.97) .. (80.51,231) ;
%Straight Lines [id:da5912111509066302] 
\draw    (130,181) -- (110,181) ;
\draw [shift={(110,181)}, rotate = 180] [color={rgb, 255:red, 0; green, 0; blue, 0 }  ][fill={rgb, 255:red, 0; green, 0; blue, 0 }  ][line width=0.75]      (0, 0) circle [x radius= 3.35, y radius= 3.35]   ;
%Straight Lines [id:da5614470323802472] 
\draw    (130,218) -- (110,218) ;
\draw [shift={(110,218)}, rotate = 180] [color={rgb, 255:red, 0; green, 0; blue, 0 }  ][fill={rgb, 255:red, 0; green, 0; blue, 0 }  ][line width=0.75]      (0, 0) circle [x radius= 3.35, y radius= 3.35]   ;

% Text Node
\draw (51,187) node [anchor=north west][inner sep=0.75pt]  [font=\tiny] [align=left] {$\displaystyle N$};
% Text Node
\draw (91,151.4) node [anchor=north west][inner sep=0.75pt]  [font=\tiny]  {$\frac{1}{4}$};
% Text Node
\draw (91,172.4) node [anchor=north west][inner sep=0.75pt]  [font=\tiny]  {$\frac{1}{4}$};
% Text Node
\draw (91,209.4) node [anchor=north west][inner sep=0.75pt]  [font=\tiny]  {$\frac{1}{4}$};
% Text Node
\draw (116,185) node [anchor=north west][inner sep=0.75pt]  [font=\tiny]  {$\vdots $};
% Text Node
\draw (107,148.4) node [anchor=north west][inner sep=0.75pt]  [font=\tiny]  {$2$};
% Text Node
\draw (107,170.4) node [anchor=north west][inner sep=0.75pt]  [font=\tiny]  {$2$};
% Text Node
\draw (107,207.4) node [anchor=north west][inner sep=0.75pt]  [font=\tiny]  {$2$};

\end{tikzpicture}}
\end{align}
If $H = \boldsymbol{\sigma}_1\otimes\cdots\otimes\boldsymbol{\sigma}_N$ where $\boldsymbol{\sigma}_i = k_{i,0} I + k_{i,1} X + k_{i,2} Y + k_{i,3} Z$ is a sum of Pauli terms acting on qubit $i$, it is proven in~\cite{CZhao21} that $\Tilde{H} = \mathbf{u}_1 \otimes \cdots \otimes \mathbf{u}_N$, where 
\begin{align}\label{zx_ham}
 u_i = 2 k_{i,0}^2 \vot + 2 (k_{i,1}^2 + k_{i,3}^2) \vt + 2 k_{i,2}^2 \votm
\end{align}
and
\begin{align}\label{123vectors}
 \vot = \begin{bmatrix}
  1\\
  0\\
  1
 \end{bmatrix}, \qquad \vt = \begin{bmatrix}
  0\\
  1\\
  0
 \end{bmatrix}, \qquad \votm = \begin{bmatrix}
  1\\
  0\\
  -1
 \end{bmatrix}.
\end{align}

Continuing the example for the $3$-qubit qMPS from the beginning of this Appendix, the variance of the gradient of the first (top left) parameter for the observable $H = X_3$ can be found by evaluating
\begin{align}
 \hspace{-15mm}\scalebox{1}{\input{tikz_files/mps_zxtn_diagram}}
\end{align}

\newpage

\section{Quantum matrix product states}
\label{app_mps}

The qMPS ansatz of Eq.~\eqref{mps_ans} for $N$ qubits has the form
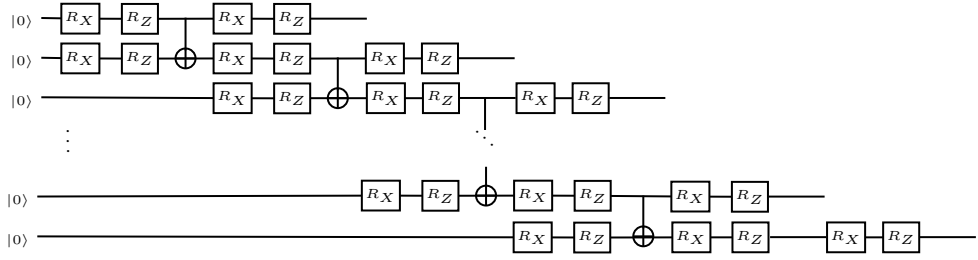
\begin{figure}[h]
\centering
\tikzset{every picture/.style={line width=0.75pt}} %set default line width to 0.75pt        

\begin{tikzpicture}[x=0.75pt,y=0.75pt,yscale=-1,xscale=1]
%uncomment if require: \path (0,300); %set diagram left start at 0, and has height of 300

\draw   (105.5,40.5) .. controls (105.5,37.74) and (107.74,35.5) .. (110.5,35.5) .. controls (113.26,35.5) and (115.5,37.74) .. (115.5,40.5) .. controls (115.5,43.26) and (113.26,45.5) .. (110.5,45.5) .. controls (107.74,45.5) and (105.5,43.26) .. (105.5,40.5) -- cycle ; \draw   (105.5,40.5) -- (115.5,40.5) ; \draw   (110.5,35.5) -- (110.5,45.5) ;
%Straight Lines [id:da09744905377489355] 
\draw    (110.5,20) -- (110.5,36.5) ;
\draw   (181.5,60.5) .. controls (181.5,57.74) and (183.74,55.5) .. (186.5,55.5) .. controls (189.26,55.5) and (191.5,57.74) .. (191.5,60.5) .. controls (191.5,63.26) and (189.26,65.5) .. (186.5,65.5) .. controls (183.74,65.5) and (181.5,63.26) .. (181.5,60.5) -- cycle ; \draw   (181.5,60.5) -- (191.5,60.5) ; \draw   (186.5,55.5) -- (186.5,65.5) ;
%Straight Lines [id:da2279031340831652] 
\draw    (186.5,40) -- (186.5,56.5) ;
%Straight Lines [id:da2402118632043717] 
\draw    (173,60.5) -- (201,60.5) ;
%Straight Lines [id:da992118952248064] 
\draw    (173,20.5) -- (201,20.5) ;
\draw   (255.5,109.5) .. controls (255.5,106.74) and (257.74,104.5) .. (260.5,104.5) .. controls (263.26,104.5) and (265.5,106.74) .. (265.5,109.5) .. controls (265.5,112.26) and (263.26,114.5) .. (260.5,114.5) .. controls (257.74,114.5) and (255.5,112.26) .. (255.5,109.5) -- cycle ; \draw   (255.5,109.5) -- (265.5,109.5) ; \draw   (260.5,104.5) -- (260.5,114.5) ;
%Straight Lines [id:da6457246629096087] 
\draw    (260.5,95) -- (260.5,105.5) ;
\draw   (334,130) .. controls (334,127.24) and (336.24,125) .. (339,125) .. controls (341.76,125) and (344,127.24) .. (344,130) .. controls (344,132.76) and (341.76,135) .. (339,135) .. controls (336.24,135) and (334,132.76) .. (334,130) -- cycle ; \draw   (334,130) -- (344,130) ; \draw   (339,125) -- (339,135) ;
%Straight Lines [id:da4018335530806998] 
\draw    (339,109.5) -- (339,126) ;
%Straight Lines [id:da526322370173409] 
\draw    (401.5,110) -- (429.5,110) ;
%Straight Lines [id:da15984553006383195] 
\draw    (246.78,40.4) -- (274.78,40.4) ;
%Straight Lines [id:da1640972581809248] 
\draw    (247.28,60.4) -- (275.28,60.4) ;
%Straight Lines [id:da0527770498962814] 
\draw    (260,60) -- (260,76.5) ;
%Straight Lines [id:da1385326892573906] 
\draw    (322,60.4) -- (350,60.4) ;
%Straight Lines [id:da1399135055075753] 
\draw    (402.28,130.4) -- (430.28,130.4) ;
%Straight Lines [id:da11921284844975633] 
\draw    (477,130.4) -- (505,130.4) ;

% Text Node
\draw    (48.5,13) -- (67.5,13) -- (67.5,28) -- (48.5,28) -- cycle  ;
\draw (58,20.5) node  [font=\tiny]  {$R_{X}$};
% Text Node
\draw (21.5,17.4) node [anchor=north west][inner sep=0.75pt]  [font=\tiny]  {$\ket{0}$};
% Text Node
\draw    (78.72,13.1) -- (96.72,13.1) -- (96.72,28.1) -- (78.72,28.1) -- cycle  ;
\draw (87.72,20.6) node  [font=\tiny]  {$R_{Z}$};
% Text Node
\draw    (48.5,33) -- (67.5,33) -- (67.5,48) -- (48.5,48) -- cycle  ;
\draw (58,40.5) node  [font=\tiny]  {$R_{X}$};
% Text Node
\draw (21.5,37.4) node [anchor=north west][inner sep=0.75pt]  [font=\tiny]  {$\ket{0}$};
% Text Node
\draw    (78.72,33.1) -- (96.72,33.1) -- (96.72,48.1) -- (78.72,48.1) -- cycle  ;
\draw (87.72,40.6) node  [font=\tiny]  {$R_{Z}$};
% Text Node
\draw    (124.5,13) -- (143.5,13) -- (143.5,28) -- (124.5,28) -- cycle  ;
\draw (134,20.5) node  [font=\tiny]  {$R_{X}$};
% Text Node
\draw    (154.72,13.1) -- (172.72,13.1) -- (172.72,28.1) -- (154.72,28.1) -- cycle  ;
\draw (163.72,20.6) node  [font=\tiny]  {$R_{Z}$};
% Text Node
\draw    (124.5,33) -- (143.5,33) -- (143.5,48) -- (124.5,48) -- cycle  ;
\draw (134,40.5) node  [font=\tiny]  {$R_{X}$};
% Text Node
\draw    (154.72,33.1) -- (172.72,33.1) -- (172.72,48.1) -- (154.72,48.1) -- cycle  ;
\draw (163.72,40.6) node  [font=\tiny]  {$R_{Z}$};
% Text Node
\draw    (124.5,53) -- (143.5,53) -- (143.5,68) -- (124.5,68) -- cycle  ;
\draw (134,60.5) node  [font=\tiny]  {$R_{X}$};
% Text Node
\draw    (154.72,53.1) -- (172.72,53.1) -- (172.72,68.1) -- (154.72,68.1) -- cycle  ;
\draw (163.72,60.6) node  [font=\tiny]  {$R_{Z}$};
% Text Node
\draw    (200.5,33) -- (219.5,33) -- (219.5,48) -- (200.5,48) -- cycle  ;
\draw (210,40.5) node  [font=\tiny]  {$R_{X}$};
% Text Node
\draw (21.5,57.4) node [anchor=north west][inner sep=0.75pt]  [font=\tiny]  {$\ket{0}$};
% Text Node
\draw (19.5,107.4) node [anchor=north west][inner sep=0.75pt]  [font=\tiny]  {$\ket{0}$};
% Text Node
\draw (19.5,127.4) node [anchor=north west][inner sep=0.75pt]  [font=\tiny]  {$\ket{0}$};
% Text Node
\draw    (274.5,102) -- (293.5,102) -- (293.5,117) -- (274.5,117) -- cycle  ;
\draw (284,109.5) node  [font=\tiny]  {$R_{X}$};
% Text Node
\draw    (304.72,102.1) -- (322.72,102.1) -- (322.72,117.1) -- (304.72,117.1) -- cycle  ;
\draw (313.72,109.6) node  [font=\tiny]  {$R_{Z}$};
% Text Node
\draw    (198.5,102) -- (217.5,102) -- (217.5,117) -- (198.5,117) -- cycle  ;
\draw (208,109.5) node  [font=\tiny]  {$R_{X}$};
% Text Node
\draw    (228.72,102.1) -- (246.72,102.1) -- (246.72,117.1) -- (228.72,117.1) -- cycle  ;
\draw (237.72,109.6) node  [font=\tiny]  {$R_{Z}$};
% Text Node
\draw    (353,102.5) -- (372,102.5) -- (372,117.5) -- (353,117.5) -- cycle  ;
\draw (362.5,110) node  [font=\tiny]  {$R_{X}$};
% Text Node
\draw    (383.22,102.6) -- (401.22,102.6) -- (401.22,117.6) -- (383.22,117.6) -- cycle  ;
\draw (392.22,110.1) node  [font=\tiny]  {$R_{Z}$};
% Text Node
\draw    (353.5,123) -- (372.5,123) -- (372.5,138) -- (353.5,138) -- cycle  ;
\draw (363,130.5) node  [font=\tiny]  {$R_{X}$};
% Text Node
\draw    (274.28,123) -- (293.28,123) -- (293.28,138) -- (274.28,138) -- cycle  ;
\draw (283.78,130.5) node  [font=\tiny]  {$R_{X}$};
% Text Node
\draw    (304.5,123.1) -- (322.5,123.1) -- (322.5,138.1) -- (304.5,138.1) -- cycle  ;
\draw (313.5,130.6) node  [font=\tiny]  {$R_{Z}$};
% Text Node
\draw    (383.5,123) -- (401.5,123) -- (401.5,138) -- (383.5,138) -- cycle  ;
\draw (392.5,130.5) node  [font=\tiny]  {$R_{Z}$};
% Text Node
\draw    (228.5,33) -- (246.5,33) -- (246.5,48) -- (228.5,48) -- cycle  ;
\draw (237.5,40.5) node  [font=\tiny]  {$R_{Z}$};
% Text Node
\draw (256.23,72.54) node [anchor=north west][inner sep=0.75pt]  [font=\scriptsize,rotate=-42.8]  {$\cdots $};
% Text Node
\draw (55.21,74.04) node [anchor=north west][inner sep=0.75pt]  [font=\scriptsize,rotate=-88.98]  {$\cdots $};
% Text Node
\draw    (201,53) -- (220,53) -- (220,68) -- (201,68) -- cycle  ;
\draw (210.5,60.5) node  [font=\tiny]  {$R_{X}$};
% Text Node
\draw    (229,53) -- (247,53) -- (247,68) -- (229,68) -- cycle  ;
\draw (238,60.5) node  [font=\tiny]  {$R_{Z}$};
% Text Node
\draw    (275.72,53) -- (294.72,53) -- (294.72,68) -- (275.72,68) -- cycle  ;
\draw (285.22,60.5) node  [font=\tiny]  {$R_{X}$};
% Text Node
\draw    (303.72,53) -- (321.72,53) -- (321.72,68) -- (303.72,68) -- cycle  ;
\draw (312.72,60.5) node  [font=\tiny]  {$R_{Z}$};
% Text Node
\draw    (430.72,123) -- (449.72,123) -- (449.72,138) -- (430.72,138) -- cycle  ;
\draw (440.22,130.5) node  [font=\tiny]  {$R_{X}$};
% Text Node
\draw    (458.72,123) -- (476.72,123) -- (476.72,138) -- (458.72,138) -- cycle  ;
\draw (467.72,130.5) node  [font=\tiny]  {$R_{Z}$};
% Connection
\draw    (48.5,20.34) -- (38.5,20.17) ;
% Connection
\draw    (78.72,20.57) -- (67.5,20.53) ;
% Connection
\draw    (48.5,40.34) -- (38.5,40.17) ;
% Connection
\draw    (78.72,40.57) -- (67.5,40.53) ;
% Connection
\draw    (154.72,20.57) -- (143.5,20.53) ;
% Connection
\draw    (96.72,20.58) -- (124.5,20.52) ;
% Connection
\draw    (154.72,40.57) -- (143.5,40.53) ;
% Connection
\draw    (154.72,60.57) -- (143.5,60.53) ;
% Connection
\draw    (172.72,40.58) -- (200.5,40.52) ;
% Connection
\draw    (96.72,40.58) -- (124.5,40.52) ;
% Connection
\draw [line width=0.75]    (38.5,60.05) -- (124.5,60.45) ;
% Connection
\draw    (304.72,109.57) -- (293.5,109.53) ;
% Connection
\draw    (228.72,109.57) -- (217.5,109.53) ;
% Connection
\draw    (246.72,109.58) -- (274.5,109.52) ;
% Connection
\draw    (36.5,109.97) -- (198.5,109.53) ;
% Connection
\draw    (383.22,110.07) -- (372,110.03) ;
% Connection
\draw    (322.72,109.68) -- (353,109.92) ;
% Connection
\draw    (304.5,130.57) -- (293.28,130.53) ;
% Connection
\draw    (322.5,130.58) -- (353.5,130.52) ;
% Connection
\draw    (36.5,130.02) -- (274.28,130.48) ;
% Connection
\draw    (219.5,40.5) -- (228.5,40.5) ;
% Connection
\draw    (372.5,130.5) -- (383.5,130.5) ;
% Connection
\draw    (220,60.5) -- (229,60.5) ;
% Connection
\draw    (294.72,60.5) -- (303.72,60.5) ;
% Connection
\draw    (449.72,130.5) -- (458.72,130.5) ;

\end{tikzpicture}
\caption{\label{fig:mps_circ}
The qMPS circuit considered in this article.
}
\end{figure}

\noindent We index parameters using the index pair $(j,k)$ which refers to the $k$-th parameter in qubit register $j = 1, \ldots, N$.

\thref{bp_theorem_zx} and App.~\ref{app_zx_meth} imply that
\begin{align}\label{mps_var}
 \scalebox{0.85}{\input{tikz_files/mps_tn}}
\end{align}
where the gradient is calculated for the first parameter on the first qubit register and the vectors $u_i$ are related to the observables via Eq.~\eqref{zx_ham}.
To consider general parameters $(j,k)$ we simply move the projection $P_2$ to the copy tensor at position $(j,k)$.
Using the identities 
\begin{align}\label{zx_identities1}
 2 M \vot = \vot, \quad 2 M \vt = \frac{1}{2}(\vt + \votm), \quad 2 M \votm = \vt,
\end{align}
and 
\begin{align}\label{zx_identities2}
 \scalebox{0.8}{\tikzset{every picture/.style={line width=0.75pt}} %set default line width to 0.75pt        

\begin{tikzpicture}[x=0.75pt,y=0.75pt,yscale=-1,xscale=1]
%uncomment if require: \path (0,300); %set diagram left start at 0, and has height of 300

%Shape: Output [id:dp2187272395455171] 
\draw  [fill={rgb, 255:red, 0; green, 0; blue, 0 }  ,fill opacity=1 ] (210,75) .. controls (210,77.76) and (207.76,80) .. (205,80) .. controls (202.24,80) and (200,77.76) .. (200,75) .. controls (200,72.24) and (202.24,70) .. (205,70) .. controls (207.76,70) and (210,72.24) .. (210,75) -- cycle (215,75) -- (210,75) (195,75) -- (200,75) ;
%Shape: Output [id:dp5492089357379177] 
\draw   (190,75) .. controls (190,77.76) and (187.76,80) .. (185,80) .. controls (182.24,80) and (180,77.76) .. (180,75) .. controls (180,72.24) and (182.24,70) .. (185,70) .. controls (187.76,70) and (190,72.24) .. (190,75) -- cycle (195,75) -- (190,75) (175,75) -- (180,75) ;
%Shape: Output [id:dp44748875960195855] 
\draw  [fill={rgb, 255:red, 0; green, 0; blue, 0 }  ,fill opacity=1 ] (210,45) .. controls (210,47.76) and (207.76,50) .. (205,50) .. controls (202.24,50) and (200,47.76) .. (200,45) .. controls (200,42.24) and (202.24,40) .. (205,40) .. controls (207.76,40) and (210,42.24) .. (210,45) -- cycle (215,45) -- (210,45) (195,45) -- (200,45) ;
%Shape: Output [id:dp2919849998608235] 
\draw   (190,45.04) .. controls (190,47.8) and (187.76,50.04) .. (185,50.04) .. controls (182.24,50.04) and (180,47.8) .. (180,45.04) .. controls (180,42.28) and (182.24,40.04) .. (185,40.04) .. controls (187.76,40.04) and (190,42.28) .. (190,45.04) -- cycle (195,45.04) -- (190,45.04) (175,45.04) -- (180,45.04) ;
%Shape: Output [id:dp7740240088158645] 
\draw   (205,55) .. controls (207.76,55) and (210,57.24) .. (210,60) .. controls (210,62.76) and (207.76,65) .. (205,65) .. controls (202.24,65) and (200,62.76) .. (200,60) .. controls (200,57.24) and (202.24,55) .. (205,55) -- cycle (205,50) -- (205,55) (205,70) -- (205,65) ;
%Straight Lines [id:da3819177854705733] 
\draw    (280,76) -- (290,76) ;
%Straight Lines [id:da9585846691780942] 
\draw    (324,75) -- (334,75) ;
%Straight Lines [id:da6201829781470727] 
\draw    (300,47) -- (310,47) ;

% Text Node
\draw (202,32.4) node [anchor=north west][inner sep=0.75pt]  [font=\tiny]  {$2$};
% Text Node
\draw (202,82.4) node [anchor=north west][inner sep=0.75pt]  [font=\tiny]  {$2$};
% Text Node
\draw (192,57.4) node [anchor=north west][inner sep=0.75pt]  [font=\tiny]  {$4$};
% Text Node
\draw    (215,38) -- (234,38) -- (234,53) -- (215,53) -- cycle  ;
\draw (216,42.4) node [anchor=north west][inner sep=0.75pt]  [font=\tiny]  {$v_{13}$};
% Text Node
\draw    (215,67) -- (233,67) -- (233,82) -- (215,82) -- cycle  ;
\draw (218,71.4) node [anchor=north west][inner sep=0.75pt]  [font=\tiny]  {$v_{2}$};
% Text Node
\draw (238,53.4) node [anchor=north west][inner sep=0.75pt]  [font=\scriptsize]  {$=$};
% Text Node
\draw (258,50) node [anchor=north west][inner sep=0.75pt]  [font=\small]  {$\frac{1}{2}$};
% Text Node
\draw (268,40.4) node [anchor=north west][inner sep=0.75pt]    {$\begin{cases}
 & \\
 & 
\end{cases}$};
% Text Node
\draw    (290,68) -- (308,68) -- (308,83) -- (290,83) -- cycle  ;
\draw (293,72.4) node [anchor=north west][inner sep=0.75pt]  [font=\tiny]  {$v_{2}$};
% Text Node
\draw    (334,67) -- (353,67) -- (353,84) -- (334,84) -- cycle  ;
\draw (335,68) node [anchor=north west][inner sep=0.75pt]  [font=\tiny]  {$v_{13}^{-}$};
% Text Node
\draw    (310,39) -- (328,39) -- (328,54) -- (310,54) -- cycle  ;
\draw (313,43.4) node [anchor=north west][inner sep=0.75pt]  [font=\tiny]  {$v_{2}$};
% Text Node
\draw (311,68.4) node [anchor=north west][inner sep=0.75pt]  [font=\scriptsize]  {$+$};

\end{tikzpicture}}, \qquad \scalebox{0.8}{\tikzset{every picture/.style={line width=0.75pt}} %set default line width to 0.75pt        

\begin{tikzpicture}[x=0.75pt,y=0.75pt,yscale=-1,xscale=1]
%uncomment if require: \path (0,300); %set diagram left start at 0, and has height of 300

%Shape: Output [id:dp29625998119911445] 
\draw  [fill={rgb, 255:red, 0; green, 0; blue, 0 }  ,fill opacity=1 ] (210,75) .. controls (210,77.76) and (207.76,80) .. (205,80) .. controls (202.24,80) and (200,77.76) .. (200,75) .. controls (200,72.24) and (202.24,70) .. (205,70) .. controls (207.76,70) and (210,72.24) .. (210,75) -- cycle (215,75) -- (210,75) (195,75) -- (200,75) ;
%Shape: Output [id:dp940287350841394] 
\draw   (190,75) .. controls (190,77.76) and (187.76,80) .. (185,80) .. controls (182.24,80) and (180,77.76) .. (180,75) .. controls (180,72.24) and (182.24,70) .. (185,70) .. controls (187.76,70) and (190,72.24) .. (190,75) -- cycle (195,75) -- (190,75) (175,75) -- (180,75) ;
%Shape: Output [id:dp83818559997976] 
\draw  [fill={rgb, 255:red, 0; green, 0; blue, 0 }  ,fill opacity=1 ] (210,45) .. controls (210,47.76) and (207.76,50) .. (205,50) .. controls (202.24,50) and (200,47.76) .. (200,45) .. controls (200,42.24) and (202.24,40) .. (205,40) .. controls (207.76,40) and (210,42.24) .. (210,45) -- cycle (215,45) -- (210,45) (195,45) -- (200,45) ;
%Shape: Output [id:dp5679466385647587] 
\draw   (190,45.04) .. controls (190,47.8) and (187.76,50.04) .. (185,50.04) .. controls (182.24,50.04) and (180,47.8) .. (180,45.04) .. controls (180,42.28) and (182.24,40.04) .. (185,40.04) .. controls (187.76,40.04) and (190,42.28) .. (190,45.04) -- cycle (195,45.04) -- (190,45.04) (175,45.04) -- (180,45.04) ;
%Shape: Output [id:dp12601928771852489] 
\draw   (205,55) .. controls (207.76,55) and (210,57.24) .. (210,60) .. controls (210,62.76) and (207.76,65) .. (205,65) .. controls (202.24,65) and (200,62.76) .. (200,60) .. controls (200,57.24) and (202.24,55) .. (205,55) -- cycle (205,50) -- (205,55) (205,70) -- (205,65) ;
%Straight Lines [id:da36057900516698904] 
\draw    (271,73) -- (281,73) ;
%Straight Lines [id:da3014031816375222] 
\draw    (271,47) -- (281,47) ;

% Text Node
\draw (202,32.4) node [anchor=north west][inner sep=0.75pt]  [font=\tiny]  {$2$};
% Text Node
\draw (202,82.4) node [anchor=north west][inner sep=0.75pt]  [font=\tiny]  {$2$};
% Text Node
\draw (192,57.4) node [anchor=north west][inner sep=0.75pt]  [font=\tiny]  {$4$};
% Text Node
\draw    (215,38) -- (234,38) -- (234,53) -- (215,53) -- cycle  ;
\draw (216,42.4) node [anchor=north west][inner sep=0.75pt]  [font=\tiny]  {$v_{13}$};
% Text Node
\draw    (215,67) -- (234,67) -- (234,84) -- (215,84) -- cycle  ;
\draw (216,68) node [anchor=north west][inner sep=0.75pt]  [font=\tiny]  {$v_{13}^{-}$};
% Text Node
\draw (238,53.4) node [anchor=north west][inner sep=0.75pt]  [font=\scriptsize]  {$=$};
% Text Node
\draw (258,40.4) node [anchor=north west][inner sep=0.75pt]    {$\begin{cases}
 & \\
 & 
\end{cases}$};
% Text Node
\draw    (281,65) -- (299,65) -- (299,80) -- (281,80) -- cycle  ;
\draw (284,69.4) node [anchor=north west][inner sep=0.75pt]  [font=\tiny]  {$v_{2}$};
% Text Node
\draw    (281,39) -- (300,39) -- (300,54) -- (281,54) -- cycle  ;
\draw (282,43.4) node [anchor=north west][inner sep=0.75pt]  [font=\tiny]  {$v_{13}$};

\end{tikzpicture}}, \qquad \scalebox{0.8}{\tikzset{every picture/.style={line width=0.75pt}} %set default line width to 0.75pt        

\begin{tikzpicture}[x=0.75pt,y=0.75pt,yscale=-1,xscale=1]
%uncomment if require: \path (0,300); %set diagram left start at 0, and has height of 300

%Shape: Output [id:dp9117384275761491] 
\draw  [fill={rgb, 255:red, 0; green, 0; blue, 0 }  ,fill opacity=1 ] (210,75) .. controls (210,77.76) and (207.76,80) .. (205,80) .. controls (202.24,80) and (200,77.76) .. (200,75) .. controls (200,72.24) and (202.24,70) .. (205,70) .. controls (207.76,70) and (210,72.24) .. (210,75) -- cycle (215,75) -- (210,75) (195,75) -- (200,75) ;
%Shape: Output [id:dp11220543471493882] 
\draw   (190,75) .. controls (190,77.76) and (187.76,80) .. (185,80) .. controls (182.24,80) and (180,77.76) .. (180,75) .. controls (180,72.24) and (182.24,70) .. (185,70) .. controls (187.76,70) and (190,72.24) .. (190,75) -- cycle (195,75) -- (190,75) (175,75) -- (180,75) ;
%Shape: Output [id:dp9518278867258994] 
\draw  [fill={rgb, 255:red, 0; green, 0; blue, 0 }  ,fill opacity=1 ] (210,45) .. controls (210,47.76) and (207.76,50) .. (205,50) .. controls (202.24,50) and (200,47.76) .. (200,45) .. controls (200,42.24) and (202.24,40) .. (205,40) .. controls (207.76,40) and (210,42.24) .. (210,45) -- cycle (215,45) -- (210,45) (195,45) -- (200,45) ;
%Shape: Output [id:dp8391403228780596] 
\draw   (190,45.04) .. controls (190,47.8) and (187.76,50.04) .. (185,50.04) .. controls (182.24,50.04) and (180,47.8) .. (180,45.04) .. controls (180,42.28) and (182.24,40.04) .. (185,40.04) .. controls (187.76,40.04) and (190,42.28) .. (190,45.04) -- cycle (195,45.04) -- (190,45.04) (175,45.04) -- (180,45.04) ;
%Shape: Output [id:dp005166070611456908] 
\draw   (205,55) .. controls (207.76,55) and (210,57.24) .. (210,60) .. controls (210,62.76) and (207.76,65) .. (205,65) .. controls (202.24,65) and (200,62.76) .. (200,60) .. controls (200,57.24) and (202.24,55) .. (205,55) -- cycle (205,50) -- (205,55) (205,70) -- (205,65) ;
%Straight Lines [id:da7395741244975451] 
\draw    (270,73) -- (280,73) ;
%Straight Lines [id:da701705105436351] 
\draw    (270,47) -- (280,47) ;

% Text Node
\draw (202,32.4) node [anchor=north west][inner sep=0.75pt]  [font=\tiny]  {$2$};
% Text Node
\draw (202,82.4) node [anchor=north west][inner sep=0.75pt]  [font=\tiny]  {$2$};
% Text Node
\draw (192,57.4) node [anchor=north west][inner sep=0.75pt]  [font=\tiny]  {$4$};
% Text Node
\draw    (215,38) -- (234,38) -- (234,53) -- (215,53) -- cycle  ;
\draw (216,42.4) node [anchor=north west][inner sep=0.75pt]  [font=\tiny]  {$v_{13}$};
% Text Node
\draw    (215,67) -- (234,67) -- (234,82) -- (215,82) -- cycle  ;
\draw (216,71.4) node [anchor=north west][inner sep=0.75pt]  [font=\tiny]  {$v_{13}$};
% Text Node
\draw (238,53.4) node [anchor=north west][inner sep=0.75pt]  [font=\scriptsize]  {$=$};
% Text Node
\draw (257,40.4) node [anchor=north west][inner sep=0.75pt]    {$\begin{cases}
 & \\
 & 
\end{cases}$};
% Text Node
\draw    (280,65) -- (299,65) -- (299,80) -- (280,80) -- cycle  ;
\draw (281,69.4) node [anchor=north west][inner sep=0.75pt]  [font=\tiny]  {$v_{13}$};
% Text Node
\draw    (280,39) -- (299,39) -- (299,54) -- (280,54) -- cycle  ;
\draw (281,43.4) node [anchor=north west][inner sep=0.75pt]  [font=\tiny]  {$v_{13}$};

\end{tikzpicture}}
\end{align}
for $M$ as in Eq.~\eqref{m_matrix} and $\vot, \vt, \votm$ as in Eq.~\eqref{123vectors} we show that the contributions of $X_i, Y_i, Z_i$ to the variance are the same up to a constant factor:
\begin{lemma}
\thlabel{th_xyz_var}
Let $\sigma_i = I^{\otimes i-1}\otimes \sigma \otimes I^{\otimes n-i}$ where $\sigma \in \{X,Y,Z\}$ is a Pauli matrix.
Then
\begin{align}
 \Var[\partial_{1,1}\langle X_i\rangle] = \Var[\partial_{1,1}\langle Z_i\rangle] = c \Var[\partial_{1,1}\langle Y_i\rangle]
\end{align}
for some constant $c$.
\end{lemma}
\begin{proof}
For the first equality, note that by Eq.~\eqref{zx_ham} both observables $X_i$ and $Z_i$ yield $u_i = 2 \vt$ and $u_{i' \neq i} = 2 \vot$ so that the contraction in Eq.~\eqref{mps_var} is the same in both cases.
For the second equality, observable $Y_i$ yields $u_i = 2 \votm$ and $u_{i' \neq i} = 2 \vot$ and Eqs.~\eqref{zx_identities1}, \eqref{zx_identities2} imply
\begin{align}
 \scalebox{0.9}{\tikzset{every picture/.style={line width=0.75pt}} %set default line width to 0.75pt        

\begin{tikzpicture}[x=0.75pt,y=0.75pt,yscale=-1,xscale=1]
%uncomment if require: \path (0,300); %set diagram left start at 0, and has height of 300

%Shape: Output [id:dp7345590753651898] 
\draw  [fill={rgb, 255:red, 0; green, 0; blue, 0 }  ,fill opacity=1 ] (89,95) .. controls (89,97.76) and (86.76,100) .. (84,100) .. controls (81.24,100) and (79,97.76) .. (79,95) .. controls (79,92.24) and (81.24,90) .. (84,90) .. controls (86.76,90) and (89,92.24) .. (89,95) -- cycle (94,95) -- (89,95) (74,95) -- (79,95) ;
%Shape: Output [id:dp6930622212927964] 
\draw   (69,95) .. controls (69,97.76) and (66.76,100) .. (64,100) .. controls (61.24,100) and (59,97.76) .. (59,95) .. controls (59,92.24) and (61.24,90) .. (64,90) .. controls (66.76,90) and (69,92.24) .. (69,95) -- cycle (74,95) -- (69,95) (54,95) -- (59,95) ;
%Shape: Output [id:dp5395507251367864] 
\draw  [fill={rgb, 255:red, 0; green, 0; blue, 0 }  ,fill opacity=1 ] (89,65) .. controls (89,67.76) and (86.76,70) .. (84,70) .. controls (81.24,70) and (79,67.76) .. (79,65) .. controls (79,62.24) and (81.24,60) .. (84,60) .. controls (86.76,60) and (89,62.24) .. (89,65) -- cycle (94,65) -- (89,65) (74,65) -- (79,65) ;
%Shape: Output [id:dp3463085162598585] 
\draw   (69,65.04) .. controls (69,67.8) and (66.76,70.04) .. (64,70.04) .. controls (61.24,70.04) and (59,67.8) .. (59,65.04) .. controls (59,62.28) and (61.24,60.04) .. (64,60.04) .. controls (66.76,60.04) and (69,62.28) .. (69,65.04) -- cycle (74,65.04) -- (69,65.04) (54,65.04) -- (59,65.04) ;
%Shape: Output [id:dp9438214465743808] 
\draw   (84,75) .. controls (86.76,75) and (89,77.24) .. (89,80) .. controls (89,82.76) and (86.76,85) .. (84,85) .. controls (81.24,85) and (79,82.76) .. (79,80) .. controls (79,77.24) and (81.24,75) .. (84,75) -- cycle (84,70) -- (84,75) (84,90) -- (84,85) ;
%Straight Lines [id:da21344058053167814] 
\draw    (99,95) -- (109,95) ;
%Straight Lines [id:da8416784384548874] 
\draw    (153,95) -- (163,95) ;
%Straight Lines [id:da2059932933931936] 
\draw    (202,95) -- (212,95) ;
%Straight Lines [id:da6480935717267919] 
\draw    (278,93) -- (288,93) ;
%Straight Lines [id:da3586388967592853] 
\draw    (278,67) -- (288,67) ;
%Straight Lines [id:da36476655022576954] 
\draw    (351,93) -- (361,93) ;
%Straight Lines [id:da5528577249870568] 
\draw    (351,67) -- (361,67) ;
%Straight Lines [id:da3214799882365653] 
\draw    (430,96) -- (440,96) ;
%Straight Lines [id:da16130177879793295] 
\draw    (476,95) -- (486,95) ;
%Straight Lines [id:da6012155874321912] 
\draw    (450,67) -- (460,67) ;

% Text Node
\draw (81,52.4) node [anchor=north west][inner sep=0.75pt]  [font=\tiny]  {$2$};
% Text Node
\draw (81,102.4) node [anchor=north west][inner sep=0.75pt]  [font=\tiny]  {$2$};
% Text Node
\draw (71,77.4) node [anchor=north west][inner sep=0.75pt]  [font=\tiny]  {$4$};
% Text Node
\draw    (94,58) -- (113,58) -- (113,73) -- (94,73) -- cycle  ;
\draw (93,62.4) node [anchor=north west][inner sep=0.75pt]  [font=\tiny]  {$v_{13}$};
% Text Node
\draw (93,84.4) node [anchor=north west][inner sep=0.75pt]    {$\{$};
% Text Node
\draw    (109,87) -- (148,87) -- (148,102) -- (109,102) -- cycle  ;
\draw (112,91.4) node [anchor=north west][inner sep=0.75pt]  [font=\tiny]  {$c_{13} v_{13}$};
% Text Node
\draw    (163,87) -- (190,87) -- (190,102) -- (163,102) -- cycle  ;
\draw (166,91.4) node [anchor=north west][inner sep=0.75pt]  [font=\tiny]  {$c_{2} v_{2}$};
% Text Node
\draw (145,88.4) node [anchor=north west][inner sep=0.75pt]  [font=\scriptsize]  {$+$};
% Text Node
\draw    (212,87) -- (245,87) -- (245,104) -- (212,104) -- cycle  ;
\draw (213,88) node [anchor=north west][inner sep=0.75pt]  [font=\tiny]  {$c_{13}^{-} v_{13}^{-}$};
% Text Node
\draw (189,88.4) node [anchor=north west][inner sep=0.75pt]  [font=\scriptsize]  {$+$};
% Text Node
\draw (244,86.4) node [anchor=north west][inner sep=0.75pt]    {$\}$};
% Text Node
\draw (249,73.4) node [anchor=north west][inner sep=0.75pt]  [font=\scriptsize]  {$=$};
% Text Node
\draw (265,60.4) node [anchor=north west][inner sep=0.75pt]    {$\begin{cases}
 & \\
 & 
\end{cases}$};
% Text Node
\draw    (288,85) -- (323,85) -- (323,100) -- (288,100) -- cycle  ;
\draw (291,89.4) node [anchor=north west][inner sep=0.75pt]  [font=\tiny]  {$c_{13} v_{13}$};
% Text Node
\draw    (288,59) -- (311,59) -- (311,74) -- (288,74) -- cycle  ;
\draw (291,63.4) node [anchor=north west][inner sep=0.75pt]  [font=\tiny]  {$v_{13}$};
% Text Node
\draw (338,60.4) node [anchor=north west][inner sep=0.75pt]    {$\begin{cases}
 & \\
 & 
\end{cases}$};
% Text Node
\draw    (361,85) -- (391,85) -- (391,102) -- (361,102) -- cycle  ;
\draw (362,87) node [anchor=north west][inner sep=0.75pt]  [font=\tiny]  {$c_{13}^{-} v_{2}$};
% Text Node
\draw    (361,59) -- (384,59) -- (384,74) -- (361,74) -- cycle  ;
\draw (364,63.4) node [anchor=north west][inner sep=0.75pt]  [font=\tiny]  {$v_{13}$};
% Text Node
\draw (408,72) node [anchor=north west][inner sep=0.75pt]  [font=\small]  {$\frac{1}{2}$};
% Text Node
\draw (418,60.4) node [anchor=north west][inner sep=0.75pt]    {$\begin{cases}
 & \\
 & 
\end{cases}$};
% Text Node
\draw    (440,88) -- (463,88) -- (463,103) -- (440,103) -- cycle  ;
\draw (441,92.4) node [anchor=north west][inner sep=0.75pt]  [font=\tiny]  {$c_{2} v_{2}$};
% Text Node
\draw    (486,87) -- (512,87) -- (512,104) -- (486,104) -- cycle  ;
\draw (487,88) node [anchor=north west][inner sep=0.75pt]  [font=\tiny]  {$c_{2} v_{13}^{-}$};
% Text Node
\draw    (460,59) -- (478,59) -- (478,74) -- (460,74) -- cycle  ;
\draw (463,63.4) node [anchor=north west][inner sep=0.75pt]  [font=\tiny]  {$v_{2}$};
% Text Node
\draw (464,88.4) node [anchor=north west][inner sep=0.75pt]  [font=\scriptsize]  {$+$};
% Text Node
\draw (325,73.4) node [anchor=north west][inner sep=0.75pt]  [font=\scriptsize]  {$+$};
% Text Node
\draw (392,74.4) node [anchor=north west][inner sep=0.75pt]  [font=\scriptsize]  {$+$};
% Text Node
\draw (188.5,38) node  [font=\scriptsize] [align=left] {$(i-1)$-th register};
% Text Node
\draw (121,122) node [anchor=north west][inner sep=0.75pt]  [font=\scriptsize] [align=left] {$i$-th register after contraction\\of registers $ i' >i$};
% Connection
\draw    (157.59,48) -- (114.9,61.81) ;
\draw [shift={(113,62.43)}, rotate = 342.07] [color={rgb, 255:red, 0; green, 0; blue, 0 }  ][line width=0.75]    (10.93,-3.29) .. controls (6.95,-1.4) and (3.31,-0.3) .. (0,0) .. controls (3.31,0.3) and (6.95,1.4) .. (10.93,3.29)   ;
% Connection
\draw    (182.24,118) -- (177.59,103.9) ;
\draw [shift={(176.97,102)}, rotate = 71.77] [color={rgb, 255:red, 0; green, 0; blue, 0 }  ][line width=0.75]    (10.93,-3.29) .. controls (6.95,-1.4) and (3.31,-0.3) .. (0,0) .. controls (3.31,0.3) and (6.95,1.4) .. (10.93,3.29)   ;

\end{tikzpicture}}
\end{align}
where the vector $c_{13} \vot + c_{2} \vt + c_{13}^- \votm$ for non-negative constants $c_{13}, c_{13}^-, c_{2}$ in the $i$-th register comes from contracting all registers $i' > i$ in the tensor network in~\eqref{mps_var}. 
In particular, note that the right-hand side above no longer carries a $\votm$ term on the $(i-1)$-th register.
Additionally the two terms on the right side leading with a $\vot$ on the top register do not contribute to the variance as they will eventually be discarded by the projection.
Hence the tensor network is fully determined by the third term on the right-hand side and therefore equivalent to the one corresponding to the observables $X_i$ and $Z_i$, up to the constant factor $c_2$ accrued from contracting the registers $i' > i$\footnote{
In fact the registers $i' > i$ which have a $\vot$ at the end and are not directly (nearest-neighbour) connected to a $\vt$ or $\votm$ contract to the identity.
}.
\end{proof}

This Lemma implies that it suffices to consider the $1$-local observable $X_i$ to probe the behaviour of the variance for general $1$-local operators.
Also, this Lemma trivially generalises to the qTTN and qMERA circuits and, therefore, henceforth we focus solely on observables $X_i$.

\begin{theorem}
\thlabel{th_mps2}
Let $\langle X_i\rangle_{\text{qMPS}}$ be the cost function associated with the observable $X_i$ and consider the qMPS ansatz for $N$ qubits defined in Eq.~\eqref{mps_ans}, then:
\begin{align}
 \Var[\partial_{1,1}\CostXi{}_{\text{qMPS}}] =
  \begin{cases}
   \frac{1}{4} \cdot \big(\frac{3}{8}\big)^{N-1} & \text{if } i = N,\\
   11 \cdot \big(\frac{1}{8}\big)^2 \cdot \big(\frac{3}{8}\big)^{i-1} & \text{if } 1 < i < N,\\
   \frac{11}{8^2} & \text{if } i = 1,
  \end{cases}
\end{align}
where $\partial_{1,1}\CostXi{}_{\text{qMPS}}$ refers to the gradient w.r.t.\ the $1$-st parameter in the $1$-st qubit register.
\end{theorem}
\begin{proof}
$\Var[\partial_{1,1}\CostXi{}]$ can be found for the three separate cases by contracting the tensor network in Eq.~\eqref{mps_var} with $u_i = \vt$ and $u_{i' \neq i} = \vot$.
Given Eqs.~\eqref{zx_identities1}, \eqref{zx_identities2} this is a straightforward calculation from which we also derive the useful identity
\begin{align}\label{zx_identities3}
 \scalebox{0.9}{\tikzset{every picture/.style={line width=0.75pt}} %set default line width to 0.75pt        

\begin{tikzpicture}[x=0.75pt,y=0.75pt,yscale=-1,xscale=1]
%uncomment if require: \path (0,300); %set diagram left start at 0, and has height of 300

%Shape: Output [id:dp07818236478674745] 
\draw  [fill={rgb, 255:red, 0; green, 0; blue, 0 }  ,fill opacity=1 ] (125,95) .. controls (125,97.76) and (122.76,100) .. (120,100) .. controls (117.24,100) and (115,97.76) .. (115,95) .. controls (115,92.24) and (117.24,90) .. (120,90) .. controls (122.76,90) and (125,92.24) .. (125,95) -- cycle (130,95) -- (125,95) (110,95) -- (115,95) ;
%Shape: Output [id:dp09759406555231265] 
\draw   (105,95) .. controls (105,97.76) and (102.76,100) .. (100,100) .. controls (97.24,100) and (95,97.76) .. (95,95) .. controls (95,92.24) and (97.24,90) .. (100,90) .. controls (102.76,90) and (105,92.24) .. (105,95) -- cycle (110,95) -- (105,95) (90,95) -- (95,95) ;
%Shape: Output [id:dp9778463797047954] 
\draw  [fill={rgb, 255:red, 0; green, 0; blue, 0 }  ,fill opacity=1 ] (125,65) .. controls (125,67.76) and (122.76,70) .. (120,70) .. controls (117.24,70) and (115,67.76) .. (115,65) .. controls (115,62.24) and (117.24,60) .. (120,60) .. controls (122.76,60) and (125,62.24) .. (125,65) -- cycle (130,65) -- (125,65) (110,65) -- (115,65) ;
%Shape: Output [id:dp7729399387914844] 
\draw   (105,65.04) .. controls (105,67.8) and (102.76,70.04) .. (100,70.04) .. controls (97.24,70.04) and (95,67.8) .. (95,65.04) .. controls (95,62.28) and (97.24,60.04) .. (100,60.04) .. controls (102.76,60.04) and (105,62.28) .. (105,65.04) -- cycle (110,65.04) -- (105,65.04) (90,65.04) -- (95,65.04) ;
%Shape: Output [id:dp15502835444323448] 
\draw   (120,75) .. controls (122.76,75) and (125,77.24) .. (125,80) .. controls (125,82.76) and (122.76,85) .. (120,85) .. controls (117.24,85) and (115,82.76) .. (115,80) .. controls (115,77.24) and (117.24,75) .. (120,75) -- cycle (120,70) -- (120,75) (120,90) -- (120,85) ;
%Straight Lines [id:da6791607512525375] 
\draw    (135,95) -- (145,95) ;
%Straight Lines [id:da5272707997536914] 
\draw    (195,95) -- (205,95) ;
%Straight Lines [id:da5677007112872321] 
\draw    (250,95) -- (260,95) ;
%Straight Lines [id:da7175041517292748] 
\draw    (345,67) -- (355,67) ;
%Straight Lines [id:da24808539939440233] 
\draw    (429,67) -- (439,67) ;
%Shape: Output [id:dp13038929401326094] 
\draw  [fill={rgb, 255:red, 0; green, 0; blue, 0 }  ,fill opacity=1 ] (84.99,94.92) .. controls (84.99,97.68) and (82.75,99.92) .. (79.99,99.92) .. controls (77.23,99.92) and (74.99,97.68) .. (74.99,94.92) .. controls (74.99,92.16) and (77.23,89.92) .. (79.99,89.92) .. controls (82.75,89.92) and (84.99,92.16) .. (84.99,94.92) -- cycle (89.99,94.92) -- (84.99,94.92) (69.99,94.92) -- (74.99,94.92) ;
%Shape: Output [id:dp09813007312855859] 
\draw   (64.99,94.96) .. controls (64.99,97.72) and (62.75,99.96) .. (59.99,99.96) .. controls (57.23,99.96) and (54.99,97.72) .. (54.99,94.96) .. controls (54.99,92.2) and (57.23,89.96) .. (59.99,89.96) .. controls (62.75,89.96) and (64.99,92.2) .. (64.99,94.96) -- cycle (69.99,94.96) -- (64.99,94.96) (49.99,94.96) -- (54.99,94.96) ;
%Shape: Circle [id:dp2902664477892849] 
\draw  [fill={rgb, 255:red, 0; green, 0; blue, 0 }  ,fill opacity=1 ] (45,94.9) .. controls (45,92.42) and (42.99,90.4) .. (40.5,90.4) .. controls (38.01,90.4) and (36,92.42) .. (36,94.9) .. controls (36,97.39) and (38.01,99.4) .. (40.5,99.4) .. controls (42.99,99.4) and (45,97.39) .. (45,94.9) -- cycle ;
%Straight Lines [id:da29309428872142007] 
\draw    (45,94.9) -- (49.99,94.96) ;

% Text Node
\draw (117,52.4) node [anchor=north west][inner sep=0.75pt]  [font=\tiny]  {$2$};
% Text Node
\draw (117,102.4) node [anchor=north west][inner sep=0.75pt]  [font=\tiny]  {$2$};
% Text Node
\draw (107,77.4) node [anchor=north west][inner sep=0.75pt]  [font=\tiny]  {$4$};
% Text Node
\draw    (130,58) -- (149,58) -- (149,73) -- (130,73) -- cycle  ;
\draw (133,62.4) node [anchor=north west][inner sep=0.75pt]  [font=\tiny]  {$v_{13}$};
% Text Node
\draw (129,84.4) node [anchor=north west][inner sep=0.75pt]    {$\{$};
% Text Node
\draw    (145,87) -- (176,87) -- (176,102) -- (145,102) -- cycle  ;
\draw (146,91.4) node [anchor=north west][inner sep=0.75pt]  [font=\tiny]  {$c_{13} v_{13}$};
% Text Node
\draw    (205,87) -- (229,87) -- (229,102) -- (205,102) -- cycle  ;
\draw (206,91.4) node [anchor=north west][inner sep=0.75pt]  [font=\tiny]  {$c_{2} v_{2}$};
% Text Node
\draw (180,88.4) node [anchor=north west][inner sep=0.75pt]  [font=\scriptsize]  {$+$};
% Text Node
\draw    (260,87) -- (291,87) -- (291,104) -- (260,104) -- cycle  ;
\draw (261,89) node [anchor=north west][inner sep=0.75pt]  [font=\tiny]  {$c_{13}^{-} v_{13}^{-}$};
% Text Node
\draw (234,88.4) node [anchor=north west][inner sep=0.75pt]  [font=\scriptsize]  {$+$};
% Text Node
\draw (290,86.4) node [anchor=north west][inner sep=0.75pt]    {$\}$};
% Text Node
\draw (305,73.4) node [anchor=north west][inner sep=0.75pt]  [font=\scriptsize]  {$=$};
% Text Node
\draw (327,60.4) node [anchor=north west][inner sep=0.75pt]    {$\begin{cases}
 & \\
 & 
\end{cases}$};
% Text Node
\draw    (355,59) -- (374,59) -- (374,74) -- (355,74) -- cycle  ;
\draw (356,63.4) node [anchor=north west][inner sep=0.75pt]  [font=\tiny]  {$v_{13}$};
% Text Node
\draw (412,60.4) node [anchor=north west][inner sep=0.75pt]    {$\begin{cases}
 & \\
 & 
\end{cases}$};
% Text Node
\draw    (439,59) -- (457,59) -- (457,74) -- (439,74) -- cycle  ;
\draw (442,63.4) node [anchor=north west][inner sep=0.75pt]  [font=\tiny]  {$v_{2}$};
% Text Node
\draw (395,73.4) node [anchor=north west][inner sep=0.75pt]  [font=\scriptsize]  {$+$};
% Text Node
\draw (26,86.4) node [anchor=north west][inner sep=0.75pt]  [font=\tiny]  {$\frac{1}{4}$};
% Text Node
\draw (37,82.4) node [anchor=north west][inner sep=0.75pt]  [font=\tiny]  {$2$};
% Text Node
\draw (77,82.4) node [anchor=north west][inner sep=0.75pt]  [font=\tiny]  {$2$};
% Text Node
\draw (363.5,89) node  [font=\tiny]  {$c_{13} +\frac{c_{13}^-}{4}$};
% Text Node
\draw (447.5,90) node  [font=\tiny]  {$\frac{3c_{2}}{8}$};

\end{tikzpicture}}
\end{align}
which repeats on every register in Eq.~\eqref{mps_var}, for non-negative constants $c_{13}, c_{13}^-, c_{2}$. 
\end{proof}

Computing the gradient variance for a general parameter indexed by $(j,k)$ can be done analogously by moving the projection $P_2$ in Eq.~\eqref{mps_var} to the copy tensor at position $(j,k)$.
The calculation can be simplified by, first, identifying the cases in which the triple index $(i,j,k)$ gives $\Var[\partial_{j,k}\CostXi{}_{\text{qMPS}}] = 0$. 
Figure~\ref{fig:mps_causal_cone} illustrates the causal cone corresponding to observable $X_i$ in a qMPS circuit.
We observe that in the qMPS the triple index $(i,j,k)$ for which $\Var[\partial_{j,k}\CostXi{}_{\text{qMPS}}] = 0$ satisfies
\begin{align}
 \Var[\partial_{j,k}\CostXi{}_{\text{qMPS}}] = 0 \quad\text{if } \begin{cases}
  j > i+1\; \forall k,\\
  j = i+1 \text{ and } k > 2,\\
  j < i \text{ and } k > 4 \; (k > 2 \text{ for } j = 1).
 \end{cases}
\end{align}

\begin{figure}[h]
\centering
\scalebox{0.8}{\input{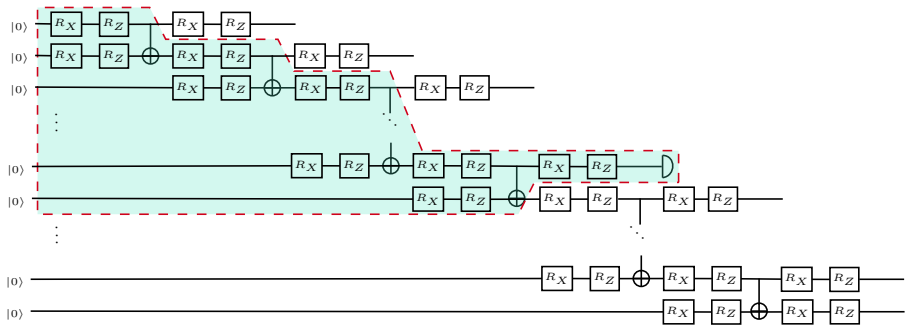}}
\caption{\label{fig:mps_causal_cone}
Causal cone for a $1$-local observable.
}
\end{figure}

The definition for the causal cone of a PQC can be extended analogously to apply to the variance tensor networks of the form of Eq.~\eqref{var_as_tensor_net} ---for example Eq.~\eqref{mps_var}.

\begin{theorem}
\thlabel{th_mps_varjk}
Let $\langle X_i\rangle_{\text{qMPS}}$ be the cost function associated with the observable $X_i$ and consider the qMPS ansatz defined in Eq.~\eqref{mps_ans}, then:
\begin{align}
 \Var[\partial_{j,1}\langle X_N\rangle_{\text{qMPS}}] &= \begin{cases}
  \frac{1}{4} \cdot \big(\frac{3}{8}\big)^{N-1} &\text{if } j < N,\\
  \frac{1}{4}\Big(1 + \big(\frac{3}{8}\big)^{N-1}\Big) &\text{if } j = N,
 \end{cases}\\
 \Var[\partial_{j,1}\CostXi{}_{\text{qMPS}}] &= \begin{cases}
  11 \cdot \big(\frac{1}{8}\big)^2\big(\frac{3}{8}\big)^{i-1} &\text{if } j < i,\\
  3 \cdot \big(\frac{1}{8}\big)^2\Big(1+\frac{11}{8} \cdot \big(\frac{3}{8}\big)^{i-2}\Big) &\text{if } j = i,\\
  3 \cdot \big(\frac{1}{8}\big)^2\Big(1+\big(\frac{3}{8}\big)^{i-1}\Big) &\text{if } j = i+1,
 \end{cases}
\end{align}
where $\partial_{j,1}\CostXi{}_{\text{qMPS}}$ refers to the gradient w.r.t.\ the $1$-st parameter in the $j$-th qubit register.
\end{theorem}
\begin{proof}
This is a straightforward contraction of the tensor network in Eq.~\eqref{mps_var} but with the projection $P_2$ replacing the copy tensor indexed by $(j,k)$ and using Eqs.~\eqref{zx_identities1}, \eqref{zx_identities2}, \eqref{zx_identities3}.
\end{proof}

\begin{remark}\thlabel{remark_varjk}
It is clear from this Theorem that $\Var[\partial_{j,1}\CostXi{}_{\text{qMPS}}] \geq \Var[\partial_{1,1}\CostXi{}_{\text{qMPS}}]$ for all $i=1,\ldots,N$ and all $j=1,\ldots,i+1$ but, in fact, this result also generalises to all $(j,k)$ for which $\Var[\partial_{j,k}\CostXi{}_{\text{qMPS}}] \neq 0$. 
Indeed, each step in the contraction of the tensor network in Eq.~\eqref{mps_var} increases the coefficients of the vectors $\vot, \vt$ and $\votm$ monotonically and, therefore, the earlier the projection $P_2$ is placed, the larger these coefficients become after the contributions of $\vot$ and $\votm$ are removed by $P_2$.
This argument applies analogously to qTTN and qMERA.
\end{remark}

Now we consider $k$-local operators of the form $X_I := X_{i_1}\otimes\cdots\otimes X_{i_k}$ for $I=\{i_1,\ldots,i_k\}$ (w.l.o.g.\ $i_1 < i_2 < \ldots < i_k$) and use techniques and results from this Appendix to justify \thref{prop_mps}.
The proof of \thref{th_mps2} and the causal cone structure in Fig.~\ref{fig:mps_causal_cone} suggest that $\Var[\partial_{1,1}\langle X_I\rangle_{\text{qMPS}}]$ vanishes exponentially with $i_k$.
Our intuition is that barren plateaus appear when the causal cone of an observable includes a large number of qubits ($\approx N$) and we know that the causal cone relating to $X_I$ contains at most $i_k + 1$ qubits for the qMPS ansatz in Fig.~\ref{fig:mps_circ}.

\begin{theorem}
\thlabel{th_mps_xx}
Let $\langle X_i X_{i+1}\rangle_{\text{qMPS}}$ be the cost function associated with the observable $X_i X_{i+1}$ and consider the qMPS ansatz of Eq.~\eqref{mps_ans}, then:
\begin{align}
 \Var[\partial_{1,1}\langle X_{i} X_{i-1}\rangle_{\text{qMPS}}] = c_i \Big(\frac{3}{8}\Big)^i,
\end{align}
where
\begin{align}
 c_i = \begin{cases}
  \frac{1}{4} \cdot \Big(\big(\frac{3}{8}\big)^2+\frac{13}{16}\Big) & \text{if } i=1,\\
  \frac{1}{4} \cdot \Big(\frac{37}{2\cdot8^2}+\frac{3}{16})\Big) & \text{if } 1 < i < N,\\
  \frac{37}{3 \cdot 8^2} & \text{if } i = N-1,
 \end{cases}
\end{align}
where $\partial_{1,1}\CostXi{}_{\text{qMPS}}$ refers to the gradient w.r.t.\ the $1$-st parameter in the $1$-st qubit register.
\end{theorem}
\begin{proof}
$\Var[\partial_{1,1}\langle X_i X_{i+1}\rangle_{\text{qMPS}}]$ can be found for the three separate cases by contracting the tensor network in Eq.~\eqref{mps_var} with $\mathbf{u_i} = \mathbf{u_{i+1}} = \vt$ and $u_{i' \neq i, i+1} = \vot$ using Eqs.~\eqref{zx_identities1}, \eqref{zx_identities2}, \eqref{zx_identities3} in addition to:
\begin{align}
\label{zx_identities4}
 \scalebox{0.8}{\tikzset{every picture/.style={line width=0.75pt}} %set default line width to 0.75pt        

\begin{tikzpicture}[x=0.75pt,y=0.75pt,yscale=-1,xscale=1]
%uncomment if require: \path (0,300); %set diagram left start at 0, and has height of 300

%Shape: Output [id:dp08946483243242165] 
\draw  [fill={rgb, 255:red, 0; green, 0; blue, 0 }  ,fill opacity=1 ] (210,78) .. controls (210,80.76) and (207.76,83) .. (205,83) .. controls (202.24,83) and (200,80.76) .. (200,78) .. controls (200,75.24) and (202.24,73) .. (205,73) .. controls (207.76,73) and (210,75.24) .. (210,78) -- cycle (215,78) -- (210,78) (195,78) -- (200,78) ;
%Shape: Output [id:dp42519674167170285] 
\draw   (190,78) .. controls (190,80.76) and (187.76,83) .. (185,83) .. controls (182.24,83) and (180,80.76) .. (180,78) .. controls (180,75.24) and (182.24,73) .. (185,73) .. controls (187.76,73) and (190,75.24) .. (190,78) -- cycle (195,78) -- (190,78) (175,78) -- (180,78) ;
%Shape: Output [id:dp41394782987272816] 
\draw  [fill={rgb, 255:red, 0; green, 0; blue, 0 }  ,fill opacity=1 ] (210,48) .. controls (210,50.76) and (207.76,53) .. (205,53) .. controls (202.24,53) and (200,50.76) .. (200,48) .. controls (200,45.24) and (202.24,43) .. (205,43) .. controls (207.76,43) and (210,45.24) .. (210,48) -- cycle (215,48) -- (210,48) (195,48) -- (200,48) ;
%Shape: Output [id:dp921418539604665] 
\draw   (190,48.04) .. controls (190,50.8) and (187.76,53.04) .. (185,53.04) .. controls (182.24,53.04) and (180,50.8) .. (180,48.04) .. controls (180,45.28) and (182.24,43.04) .. (185,43.04) .. controls (187.76,43.04) and (190,45.28) .. (190,48.04) -- cycle (195,48.04) -- (190,48.04) (175,48.04) -- (180,48.04) ;
%Shape: Output [id:dp269178006308616] 
\draw   (205,58) .. controls (207.76,58) and (210,60.24) .. (210,63) .. controls (210,65.76) and (207.76,68) .. (205,68) .. controls (202.24,68) and (200,65.76) .. (200,63) .. controls (200,60.24) and (202.24,58) .. (205,58) -- cycle (205,53) -- (205,58) (205,73) -- (205,68) ;
%Straight Lines [id:da5265251592513889] 
\draw    (280,76) -- (290,76) ;
%Straight Lines [id:da6853429588086342] 
\draw    (326,75) -- (336,75) ;
%Straight Lines [id:da5612926428141951] 
\draw    (280,48) -- (290,48) ;
%Straight Lines [id:da9741577880964105] 
\draw    (327.78,47.55) -- (337.78,47.55) ;

% Text Node
\draw (202,35.4) node [anchor=north west][inner sep=0.75pt]  [font=\tiny]  {$2$};
% Text Node
\draw (202,85.4) node [anchor=north west][inner sep=0.75pt]  [font=\tiny]  {$2$};
% Text Node
\draw (192,60.4) node [anchor=north west][inner sep=0.75pt]  [font=\tiny]  {$4$};
% Text Node
\draw    (215.47,41.09) -- (233.47,41.09) -- (233.47,56.09) -- (215.47,56.09) -- cycle  ;
\draw (224.47,48.59) node  [font=\tiny]  {$v_{2}$};
% Text Node
\draw    (215.22,70.09) -- (233.22,70.09) -- (233.22,85.09) -- (215.22,85.09) -- cycle  ;
\draw (224.22,77.59) node  [font=\tiny]  {$v_{2}$};
% Text Node
\draw (238,54.4) node [anchor=north west][inner sep=0.75pt]  [font=\scriptsize]  {$=$};
% Text Node
\draw (258,52) node [anchor=north west][inner sep=0.75pt]  [font=\footnotesize]  {$\frac{1}{4}$};
% Text Node
\draw (281.9,59.9) node    {$\begin{cases}
 \\\\
 \\
\end{cases}$};
% Text Node
\draw    (290.22,68.09) -- (308.22,68.09) -- (308.22,83.09) -- (290.22,83.09) -- cycle  ;
\draw (299.22,75.59) node  [font=\tiny]  {$v_{2}$};
% Text Node
\draw    (335.97,66.95) -- (354.97,66.95) -- (354.97,83.95) -- (335.97,83.95) -- cycle  ;
\draw (345.47,75.45) node  [font=\tiny]  {$v_{13}^{-}$};
% Text Node
\draw (313,69.4) node [anchor=north west][inner sep=0.75pt]  [font=\scriptsize]  {$+$};
% Text Node
\draw    (290,40.65) -- (308,40.65) -- (308,55.65) -- (290,55.65) -- cycle  ;
\draw (299,48.15) node  [font=\tiny]  {$v_{2}$};
% Text Node
\draw    (337.74,39.5) -- (356.74,39.5) -- (356.74,56.5) -- (337.74,56.5) -- cycle  ;
\draw (347.24,48) node  [font=\tiny]  {$v_{13}^{-}$};
% Text Node
\draw (314.78,41.95) node [anchor=north west][inner sep=0.75pt]  [font=\scriptsize]  {$+$};

\end{tikzpicture}}, \qquad \scalebox{0.8}{\tikzset{every picture/.style={line width=0.75pt}} %set default line width to 0.75pt        

\begin{tikzpicture}[x=0.75pt,y=0.75pt,yscale=-1,xscale=1]
%uncomment if require: \path (0,300); %set diagram left start at 0, and has height of 300

%Shape: Output [id:dp5055194556444669] 
\draw  [fill={rgb, 255:red, 0; green, 0; blue, 0 }  ,fill opacity=1 ] (210,76) .. controls (210,78.76) and (207.76,81) .. (205,81) .. controls (202.24,81) and (200,78.76) .. (200,76) .. controls (200,73.24) and (202.24,71) .. (205,71) .. controls (207.76,71) and (210,73.24) .. (210,76) -- cycle (215,76) -- (210,76) (195,76) -- (200,76) ;
%Shape: Output [id:dp7480651110764933] 
\draw   (190,76) .. controls (190,78.76) and (187.76,81) .. (185,81) .. controls (182.24,81) and (180,78.76) .. (180,76) .. controls (180,73.24) and (182.24,71) .. (185,71) .. controls (187.76,71) and (190,73.24) .. (190,76) -- cycle (195,76) -- (190,76) (175,76) -- (180,76) ;
%Shape: Output [id:dp17208887026677333] 
\draw  [fill={rgb, 255:red, 0; green, 0; blue, 0 }  ,fill opacity=1 ] (210,46) .. controls (210,48.76) and (207.76,51) .. (205,51) .. controls (202.24,51) and (200,48.76) .. (200,46) .. controls (200,43.24) and (202.24,41) .. (205,41) .. controls (207.76,41) and (210,43.24) .. (210,46) -- cycle (215,46) -- (210,46) (195,46) -- (200,46) ;
%Shape: Output [id:dp7517697010126108] 
\draw   (190,46.04) .. controls (190,48.8) and (187.76,51.04) .. (185,51.04) .. controls (182.24,51.04) and (180,48.8) .. (180,46.04) .. controls (180,43.28) and (182.24,41.04) .. (185,41.04) .. controls (187.76,41.04) and (190,43.28) .. (190,46.04) -- cycle (195,46.04) -- (190,46.04) (175,46.04) -- (180,46.04) ;
%Shape: Output [id:dp0358811492661264] 
\draw   (205,56) .. controls (207.76,56) and (210,58.24) .. (210,61) .. controls (210,63.76) and (207.76,66) .. (205,66) .. controls (202.24,66) and (200,63.76) .. (200,61) .. controls (200,58.24) and (202.24,56) .. (205,56) -- cycle (205,51) -- (205,56) (205,71) -- (205,66) ;
%Straight Lines [id:da9202356500473989] 
\draw    (280,76) -- (290,76) ;
%Straight Lines [id:da9989251033456328] 
\draw    (324,75) -- (334,75) ;
%Straight Lines [id:da6623927046303291] 
\draw    (300,47) -- (310,47) ;

% Text Node
\draw (202,33.4) node [anchor=north west][inner sep=0.75pt]  [font=\tiny]  {$2$};
% Text Node
\draw (202,83.4) node [anchor=north west][inner sep=0.75pt]  [font=\tiny]  {$2$};
% Text Node
\draw (192,58.4) node [anchor=north west][inner sep=0.75pt]  [font=\tiny]  {$4$};
% Text Node
\draw    (214.97,38.09) -- (233.97,38.09) -- (233.97,55.09) -- (214.97,55.09) -- cycle  ;
\draw (224.47,46.59) node  [font=\tiny]  {$v_{13}^{-}$};
% Text Node
\draw    (215.22,68.09) -- (233.22,68.09) -- (233.22,83.09) -- (215.22,83.09) -- cycle  ;
\draw (224.22,75.59) node  [font=\tiny]  {$v_{2}$};
% Text Node
\draw (238,53.4) node [anchor=north west][inner sep=0.75pt]  [font=\scriptsize]  {$=$};
% Text Node
\draw (258,52) node [anchor=north west][inner sep=0.75pt]  [font=\footnotesize]  {$\frac{1}{2}$};
% Text Node
\draw (283.54,59.9) node    {$\begin{cases}
 \\\\
 \\
\end{cases}$};
% Text Node
\draw    (290.22,68.09) -- (308.22,68.09) -- (308.22,83.09) -- (290.22,83.09) -- cycle  ;
\draw (299.22,75.59) node  [font=\tiny]  {$v_{2}$};
% Text Node
\draw    (333.97,66.95) -- (352.97,66.95) -- (352.97,83.95) -- (333.97,83.95) -- cycle  ;
\draw (343.47,75.45) node  [font=\tiny]  {$v_{13}^{-}$};
% Text Node
\draw    (310.22,39.09) -- (328.22,39.09) -- (328.22,54.09) -- (310.22,54.09) -- cycle  ;
\draw (319.22,46.59) node  [font=\tiny]  {$v_{2}$};
% Text Node
\draw (311,68.4) node [anchor=north west][inner sep=0.75pt]  [font=\scriptsize]  {$+$};

\end{tikzpicture}}, \qquad \scalebox{0.8}{\tikzset{every picture/.style={line width=0.75pt}} %set default line width to 0.75pt        

\begin{tikzpicture}[x=0.75pt,y=0.75pt,yscale=-1,xscale=1]
%uncomment if require: \path (0,300); %set diagram left start at 0, and has height of 300

%Shape: Output [id:dp10115163790595161] 
\draw  [fill={rgb, 255:red, 0; green, 0; blue, 0 }  ,fill opacity=1 ] (210,76) .. controls (210,78.76) and (207.76,81) .. (205,81) .. controls (202.24,81) and (200,78.76) .. (200,76) .. controls (200,73.24) and (202.24,71) .. (205,71) .. controls (207.76,71) and (210,73.24) .. (210,76) -- cycle (215,76) -- (210,76) (195,76) -- (200,76) ;
%Shape: Output [id:dp9582320826829707] 
\draw   (190,76) .. controls (190,78.76) and (187.76,81) .. (185,81) .. controls (182.24,81) and (180,78.76) .. (180,76) .. controls (180,73.24) and (182.24,71) .. (185,71) .. controls (187.76,71) and (190,73.24) .. (190,76) -- cycle (195,76) -- (190,76) (175,76) -- (180,76) ;
%Shape: Output [id:dp29521783032523796] 
\draw  [fill={rgb, 255:red, 0; green, 0; blue, 0 }  ,fill opacity=1 ] (210,46) .. controls (210,48.76) and (207.76,51) .. (205,51) .. controls (202.24,51) and (200,48.76) .. (200,46) .. controls (200,43.24) and (202.24,41) .. (205,41) .. controls (207.76,41) and (210,43.24) .. (210,46) -- cycle (215,46) -- (210,46) (195,46) -- (200,46) ;
%Shape: Output [id:dp3656585082758519] 
\draw   (190,46.04) .. controls (190,48.8) and (187.76,51.04) .. (185,51.04) .. controls (182.24,51.04) and (180,48.8) .. (180,46.04) .. controls (180,43.28) and (182.24,41.04) .. (185,41.04) .. controls (187.76,41.04) and (190,43.28) .. (190,46.04) -- cycle (195,46.04) -- (190,46.04) (175,46.04) -- (180,46.04) ;
%Shape: Output [id:dp0309149110156548] 
\draw   (205,56) .. controls (207.76,56) and (210,58.24) .. (210,61) .. controls (210,63.76) and (207.76,66) .. (205,66) .. controls (202.24,66) and (200,63.76) .. (200,61) .. controls (200,58.24) and (202.24,56) .. (205,56) -- cycle (205,51) -- (205,56) (205,71) -- (205,66) ;
%Straight Lines [id:da4069475170042176] 
\draw    (280,76) -- (290,76) ;
%Straight Lines [id:da8809077783531192] 
\draw    (280,48) -- (290,48) ;

% Text Node
\draw (202,33.4) node [anchor=north west][inner sep=0.75pt]  [font=\tiny]  {$2$};
% Text Node
\draw (202,83.4) node [anchor=north west][inner sep=0.75pt]  [font=\tiny]  {$2$};
% Text Node
\draw (192,58.4) node [anchor=north west][inner sep=0.75pt]  [font=\tiny]  {$4$};
% Text Node
\draw    (214.97,38.09) -- (233.97,38.09) -- (233.97,55.09) -- (214.97,55.09) -- cycle  ;
\draw (224.47,46.59) node  [font=\tiny]  {$v_{13}^{-}$};
% Text Node
\draw    (214.72,67.09) -- (233.72,67.09) -- (233.72,84.09) -- (214.72,84.09) -- cycle  ;
\draw (224.22,75.59) node  [font=\tiny]  {$v_{13}^{-}$};
% Text Node
\draw (238,54.4) node [anchor=north west][inner sep=0.75pt]  [font=\scriptsize]  {$=$};
% Text Node
\draw (276.5,59.5) node    {$\begin{cases}
 \\\\
 \\
\end{cases}$};
% Text Node
\draw    (290.22,68.09) -- (308.22,68.09) -- (308.22,83.09) -- (290.22,83.09) -- cycle  ;
\draw (299.22,75.59) node  [font=\tiny]  {$v_{2}$};
% Text Node
\draw    (290,40.65) -- (308,40.65) -- (308,55.65) -- (290,55.65) -- cycle  ;
\draw (299,48.15) node  [font=\tiny]  {$v_{2}$};

\end{tikzpicture}}
\end{align}
\end{proof}

With the techniques from this Appendix, we are ready to discuss the following proposition:
\begin{prop}
Provided $k \ll N$ the $k$-local operators of the form $X_{I}$ acting on qubits $I=\{i_1,\ldots,i_k\}$ with $i_1<\ldots<i_k$ satisfy $\Var[\partial_{1,1}\langle X_{I}\rangle_{\text{qMPS}}] \in \Omega(c^{-i_k})$ for some $c > 1$.
\end{prop}
The cases $k = 1,2$ are covered in \thref{th_mps} and \thref{th_mpsk2}.
For $k > 2$ we argue as follows:
Given a $k$-local operator acting on qubits $I = \{i_1,\ldots,i_k\}$ with $i_1<\ldots<i_k$, then $\Var[\partial_{1,1}\langle X_{I}\rangle_{\text{qMPS}}]$ corresponds to a tensor network as in Eq.~\eqref{mps_var} but where all registers below the $(i_k+1)$-th qubit do not contribute to the variance.
Hence when contracting the network we accrue contributions from at most $(i_k+1)$ registers.

\newpage

\section{Quantum tree tensor networks}
\label{app_ttn}

For $N = 2^n$ qubits the qTTN ansatz of Eq.~\eqref{ttn_ans} is
\begin{figure}[h]
\centering
\input{tikz_files/ttn_circ}
\caption{\label{fig:ttn_circ}
qTTN circuit considered in this article.
}
\end{figure}

\thref{bp_theorem_zx} and Appendix \ref{app_zx_meth} imply that
\begin{align}
\label{ttn_var}
\scalebox{0.85}{\input{tikz_files/ttn_tn}}
\end{align}
In particular note that the causal cone of any observable in the circuit of Fig.~\ref{fig:ttn_circ} ---equivalently, any vector $u_i$ in Eq.~\eqref{ttn_var}--- contains exactly $1+\log(N) = 1+n$ registers.
We deduce:
\begin{theorem}
\thlabel{th_ttn_app}
Let $\langle X_i\rangle_{\text{qTTN}}$ be the cost function associated with the observable $X_i$ and consider the qTTN ansatz defined in Eq.~\eqref{ttn_ans}, then:
\begin{enumerate}
 \item $\Var[\partial_{1,1}\langle X_N\rangle_{\text{qTTN}}] = \frac{1}{4}\cdot\big(\frac{3}{8}\big)^{n}$,
 \item $\Var[\partial_{1,1}\langle X_1\rangle_{\text{qTTN}}] \in \Omega\Big(\big(\frac{\lambda_2}{4}\big)^{n}\Big)$ where $\lambda_2 \approx 2.3187$.
\end{enumerate}
\end{theorem}
\begin{proof}
We consider $\Var[\partial_{1,1}\langle X_N\rangle_{\text{qTTN}}]$ first:
The parameters not causally connected to $X_N$ in Eq.~\eqref{ttn_var} contract to the identity, reducing the variance to:
\begin{align}
\label{ttn_var_xn}
\scalebox{0.75}{\input{tikz_files/ttn_tn_xn}}
\end{align}
where there are a total of $1+n$ registers.
The contraction here is identical to the one in the proof of \thref{th_mps2} but with an extra register, which gives $\Var[\partial_{1,1}\langle X_N\rangle_{\text{qTTN}}] = \frac{1}{4}\cdot\big(\frac{3}{8}\big)^{n}$.
For the second statement, contracting all the registers not causally connected to $X_1$ yields
\begin{align}
\label{ttn_var_x1}
\scalebox{0.75}{\input{tikz_files/ttn_tn_x1}}
\end{align}
which we simplify by noting that
\begin{align}\label{qttn_identities}
 \scalebox{0.75}{\tikzset{every picture/.style={line width=0.75pt}} %set default line width to 0.75pt        

\begin{tikzpicture}[x=0.75pt,y=0.75pt,yscale=-1,xscale=1]
%uncomment if require: \path (0,300); %set diagram left start at 0, and has height of 300

%Shape: Output [id:dp4625372500555105] 
\draw  [fill={rgb, 255:red, 0; green, 0; blue, 0 }  ,fill opacity=1 ] (337,129.96) .. controls (337,132.72) and (334.76,134.96) .. (332,134.96) .. controls (329.24,134.96) and (327,132.72) .. (327,129.96) .. controls (327,127.2) and (329.24,124.96) .. (332,124.96) .. controls (334.76,124.96) and (337,127.2) .. (337,129.96) -- cycle (342,129.96) -- (337,129.96) (322,129.96) -- (327,129.96) ;
%Shape: Output [id:dp2653814950350526] 
\draw   (317,130) .. controls (317,132.76) and (314.76,135) .. (312,135) .. controls (309.24,135) and (307,132.76) .. (307,130) .. controls (307,127.24) and (309.24,125) .. (312,125) .. controls (314.76,125) and (317,127.24) .. (317,130) -- cycle (322,130) -- (317,130) (302,130) -- (307,130) ;
%Shape: Output [id:dp32922045973656044] 
\draw  [fill={rgb, 255:red, 0; green, 0; blue, 0 }  ,fill opacity=1 ] (337,130) .. controls (337,132.76) and (334.76,135) .. (332,135) .. controls (329.24,135) and (327,132.76) .. (327,130) .. controls (327,127.24) and (329.24,125) .. (332,125) .. controls (334.76,125) and (337,127.24) .. (337,130) -- cycle (342,130) -- (337,130) (322,130) -- (327,130) ;
%Shape: Output [id:dp8252082788223032] 
\draw  [fill={rgb, 255:red, 0; green, 0; blue, 0 }  ,fill opacity=1 ] (377,130) .. controls (377,132.76) and (374.76,135) .. (372,135) .. controls (369.24,135) and (367,132.76) .. (367,130) .. controls (367,127.24) and (369.24,125) .. (372,125) .. controls (374.76,125) and (377,127.24) .. (377,130) -- cycle (382,130) -- (377,130) (362,130) -- (367,130) ;
%Shape: Output [id:dp14386625766061711] 
\draw   (357,130.04) .. controls (357,132.8) and (354.76,135.04) .. (352,135.04) .. controls (349.24,135.04) and (347,132.8) .. (347,130.04) .. controls (347,127.28) and (349.24,125.04) .. (352,125.04) .. controls (354.76,125.04) and (357,127.28) .. (357,130.04) -- cycle (362,130.04) -- (357,130.04) (342,130.04) -- (347,130.04) ;
%Shape: Output [id:dp6733311842174865] 
\draw   (372,140) .. controls (374.76,140) and (377,142.24) .. (377,145) .. controls (377,147.76) and (374.76,150) .. (372,150) .. controls (369.24,150) and (367,147.76) .. (367,145) .. controls (367,142.24) and (369.24,140) .. (372,140) -- cycle (372,135) -- (372,140) (372,155) -- (372,150) ;
%Straight Lines [id:da48615788901149726] 
\draw    (382,130) -- (395,130) ;
%Shape: Output [id:dp4443158663734583] 
\draw  [fill={rgb, 255:red, 0; green, 0; blue, 0 }  ,fill opacity=1 ] (337,159.96) .. controls (337,162.72) and (334.76,164.96) .. (332,164.96) .. controls (329.24,164.96) and (327,162.72) .. (327,159.96) .. controls (327,157.2) and (329.24,154.96) .. (332,154.96) .. controls (334.76,154.96) and (337,157.2) .. (337,159.96) -- cycle (342,159.96) -- (337,159.96) (322,159.96) -- (327,159.96) ;
%Shape: Output [id:dp566270305421152] 
\draw   (317,160) .. controls (317,162.76) and (314.76,165) .. (312,165) .. controls (309.24,165) and (307,162.76) .. (307,160) .. controls (307,157.24) and (309.24,155) .. (312,155) .. controls (314.76,155) and (317,157.24) .. (317,160) -- cycle (322,160) -- (317,160) (302,160) -- (307,160) ;
%Shape: Circle [id:dp37591570683761666] 
\draw  [fill={rgb, 255:red, 0; green, 0; blue, 0 }  ,fill opacity=1 ] (297.01,158.94) .. controls (297.01,156.46) and (294.99,154.44) .. (292.51,154.44) .. controls (290.02,154.44) and (288.01,156.46) .. (288.01,158.94) .. controls (288.01,161.43) and (290.02,163.44) .. (292.51,163.44) .. controls (294.99,163.44) and (297.01,161.43) .. (297.01,158.94) -- cycle ;
%Straight Lines [id:da055479049203680963] 
\draw    (297.01,159.94) -- (302,160) ;
%Shape: Output [id:dp13189450641992106] 
\draw  [fill={rgb, 255:red, 0; green, 0; blue, 0 }  ,fill opacity=1 ] (377,160) .. controls (377,162.76) and (374.76,165) .. (372,165) .. controls (369.24,165) and (367,162.76) .. (367,160) .. controls (367,157.24) and (369.24,155) .. (372,155) .. controls (374.76,155) and (377,157.24) .. (377,160) -- cycle (382,160) -- (377,160) (362,160) -- (367,160) ;
%Shape: Output [id:dp13976371857200243] 
\draw   (357,160) .. controls (357,162.76) and (354.76,165) .. (352,165) .. controls (349.24,165) and (347,162.76) .. (347,160) .. controls (347,157.24) and (349.24,155) .. (352,155) .. controls (354.76,155) and (357,157.24) .. (357,160) -- cycle (362,160) -- (357,160) (342,160) -- (347,160) ;
%Straight Lines [id:da4652313539359012] 
\draw    (280,130) -- (302,130) ;
%Straight Lines [id:da37385233972915977] 
\draw    (382,160) -- (395,160) ;

% Text Node
\draw (369,117.4) node [anchor=north west][inner sep=0.75pt]  [font=\tiny]  {$2$};
% Text Node
\draw (328,117.4) node [anchor=north west][inner sep=0.75pt]  [font=\tiny]  {$2$};
% Text Node
\draw (359,142.4) node [anchor=north west][inner sep=0.75pt]  [font=\tiny]  {$4$};
% Text Node
\draw (289,146.4) node [anchor=north west][inner sep=0.75pt]  [font=\tiny]  {$2$};
% Text Node
\draw (329,147.4) node [anchor=north west][inner sep=0.75pt]  [font=\tiny]  {$2$};
% Text Node
\draw (369,167.4) node [anchor=north west][inner sep=0.75pt]  [font=\tiny]  {$2$};
% Text Node
\draw    (395,153) -- (413,153) -- (413,168) -- (395,168) -- cycle  ;
\draw (398,157.4) node [anchor=north west][inner sep=0.75pt]  [font=\tiny]  {$v_{13}$};
% Text Node
\draw    (395,123) -- (413,123) -- (413,138) -- (395,138) -- cycle  ;
\draw (398,127.4) node [anchor=north west][inner sep=0.75pt]  [font=\tiny]  {$v_{2}$};

\end{tikzpicture}} = \frac{1}{4}(3 \vt + \votm), \qquad \scalebox{0.75}{\tikzset{every picture/.style={line width=0.75pt}} %set default line width to 0.75pt        

\begin{tikzpicture}[x=0.75pt,y=0.75pt,yscale=-1,xscale=1]
%uncomment if require: \path (0,300); %set diagram left start at 0, and has height of 300

%Shape: Output [id:dp4625372500555105] 
\draw  [fill={rgb, 255:red, 0; green, 0; blue, 0 }  ,fill opacity=1 ] (337,129.96) .. controls (337,132.72) and (334.76,134.96) .. (332,134.96) .. controls (329.24,134.96) and (327,132.72) .. (327,129.96) .. controls (327,127.2) and (329.24,124.96) .. (332,124.96) .. controls (334.76,124.96) and (337,127.2) .. (337,129.96) -- cycle (342,129.96) -- (337,129.96) (322,129.96) -- (327,129.96) ;
%Shape: Output [id:dp2653814950350526] 
\draw   (317,130) .. controls (317,132.76) and (314.76,135) .. (312,135) .. controls (309.24,135) and (307,132.76) .. (307,130) .. controls (307,127.24) and (309.24,125) .. (312,125) .. controls (314.76,125) and (317,127.24) .. (317,130) -- cycle (322,130) -- (317,130) (302,130) -- (307,130) ;
%Shape: Output [id:dp32922045973656044] 
\draw  [fill={rgb, 255:red, 0; green, 0; blue, 0 }  ,fill opacity=1 ] (337,130) .. controls (337,132.76) and (334.76,135) .. (332,135) .. controls (329.24,135) and (327,132.76) .. (327,130) .. controls (327,127.24) and (329.24,125) .. (332,125) .. controls (334.76,125) and (337,127.24) .. (337,130) -- cycle (342,130) -- (337,130) (322,130) -- (327,130) ;
%Shape: Output [id:dp8252082788223032] 
\draw  [fill={rgb, 255:red, 0; green, 0; blue, 0 }  ,fill opacity=1 ] (377,130) .. controls (377,132.76) and (374.76,135) .. (372,135) .. controls (369.24,135) and (367,132.76) .. (367,130) .. controls (367,127.24) and (369.24,125) .. (372,125) .. controls (374.76,125) and (377,127.24) .. (377,130) -- cycle (382,130) -- (377,130) (362,130) -- (367,130) ;
%Shape: Output [id:dp14386625766061711] 
\draw   (357,130.04) .. controls (357,132.8) and (354.76,135.04) .. (352,135.04) .. controls (349.24,135.04) and (347,132.8) .. (347,130.04) .. controls (347,127.28) and (349.24,125.04) .. (352,125.04) .. controls (354.76,125.04) and (357,127.28) .. (357,130.04) -- cycle (362,130.04) -- (357,130.04) (342,130.04) -- (347,130.04) ;
%Shape: Output [id:dp6733311842174865] 
\draw   (372,140) .. controls (374.76,140) and (377,142.24) .. (377,145) .. controls (377,147.76) and (374.76,150) .. (372,150) .. controls (369.24,150) and (367,147.76) .. (367,145) .. controls (367,142.24) and (369.24,140) .. (372,140) -- cycle (372,135) -- (372,140) (372,155) -- (372,150) ;
%Straight Lines [id:da48615788901149726] 
\draw    (382,130) -- (395,130) ;
%Shape: Output [id:dp4443158663734583] 
\draw  [fill={rgb, 255:red, 0; green, 0; blue, 0 }  ,fill opacity=1 ] (337,159.96) .. controls (337,162.72) and (334.76,164.96) .. (332,164.96) .. controls (329.24,164.96) and (327,162.72) .. (327,159.96) .. controls (327,157.2) and (329.24,154.96) .. (332,154.96) .. controls (334.76,154.96) and (337,157.2) .. (337,159.96) -- cycle (342,159.96) -- (337,159.96) (322,159.96) -- (327,159.96) ;
%Shape: Output [id:dp566270305421152] 
\draw   (317,160) .. controls (317,162.76) and (314.76,165) .. (312,165) .. controls (309.24,165) and (307,162.76) .. (307,160) .. controls (307,157.24) and (309.24,155) .. (312,155) .. controls (314.76,155) and (317,157.24) .. (317,160) -- cycle (322,160) -- (317,160) (302,160) -- (307,160) ;
%Shape: Circle [id:dp37591570683761666] 
\draw  [fill={rgb, 255:red, 0; green, 0; blue, 0 }  ,fill opacity=1 ] (297.01,158.94) .. controls (297.01,156.46) and (294.99,154.44) .. (292.51,154.44) .. controls (290.02,154.44) and (288.01,156.46) .. (288.01,158.94) .. controls (288.01,161.43) and (290.02,163.44) .. (292.51,163.44) .. controls (294.99,163.44) and (297.01,161.43) .. (297.01,158.94) -- cycle ;
%Straight Lines [id:da055479049203680963] 
\draw    (297.01,159.94) -- (302,160) ;
%Shape: Output [id:dp13189450641992106] 
\draw  [fill={rgb, 255:red, 0; green, 0; blue, 0 }  ,fill opacity=1 ] (377,160) .. controls (377,162.76) and (374.76,165) .. (372,165) .. controls (369.24,165) and (367,162.76) .. (367,160) .. controls (367,157.24) and (369.24,155) .. (372,155) .. controls (374.76,155) and (377,157.24) .. (377,160) -- cycle (382,160) -- (377,160) (362,160) -- (367,160) ;
%Shape: Output [id:dp13976371857200243] 
\draw   (357,160) .. controls (357,162.76) and (354.76,165) .. (352,165) .. controls (349.24,165) and (347,162.76) .. (347,160) .. controls (347,157.24) and (349.24,155) .. (352,155) .. controls (354.76,155) and (357,157.24) .. (357,160) -- cycle (362,160) -- (357,160) (342,160) -- (347,160) ;
%Straight Lines [id:da4652313539359012] 
\draw    (280,130) -- (302,130) ;
%Straight Lines [id:da37385233972915977] 
\draw    (382,160) -- (395,160) ;

% Text Node
\draw (369,117.4) node [anchor=north west][inner sep=0.75pt]  [font=\tiny]  {$2$};
% Text Node
\draw (328,117.4) node [anchor=north west][inner sep=0.75pt]  [font=\tiny]  {$2$};
% Text Node
\draw (359,142.4) node [anchor=north west][inner sep=0.75pt]  [font=\tiny]  {$4$};
% Text Node
\draw (289,146.4) node [anchor=north west][inner sep=0.75pt]  [font=\tiny]  {$2$};
% Text Node
\draw (329,147.4) node [anchor=north west][inner sep=0.75pt]  [font=\tiny]  {$2$};
% Text Node
\draw (369,167.4) node [anchor=north west][inner sep=0.75pt]  [font=\tiny]  {$2$};
% Text Node
\draw    (395,153) -- (413,153) -- (413,168) -- (395,168) -- cycle  ;
\draw (398,157.4) node [anchor=north west][inner sep=0.75pt]  [font=\tiny]  {$v_{13}$};
% Text Node
\draw    (395,123) -- (415,123) -- (415,138) -- (395,138) -- cycle  ;
\draw (398,122) node [anchor=north west][inner sep=0.75pt]  [font=\tiny]  {$v_{13}^-$};

\end{tikzpicture}} = 2(\vt + \votm).
\end{align}
If we denote the resulting vector after the $k$-th application of the term within the square brackets in Eq.~\eqref{ttn_var_x1} by $\mathbf{v_k} = \frac{1}{4}(\alpha_k \vt + \beta_k \votm)$ with $\mathbf{v_0} = \frac{1}{4}(3 \vt + \votm)$, then by using the identities in Eq.~\eqref{qttn_identities} we find that any subsequent term is given by $\mathbf{v_{k+1}} = \frac{1}{4}(\alpha_{k+1} \vt + \beta_{k+1} \votm)$ for 
\begin{align}
 \alpha_{k+1} = \frac{1}{4}(3 \alpha_k + 8 \beta_k) \quad \text{and} \quad \beta_{k+1} = \frac{1}{4}(\alpha_k + 8 \beta_k).
\end{align}
Let $\mathbf{u_k} := \frac{1}{4}[\alpha_k, \beta_k]^T$ be the coefficient vector associated with $\mathbf{v_k}$, then the transformation $\mathbf{v_k} \rightarrow \mathbf{v_{k+1}}$ is determined by the linear map:
\begin{align}
 M:\mathbf{u_k} \mapsto \mathbf{u_{k+1}}, \quad M = \frac{1}{4} \begin{bmatrix}
  3 & 8 \\
  1 & 8
 \end{bmatrix}.
\end{align}
$M$ has eigenvalues $\lambda_1 \approx 0.4313$ and $\lambda_2 \approx 2.3187$ and respective eigenvectors $\mathbf{w_1, w_2}$ so that the spectral theorem implies that after the application of the $(n-1)$ terms in the square brackets we obtain
\begin{align}
 \mathbf{u_{n-1}} =
 \begin{bmatrix}
  \alpha_{n-1}\\
  \beta_{n-1}
 \end{bmatrix}
 = M^{n-1} \mathbf{u_0} = [\mathbf{w_1}, \mathbf{w_2}] \begin{bmatrix}
  \lambda_1^{n-1} & 0\\
  0 & \lambda_2^{n-1}
 \end{bmatrix} [\mathbf{w_1}, \mathbf{w_2}]^{-1} \mathbf{u_0}.
\end{align}
Contracting the rest of the tensor network then gives
\begin{align}
 \scalebox{0.85}{\tikzset{every picture/.style={line width=0.75pt}} %set default line width to 0.75pt        

\begin{tikzpicture}[x=0.75pt,y=0.75pt,yscale=-1,xscale=1]
%uncomment if require: \path (0,300); %set diagram left start at 0, and has height of 300

%Shape: Output [id:dp2242499425516693] 
\draw  [fill={rgb, 255:red, 0; green, 0; blue, 0 }  ,fill opacity=1 ] (185.99,74.92) .. controls (185.99,77.68) and (183.75,79.92) .. (180.99,79.92) .. controls (178.23,79.92) and (175.99,77.68) .. (175.99,74.92) .. controls (175.99,72.16) and (178.23,69.92) .. (180.99,69.92) .. controls (183.75,69.92) and (185.99,72.16) .. (185.99,74.92) -- cycle (190.99,74.92) -- (185.99,74.92) (170.99,74.92) -- (175.99,74.92) ;
%Shape: Output [id:dp1041237647180373] 
\draw   (165.99,74.96) .. controls (165.99,77.72) and (163.75,79.96) .. (160.99,79.96) .. controls (158.23,79.96) and (155.99,77.72) .. (155.99,74.96) .. controls (155.99,72.2) and (158.23,69.96) .. (160.99,69.96) .. controls (163.75,69.96) and (165.99,72.2) .. (165.99,74.96) -- cycle (170.99,74.96) -- (165.99,74.96) (150.99,74.96) -- (155.99,74.96) ;
%Shape: Circle [id:dp7869281669766608] 
\draw  [fill={rgb, 255:red, 0; green, 0; blue, 0 }  ,fill opacity=1 ] (146,74.9) .. controls (146,72.42) and (143.99,70.4) .. (141.5,70.4) .. controls (139.01,70.4) and (137,72.42) .. (137,74.9) .. controls (137,77.39) and (139.01,79.4) .. (141.5,79.4) .. controls (143.99,79.4) and (146,77.39) .. (146,74.9) -- cycle ;
%Straight Lines [id:da4221967159659923] 
\draw    (146,74.9) -- (150.99,74.96) ;
%Shape: Output [id:dp6497672811210435] 
\draw  [fill={rgb, 255:red, 0; green, 0; blue, 0 }  ,fill opacity=1 ] (226,75) .. controls (226,77.76) and (223.76,80) .. (221,80) .. controls (218.24,80) and (216,77.76) .. (216,75) .. controls (216,72.24) and (218.24,70) .. (221,70) .. controls (223.76,70) and (226,72.24) .. (226,75) -- cycle (231,75) -- (226,75) (211,75) -- (216,75) ;
%Shape: Output [id:dp4388757618589736] 
\draw   (206,75) .. controls (206,77.76) and (203.76,80) .. (201,80) .. controls (198.24,80) and (196,77.76) .. (196,75) .. controls (196,72.24) and (198.24,70) .. (201,70) .. controls (203.76,70) and (206,72.24) .. (206,75) -- cycle (211,75) -- (206,75) (191,75) -- (196,75) ;
%Shape: Output [id:dp9181782291286571] 
\draw  [fill={rgb, 255:red, 0; green, 0; blue, 0 }  ,fill opacity=1 ] (226,45) .. controls (226,47.76) and (223.76,50) .. (221,50) .. controls (218.24,50) and (216,47.76) .. (216,45) .. controls (216,42.24) and (218.24,40) .. (221,40) .. controls (223.76,40) and (226,42.24) .. (226,45) -- cycle (231,45) -- (226,45) (211,45) -- (216,45) ;
%Shape: Output [id:dp10593312992149317] 
\draw   (206,45.04) .. controls (206,47.8) and (203.76,50.04) .. (201,50.04) .. controls (198.24,50.04) and (196,47.8) .. (196,45.04) .. controls (196,42.28) and (198.24,40.04) .. (201,40.04) .. controls (203.76,40.04) and (206,42.28) .. (206,45.04) -- cycle (211,45.04) -- (206,45.04) (191,45.04) -- (196,45.04) ;
%Straight Lines [id:da7973491233074059] 
\draw    (151,45) -- (191,45.04) ;
%Shape: Output [id:dp7377629782406394] 
\draw   (221,55) .. controls (223.76,55) and (226,57.24) .. (226,60) .. controls (226,62.76) and (223.76,65) .. (221,65) .. controls (218.24,65) and (216,62.76) .. (216,60) .. controls (216,57.24) and (218.24,55) .. (221,55) -- cycle (221,50) -- (221,55) (221,70) -- (221,65) ;
%Straight Lines [id:da01997726683492851] 
\draw    (231,45) -- (244,45) ;
%Straight Lines [id:da9229579370104832] 
\draw    (231,75) -- (244,75) ;
%Shape: Output [id:dp7406481455636336] 
\draw  [fill={rgb, 255:red, 0; green, 0; blue, 0 }  ,fill opacity=1 ] (189.99,143.92) .. controls (189.99,146.68) and (187.75,148.92) .. (184.99,148.92) .. controls (182.23,148.92) and (179.99,146.68) .. (179.99,143.92) .. controls (179.99,141.16) and (182.23,138.92) .. (184.99,138.92) .. controls (187.75,138.92) and (189.99,141.16) .. (189.99,143.92) -- cycle (194.99,143.92) -- (189.99,143.92) (174.99,143.92) -- (179.99,143.92) ;
%Shape: Output [id:dp168206332616202] 
\draw   (169.99,143.96) .. controls (169.99,146.72) and (167.75,148.96) .. (164.99,148.96) .. controls (162.23,148.96) and (159.99,146.72) .. (159.99,143.96) .. controls (159.99,141.2) and (162.23,138.96) .. (164.99,138.96) .. controls (167.75,138.96) and (169.99,141.2) .. (169.99,143.96) -- cycle (174.99,143.96) -- (169.99,143.96) (154.99,143.96) -- (159.99,143.96) ;
%Shape: Circle [id:dp3999758888677909] 
\draw  [fill={rgb, 255:red, 0; green, 0; blue, 0 }  ,fill opacity=1 ] (150,143.9) .. controls (150,141.42) and (147.99,139.4) .. (145.5,139.4) .. controls (143.01,139.4) and (141,141.42) .. (141,143.9) .. controls (141,146.39) and (143.01,148.4) .. (145.5,148.4) .. controls (147.99,148.4) and (150,146.39) .. (150,143.9) -- cycle ;
%Straight Lines [id:da29841315795721846] 
\draw    (150,143.9) -- (154.99,143.96) ;
%Straight Lines [id:da3454843760219042] 
\draw    (155,114) -- (195,114.04) ;
%Straight Lines [id:da17242114726231716] 
\draw    (245,114) -- (258,114) ;
%Straight Lines [id:da6502644514377609] 
\draw    (245,144) -- (258,144) ;
%Straight Lines [id:da787156616850408] 
\draw    (349,114) -- (362,114) ;
%Straight Lines [id:da06599080594920714] 
\draw    (349,144) -- (362,144) ;

% Text Node
\draw    (133,38) -- (151,38) -- (151,53) -- (133,53) -- cycle  ;
\draw (142,45.5) node  [font=\tiny]  {$P_{2}$};
% Text Node
\draw (106,37.4) node [anchor=north west][inner sep=0.75pt]  [font=\tiny]  {$\frac{1}{4^{n+1}}$};
% Text Node
\draw (138,62.4) node [anchor=north west][inner sep=0.75pt]  [font=\tiny]  {$2$};
% Text Node
\draw (178,62.4) node [anchor=north west][inner sep=0.75pt]  [font=\tiny]  {$2$};
% Text Node
\draw (218,32.4) node [anchor=north west][inner sep=0.75pt]  [font=\tiny]  {$2$};
% Text Node
\draw (218,82.4) node [anchor=north west][inner sep=0.75pt]  [font=\tiny]  {$2$};
% Text Node
\draw (208,57.4) node [anchor=north west][inner sep=0.75pt]  [font=\tiny]  {$4$};
% Text Node
\draw (138,30.4) node [anchor=north west][inner sep=0.75pt]  [font=\tiny]  {$2$};
% Text Node
\draw (0,54.4) node [anchor=north west][inner sep=0.75pt]  [font=\scriptsize]  {$\text{Var}[ \partial _{1}{}_{,}{}_{1} \langle X_{1} \rangle _{TTN}] =$};
% Text Node
\draw    (244,68) -- (264,68) -- (264,83) -- (244,83) -- cycle  ;
\draw (247,72.4) node [anchor=north west][inner sep=0.75pt]  [font=\tiny]  {$v_{13}$};
% Text Node
\draw    (244,38) -- (335,38) -- (335,54) -- (244,54) -- cycle  ;
\draw (247,39) node [anchor=north west][inner sep=0.75pt]  [font=\tiny]  {$\alpha _{n-1} v_{2} +\beta _{n-1} v_{13}^{-}$};
% Text Node
\draw    (137,107) -- (155,107) -- (155,122) -- (137,122) -- cycle  ;
\draw (146,114.5) node  [font=\tiny]  {$P_{2}$};
% Text Node
\draw (110,106.4) node [anchor=north west][inner sep=0.75pt]  [font=\tiny]  {$\frac{1}{4^{n+1}}$};
% Text Node
\draw (142,131.4) node [anchor=north west][inner sep=0.75pt]  [font=\tiny]  {$2$};
% Text Node
\draw (182,131.4) node [anchor=north west][inner sep=0.75pt]  [font=\tiny]  {$2$};
% Text Node
\draw (142,99.4) node [anchor=north west][inner sep=0.75pt]  [font=\tiny]  {$2$};
% Text Node
\draw (84,123.4) node [anchor=north west][inner sep=0.75pt]  [font=\scriptsize]  {$=$};
% Text Node
\draw    (258,137) -- (279,137) -- (279,152) -- (258,152) -- cycle  ;
\draw (261,141.4) node [anchor=north west][inner sep=0.75pt]  [font=\tiny]  {$v_{2}$};
% Text Node
\draw    (258,107) -- (300,107) -- (300,124) -- (258,124) -- cycle  ;
\draw (261,108) node [anchor=north west][inner sep=0.75pt]  [font=\tiny]  {$v_{2} +v_{13}^{-}$};
% Text Node
\draw    (362,137) -- (383,137) -- (383,152) -- (362,152) -- cycle  ;
\draw (365,141.4) node [anchor=north west][inner sep=0.75pt]  [font=\tiny]  {$v_{13}$};
% Text Node
\draw    (362,107) -- (380,107) -- (380,122) -- (362,122) -- cycle  ;
\draw (365,111.4) node [anchor=north west][inner sep=0.75pt]  [font=\tiny]  {$v_{2}$};
% Text Node
\draw (224.47,131.41) node  [font=\tiny]  {$+\frac{\alpha _{n-1}}{2}\begin{cases}
 \\\\
 \begin{array}{l}
\end{array}\\
 \\
\end{cases}$};
% Text Node
\draw (330.27,130.88) node  [font=\tiny]  {$+\beta _{n-1}\begin{cases}
 \\\\
 \begin{array}{l}
\end{array}\\
 \\
\end{cases}$};

\end{tikzpicture}}\\
 \label{ttn_var_x1_exact}
 = \frac{1}{4^{n+1}}(\alpha_{n-1} + 8 \beta_{n-1}).
\end{align}
We approximate the above by noticing that for $n$ large enough, $\lambda_2^{n-1} \gg \lambda_1^{n-1} \approx 0$ and so $\alpha_n, \beta_n \in O(\lambda_2^{n-1})$ so that
\begin{align}
\label{ttn_var_x1_asymp}
 \Var[\partial_{1,1}\langle X_1\rangle_{\text{qTTN}}] = \frac{1}{4^{n+1}}(\alpha_{n-1} + 8 \beta_{n-1}) \in \Omega\Big(\big(\frac{\lambda_2}{4}\big)^n \Big).
\end{align}
\end{proof}

In general we obtain the gradient variance corresponding to any observable of the form $X_i$ analytically by contracting the tensor network in Eq.~\eqref{ttn_var} using the identities in Eqs.~\eqref{zx_identities1}, \eqref{zx_identities2}, \eqref{zx_identities3} as was demonstrated in the proof of \thref{th_ttn_app}.
When performing this contraction, we encounter two types of operations:
\begin{align}
\label{ent_up}
 \scalebox{0.9}{\tikzset{every picture/.style={line width=0.75pt}} %set default line width to 0.75pt        

\begin{tikzpicture}[x=0.75pt,y=0.75pt,yscale=-1,xscale=1]
%uncomment if require: \path (0,300); %set diagram left start at 0, and has height of 300

%Shape: Output [id:dp07818236478674745] 
\draw  [fill={rgb, 255:red, 0; green, 0; blue, 0 }  ,fill opacity=1 ] (125,95) .. controls (125,97.76) and (122.76,100) .. (120,100) .. controls (117.24,100) and (115,97.76) .. (115,95) .. controls (115,92.24) and (117.24,90) .. (120,90) .. controls (122.76,90) and (125,92.24) .. (125,95) -- cycle (130,95) -- (125,95) (110,95) -- (115,95) ;
%Shape: Output [id:dp09759406555231265] 
\draw   (105,95) .. controls (105,97.76) and (102.76,100) .. (100,100) .. controls (97.24,100) and (95,97.76) .. (95,95) .. controls (95,92.24) and (97.24,90) .. (100,90) .. controls (102.76,90) and (105,92.24) .. (105,95) -- cycle (110,95) -- (105,95) (90,95) -- (95,95) ;
%Shape: Output [id:dp9778463797047954] 
\draw  [fill={rgb, 255:red, 0; green, 0; blue, 0 }  ,fill opacity=1 ] (125,65) .. controls (125,67.76) and (122.76,70) .. (120,70) .. controls (117.24,70) and (115,67.76) .. (115,65) .. controls (115,62.24) and (117.24,60) .. (120,60) .. controls (122.76,60) and (125,62.24) .. (125,65) -- cycle (130,65) -- (125,65) (110,65) -- (115,65) ;
%Shape: Output [id:dp7729399387914844] 
\draw   (105,65.04) .. controls (105,67.8) and (102.76,70.04) .. (100,70.04) .. controls (97.24,70.04) and (95,67.8) .. (95,65.04) .. controls (95,62.28) and (97.24,60.04) .. (100,60.04) .. controls (102.76,60.04) and (105,62.28) .. (105,65.04) -- cycle (110,65.04) -- (105,65.04) (90,65.04) -- (95,65.04) ;
%Shape: Output [id:dp15502835444323448] 
\draw   (120,75) .. controls (122.76,75) and (125,77.24) .. (125,80) .. controls (125,82.76) and (122.76,85) .. (120,85) .. controls (117.24,85) and (115,82.76) .. (115,80) .. controls (115,77.24) and (117.24,75) .. (120,75) -- cycle (120,70) -- (120,75) (120,90) -- (120,85) ;
%Straight Lines [id:da6791607512525375] 
\draw    (135,95) -- (145,95) ;
%Straight Lines [id:da5272707997536914] 
\draw    (195,95) -- (205,95) ;
%Straight Lines [id:da5677007112872321] 
\draw    (250,95) -- (260,95) ;
%Straight Lines [id:da7175041517292748] 
\draw    (345,67) -- (355,67) ;
%Straight Lines [id:da24808539939440233] 
\draw    (429,67) -- (439,67) ;
%Shape: Output [id:dp13038929401326094] 
\draw  [fill={rgb, 255:red, 0; green, 0; blue, 0 }  ,fill opacity=1 ] (84.99,94.92) .. controls (84.99,97.68) and (82.75,99.92) .. (79.99,99.92) .. controls (77.23,99.92) and (74.99,97.68) .. (74.99,94.92) .. controls (74.99,92.16) and (77.23,89.92) .. (79.99,89.92) .. controls (82.75,89.92) and (84.99,92.16) .. (84.99,94.92) -- cycle (89.99,94.92) -- (84.99,94.92) (69.99,94.92) -- (74.99,94.92) ;
%Shape: Output [id:dp09813007312855859] 
\draw   (64.99,94.96) .. controls (64.99,97.72) and (62.75,99.96) .. (59.99,99.96) .. controls (57.23,99.96) and (54.99,97.72) .. (54.99,94.96) .. controls (54.99,92.2) and (57.23,89.96) .. (59.99,89.96) .. controls (62.75,89.96) and (64.99,92.2) .. (64.99,94.96) -- cycle (69.99,94.96) -- (64.99,94.96) (49.99,94.96) -- (54.99,94.96) ;
%Shape: Circle [id:dp2902664477892849] 
\draw  [fill={rgb, 255:red, 0; green, 0; blue, 0 }  ,fill opacity=1 ] (45,94.9) .. controls (45,92.42) and (42.99,90.4) .. (40.5,90.4) .. controls (38.01,90.4) and (36,92.42) .. (36,94.9) .. controls (36,97.39) and (38.01,99.4) .. (40.5,99.4) .. controls (42.99,99.4) and (45,97.39) .. (45,94.9) -- cycle ;
%Straight Lines [id:da29309428872142007] 
\draw    (45,94.9) -- (49.99,94.96) ;

% Text Node
\draw (117,52.4) node [anchor=north west][inner sep=0.75pt]  [font=\tiny]  {$2$};
% Text Node
\draw (117,102.4) node [anchor=north west][inner sep=0.75pt]  [font=\tiny]  {$2$};
% Text Node
\draw (107,77.4) node [anchor=north west][inner sep=0.75pt]  [font=\tiny]  {$4$};
% Text Node
\draw    (130,58) -- (149,58) -- (149,73) -- (130,73) -- cycle  ;
\draw (133,62.4) node [anchor=north west][inner sep=0.75pt]  [font=\tiny]  {$v_{13}$};
% Text Node
\draw (129,84.4) node [anchor=north west][inner sep=0.75pt]    {$\{$};
% Text Node
\draw    (145,87) -- (176,87) -- (176,102) -- (145,102) -- cycle  ;
\draw (146,91.4) node [anchor=north west][inner sep=0.75pt]  [font=\tiny]  {$c_{13} v_{13}$};
% Text Node
\draw    (205,87) -- (229,87) -- (229,102) -- (205,102) -- cycle  ;
\draw (206,91.4) node [anchor=north west][inner sep=0.75pt]  [font=\tiny]  {$c_{2} v_{2}$};
% Text Node
\draw (180,88.4) node [anchor=north west][inner sep=0.75pt]  [font=\scriptsize]  {$+$};
% Text Node
\draw    (260,87) -- (291,87) -- (291,104) -- (260,104) -- cycle  ;
\draw (261,89) node [anchor=north west][inner sep=0.75pt]  [font=\tiny]  {$c_{13}^{-} v_{13}^{-}$};
% Text Node
\draw (234,88.4) node [anchor=north west][inner sep=0.75pt]  [font=\scriptsize]  {$+$};
% Text Node
\draw (290,86.4) node [anchor=north west][inner sep=0.75pt]    {$\}$};
% Text Node
\draw (305,73.4) node [anchor=north west][inner sep=0.75pt]  [font=\scriptsize]  {$=$};
% Text Node
\draw (327,60.4) node [anchor=north west][inner sep=0.75pt]    {$\begin{cases}
 & \\
 & 
\end{cases}$};
% Text Node
\draw    (355,59) -- (374,59) -- (374,74) -- (355,74) -- cycle  ;
\draw (356,63.4) node [anchor=north west][inner sep=0.75pt]  [font=\tiny]  {$v_{13}$};
% Text Node
\draw (412,60.4) node [anchor=north west][inner sep=0.75pt]    {$\begin{cases}
 & \\
 & 
\end{cases}$};
% Text Node
\draw    (439,59) -- (457,59) -- (457,74) -- (439,74) -- cycle  ;
\draw (442,63.4) node [anchor=north west][inner sep=0.75pt]  [font=\tiny]  {$v_{2}$};
% Text Node
\draw (395,73.4) node [anchor=north west][inner sep=0.75pt]  [font=\scriptsize]  {$+$};
% Text Node
\draw (26,86.4) node [anchor=north west][inner sep=0.75pt]  [font=\tiny]  {$\frac{1}{4}$};
% Text Node
\draw (37,82.4) node [anchor=north west][inner sep=0.75pt]  [font=\tiny]  {$2$};
% Text Node
\draw (77,82.4) node [anchor=north west][inner sep=0.75pt]  [font=\tiny]  {$2$};
% Text Node
\draw (363.5,89) node  [font=\tiny]  {$c_{13} +\frac{c_{13}^-}{4}$};
% Text Node
\draw (447.5,90) node  [font=\tiny]  {$\frac{3c_{2}}{8}$};

\end{tikzpicture}}
\end{align}
which occurs when the contribution to the variance originating from the observable travels `upwards' in the tensor network in Eq.~\eqref{ttn_var} and
\begin{align}
\label{ent_down}
\scalebox{0.9}{\tikzset{every picture/.style={line width=0.75pt}} %set default line width to 0.75pt        

\begin{tikzpicture}[x=0.75pt,y=0.75pt,yscale=-1,xscale=1]
%uncomment if require: \path (0,300); %set diagram left start at 0, and has height of 300

%Shape: Output [id:dp4400898656014989] 
\draw  [fill={rgb, 255:red, 0; green, 0; blue, 0 }  ,fill opacity=1 ] (125,105) .. controls (125,107.76) and (122.76,110) .. (120,110) .. controls (117.24,110) and (115,107.76) .. (115,105) .. controls (115,102.24) and (117.24,100) .. (120,100) .. controls (122.76,100) and (125,102.24) .. (125,105) -- cycle (130,105) -- (125,105) (110,105) -- (115,105) ;
%Shape: Output [id:dp41428638583561184] 
\draw   (105,105.04) .. controls (105,107.8) and (102.76,110.04) .. (100,110.04) .. controls (97.24,110.04) and (95,107.8) .. (95,105.04) .. controls (95,102.28) and (97.24,100.04) .. (100,100.04) .. controls (102.76,100.04) and (105,102.28) .. (105,105.04) -- cycle (110,105.04) -- (105,105.04) (90,105.04) -- (95,105.04) ;
%Shape: Output [id:dp09883701983649362] 
\draw   (120,115) .. controls (122.76,115) and (125,117.24) .. (125,120) .. controls (125,122.76) and (122.76,125) .. (120,125) .. controls (117.24,125) and (115,122.76) .. (115,120) .. controls (115,117.24) and (117.24,115) .. (120,115) -- cycle (120,110) -- (120,115) (120,130) -- (120,125) ;
%Straight Lines [id:da5625195369406149] 
\draw    (130,105) -- (143,105) ;
%Shape: Output [id:dp3310070730960133] 
\draw  [fill={rgb, 255:red, 0; green, 0; blue, 0 }  ,fill opacity=1 ] (85,134.96) .. controls (85,137.72) and (82.76,139.96) .. (80,139.96) .. controls (77.24,139.96) and (75,137.72) .. (75,134.96) .. controls (75,132.2) and (77.24,129.96) .. (80,129.96) .. controls (82.76,129.96) and (85,132.2) .. (85,134.96) -- cycle (90,134.96) -- (85,134.96) (70,134.96) -- (75,134.96) ;
%Shape: Output [id:dp2094004042365094] 
\draw   (65,135) .. controls (65,137.76) and (62.76,140) .. (60,140) .. controls (57.24,140) and (55,137.76) .. (55,135) .. controls (55,132.24) and (57.24,130) .. (60,130) .. controls (62.76,130) and (65,132.24) .. (65,135) -- cycle (70,135) -- (65,135) (50,135) -- (55,135) ;
%Shape: Circle [id:dp5737184230104588] 
\draw  [fill={rgb, 255:red, 0; green, 0; blue, 0 }  ,fill opacity=1 ] (45.01,133.94) .. controls (45.01,131.46) and (42.99,129.44) .. (40.51,129.44) .. controls (38.02,129.44) and (36.01,131.46) .. (36.01,133.94) .. controls (36.01,136.43) and (38.02,138.44) .. (40.51,138.44) .. controls (42.99,138.44) and (45.01,136.43) .. (45.01,133.94) -- cycle ;
%Straight Lines [id:da42007515430701536] 
\draw    (45.01,134.94) -- (50,135) ;
%Shape: Output [id:dp211288071826369] 
\draw  [fill={rgb, 255:red, 0; green, 0; blue, 0 }  ,fill opacity=1 ] (125,135) .. controls (125,137.76) and (122.76,140) .. (120,140) .. controls (117.24,140) and (115,137.76) .. (115,135) .. controls (115,132.24) and (117.24,130) .. (120,130) .. controls (122.76,130) and (125,132.24) .. (125,135) -- cycle (130,135) -- (125,135) (110,135) -- (115,135) ;
%Shape: Output [id:dp5462860866786869] 
\draw   (105,135) .. controls (105,137.76) and (102.76,140) .. (100,140) .. controls (97.24,140) and (95,137.76) .. (95,135) .. controls (95,132.24) and (97.24,130) .. (100,130) .. controls (102.76,130) and (105,132.24) .. (105,135) -- cycle (110,135) -- (105,135) (90,135) -- (95,135) ;
%Straight Lines [id:da26426803991638104] 
\draw    (130,135) -- (143,135) ;
%Straight Lines [id:da9779492548770545] 
\draw    (361,104.5) -- (371,104.5) ;
%Straight Lines [id:da2529645708236765] 
\draw    (445,104.5) -- (455,104.5) ;
%Straight Lines [id:da6411777877889315] 
\draw    (534,104.5) -- (544,104.5) ;
%Straight Lines [id:da8203563945805938] 
\draw    (152,105) -- (162,105) ;
%Straight Lines [id:da40613640186790656] 
\draw    (212,105) -- (222,105) ;
%Straight Lines [id:da3783455887486653] 
\draw    (267,105) -- (277,105) ;

% Text Node
\draw (117,92.4) node [anchor=north west][inner sep=0.75pt]  [font=\tiny]  {$2$};
% Text Node
\draw (107,117.4) node [anchor=north west][inner sep=0.75pt]  [font=\tiny]  {$4$};
% Text Node
\draw (21,125.4) node [anchor=north west][inner sep=0.75pt]  [font=\tiny]  {$\frac{1}{4}$};
% Text Node
\draw (37,121.4) node [anchor=north west][inner sep=0.75pt]  [font=\tiny]  {$2$};
% Text Node
\draw (77,122.4) node [anchor=north west][inner sep=0.75pt]  [font=\tiny]  {$2$};
% Text Node
\draw (117,142.4) node [anchor=north west][inner sep=0.75pt]  [font=\tiny]  {$2$};
% Text Node
\draw    (143,127) -- (162,127) -- (162,142) -- (143,142) -- cycle  ;
\draw (152.5,134.5) node  [font=\tiny]  {$v_{13}$};
% Text Node
\draw (321,110.9) node [anchor=north west][inner sep=0.75pt]  [font=\scriptsize]  {$=$};
% Text Node
\draw (365.53,117.39) node    {$\begin{cases}
 & \\
 & 
\end{cases}$};
% Text Node
\draw    (370.98,96.6) -- (389.98,96.6) -- (389.98,111.6) -- (370.98,111.6) -- cycle  ;
\draw (380.48,104.1) node  [font=\tiny]  {$v_{13}$};
% Text Node
\draw (450.53,117.39) node    {$\begin{cases}
 & \\
 & 
\end{cases}$};
% Text Node
\draw    (455.22,96.6) -- (473.22,96.6) -- (473.22,111.6) -- (455.22,111.6) -- cycle  ;
\draw (464.22,104.1) node  [font=\tiny]  {$v_{2}$};
% Text Node
\draw (417.22,115.64) node  [font=\scriptsize]  {$+$};
% Text Node
\draw (379.5,126.5) node  [font=\tiny]  {$c_{13}$};
% Text Node
\draw (462.5,127.5) node  [font=\tiny]  {$\frac{c_{2}}{8} +c_{13}^{-}$};
% Text Node
\draw (539.53,117.39) node    {$\begin{cases}
 & \\
 & 
\end{cases}$};
% Text Node
\draw    (543.72,95.6) -- (562.72,95.6) -- (562.72,112.6) -- (543.72,112.6) -- cycle  ;
\draw (553.22,104.1) node  [font=\tiny]  {$v_{13}^{-}$};
% Text Node
\draw (507.22,117.14) node  [font=\scriptsize]  {$+$};
% Text Node
\draw (552.5,127.5) node  [font=\tiny]  {$\frac{c_{2}}{8}$};
% Text Node
\draw (152.22,102.2) node    {$\{$};
% Text Node
\draw    (162.05,97.1) -- (191.05,97.1) -- (191.05,112.1) -- (162.05,112.1) -- cycle  ;
\draw (176.55,104.6) node  [font=\tiny]  {$c_{13} v_{13}$};
% Text Node
\draw    (222.09,97.1) -- (245.09,97.1) -- (245.09,112.1) -- (222.09,112.1) -- cycle  ;
\draw (233.59,104.6) node  [font=\tiny]  {$c_{2} v_{2}$};
% Text Node
\draw (197,98.4) node [anchor=north west][inner sep=0.75pt]  [font=\scriptsize]  {$+$};
% Text Node
\draw    (277.05,96.95) -- (306.05,96.95) -- (306.05,113.95) -- (277.05,113.95) -- cycle  ;
\draw (291.55,105.45) node  [font=\tiny]  {$c_{13}^{-} v_{13}^{-}$};
% Text Node
\draw (251,98.4) node [anchor=north west][inner sep=0.75pt]  [font=\scriptsize]  {$+$};
% Text Node
\draw (313.22,104.2) node    {$\}$};

\end{tikzpicture}}
\end{align}
which occurs when the contributions travels `downwards' in the network in~\eqref{ttn_var}.
We refer to these as `up' and `down' operations, respectively.
The constants $c_{13}, c_{2}, c_{13}^- \geq 0$.
Indeed, for arbitrary $i$ the tensor network contraction corresponding to $\Var[\partial_{1,1}\CostXi{}_{\text{qTTN}}]$ contains a mixture of the `up' operations~\eqref{ent_up} and the `down' operations~\eqref{ent_down}.
In the limit where all operations are `up' (`down') we obtain $\Var[\partial_{1,1}\langle X_N\rangle_{\text{qTTN}}]$ as in Eq.~\eqref{ttn_var_xn} ($\Var[\partial_{1,1}\langle X_1\rangle_{\text{qTTN}}]$ as in Eq.~\eqref{ttn_var_x1}).
We emphasize that $\Var[\partial_{1,1}\CostXi{}_{\text{qTTN}}] \in \Theta(c^{-\log N})$ for arbitrary $i$, since the observable $X_i$ is causally connected to $1+\log N$ qubits and as such, when contracting the resulting tensor network in Eq.~\eqref{ttn_var}, it can only pick up contributions from that many registers.

We show this explicitly in the following Lemma where we prove that $\Var[\partial_{1,1}\langle X_N\rangle_{\text{qTTN}}]$ is a lower-bound to
$\Var[\partial_{1,1}\CostXi{}_{\text{qTTN}}]$ for general $i$.
Hence it is not necessary to compute $\Var[\partial_{1,1}\CostXi{}_{\text{qTTN}}]$ for all $i$ to conclude that the qTTN ansatz does not have exponentially vanishing gradients.
Together with \thref{th_ttn_app} this implies that our qTTN ansatz completely avoids barren plateaus as the gradients only vanish polynomially in $N$, as claimed in \thref{th_ttn} in the main text.

\begin{lemma}
\thlabel{lemma_ttn_app}
Let $\langle X_i\rangle_{\text{qTTN}}$ be the cost function associated with the observable $X_i$ and consider the qTTN ansatz defined in Eq.~\eqref{ttn_ans}, then:
\begin{align}
 \Var[\partial_{1,1}\langle X_N\rangle_{\text{qTTN}}] \leq \Var[\partial_{1,1}\langle X_i\rangle_{\text{qTTN}}] \leq \Var[\partial_{1,1}\langle X_1\rangle_{\text{qTTN}}]
\end{align}
for all $i = 1, \ldots, N$.
\end{lemma}
\begin{proof}
After identifying the $1+\log N$ qubits causally connected with the observable $X_i$, the variance in Eq.~\eqref{ttn_var} reduces to a tensor network containing a mixture of `up' operations~\eqref{ent_up} and `down' operations~\eqref{ent_down}.
The transformation of the coefficients $c_{13}, c_{2}, c_{13}^-$ is determined by
\begin{align}
\label{ent_matrices}
 M_\text{Up}:\begin{bmatrix}
  c_{13}\\ c_{2}\\ c_{13}^-
 \end{bmatrix} \mapsto \begin{bmatrix}
  c_{13} + {c_{13}^-}/{4}\\ {3c_{2}}/{8}\\ 0
 \end{bmatrix} \quad &\Rightarrow \quad M_\text{Up} = \begin{bmatrix}
  1 &0 &1/4\\ 0&3/8&0 \\ 0&0&0
 \end{bmatrix}, \\
 \label{ent_matrices2}
 M_\text{Down}:\begin{bmatrix}
  c_{13}\\ c_{2}\\ c_{13}^-
 \end{bmatrix} \mapsto \begin{bmatrix}
  c_{13}\\ {c_{2}}/{8}+{c_{13}^-}\\ {c_{2}}/{8}
 \end{bmatrix} \quad &\Rightarrow \quad M_\text{Down} = \begin{bmatrix}
  1 &0 &0\\ 0&1/8&1 \\ 0&1/8&0
 \end{bmatrix}.
\end{align}
It suffices to show that the `up' operation~\eqref{ent_up} leads to a smaller contribution to the variance than the `down' operation~\eqref{ent_down}.
This is true because to find $\Var[\partial_{1,1}\langle X_N\rangle_{\text{qTTN}}]$ by contraction of the tensor network in Eq.~\eqref{ttn_var_xn} we need to perform operation~\eqref{ent_up} $n$ times, whereas to find $\Var[\partial_{1,1}\langle X_i\rangle_{\text{qTTN}}]$ for general $i$ by contraction of the tensor network in Eq.~\eqref{ttn_var} we need to perform a mixture of $n$ `up' and `down' operations.
Note that we are ultimately interested in the size of the coefficient $c_2$ as the other two coefficients will be disregarded by the $P_2$ at the top left parameter.
If prior to the `down' operation~\eqref{ent_down} we have an arbitrary vector $c_{13} \vot + c_2 \vt + c_{13}^- \votm$, then looking at Eqs.~\eqref{ent_matrices}, \eqref{ent_matrices2} we find that $M_\text{Up}:c_2 \mapsto \frac{3c_2}{8}$ and $M_\text{Down}: c_2 \mapsto \frac{c_2}{8} + c_{13}^-$ and so we want to prove:
\begin{align}
\label{ent_updown_req}
 \frac{3c_2}{8} \leq \frac{c_2}{8} + c_{13}^- \quad \text{or equivalently} \quad c_2 \leq 4 c_{13}^-.
\end{align}
Referencing the diagram in Eq.~\eqref{ttn_var} we show that this is always satisfied:
Notice that a `down' operation is always preceded by a $\scalebox{0.6}{\tikzset{every picture/.style={line width=0.75pt}} %set default line width to 0.75pt        

\begin{tikzpicture}[x=0.75pt,y=0.75pt,yscale=-1,xscale=1]
%uncomment if require: \path (0,300); %set diagram left start at 0, and has height of 300

%Shape: Output [id:dp5704201301128986] 
\draw  [fill={rgb, 255:red, 0; green, 0; blue, 0 }  ,fill opacity=1 ] (85,134.96) .. controls (85,137.72) and (82.76,139.96) .. (80,139.96) .. controls (77.24,139.96) and (75,137.72) .. (75,134.96) .. controls (75,132.2) and (77.24,129.96) .. (80,129.96) .. controls (82.76,129.96) and (85,132.2) .. (85,134.96) -- cycle (90,134.96) -- (85,134.96) (70,134.96) -- (75,134.96) ;
%Shape: Output [id:dp7155959758454091] 
\draw   (65,135) .. controls (65,137.76) and (62.76,140) .. (60,140) .. controls (57.24,140) and (55,137.76) .. (55,135) .. controls (55,132.24) and (57.24,130) .. (60,130) .. controls (62.76,130) and (65,132.24) .. (65,135) -- cycle (70,135) -- (65,135) (50,135) -- (55,135) ;

% Text Node
\draw (77,122.4) node [anchor=north west][inner sep=0.75pt]  [font=\tiny]  {$2$};

\end{tikzpicture}}$ on the top wire, i.e.
\begin{align}
\label{zx_id8}
 \scalebox{0.9}{\input{tikz_files/zx_id8}}
\end{align}
for some constants $c_{13}',c'_2, c^{-\prime}_{13}$.
Contracting the $\scalebox{0.6}{\tikzset{every picture/.style={line width=0.75pt}} %set default line width to 0.75pt        

\begin{tikzpicture}[x=0.75pt,y=0.75pt,yscale=-1,xscale=1]
%uncomment if require: \path (0,300); %set diagram left start at 0, and has height of 300

%Shape: Output [id:dp5704201301128986] 
\draw  [fill={rgb, 255:red, 0; green, 0; blue, 0 }  ,fill opacity=1 ] (85,134.96) .. controls (85,137.72) and (82.76,139.96) .. (80,139.96) .. controls (77.24,139.96) and (75,137.72) .. (75,134.96) .. controls (75,132.2) and (77.24,129.96) .. (80,129.96) .. controls (82.76,129.96) and (85,132.2) .. (85,134.96) -- cycle (90,134.96) -- (85,134.96) (70,134.96) -- (75,134.96) ;
%Shape: Output [id:dp7155959758454091] 
\draw   (65,135) .. controls (65,137.76) and (62.76,140) .. (60,140) .. controls (57.24,140) and (55,137.76) .. (55,135) .. controls (55,132.24) and (57.24,130) .. (60,130) .. controls (62.76,130) and (65,132.24) .. (65,135) -- cycle (70,135) -- (65,135) (50,135) -- (55,135) ;

% Text Node
\draw (77,122.4) node [anchor=north west][inner sep=0.75pt]  [font=\tiny]  {$2$};

\end{tikzpicture}}$ on the right hand side gives
\begin{align}
 \scalebox{0.9}{\tikzset{every picture/.style={line width=0.75pt}} %set default line width to 0.75pt        

\begin{tikzpicture}[x=0.75pt,y=0.75pt,yscale=-1,xscale=1]
%uncomment if require: \path (0,300); %set diagram left start at 0, and has height of 300

%Straight Lines [id:da1319947172270286] 
\draw    (92.78,124.8) -- (102.78,124.8) ;
%Straight Lines [id:da2569253606222941] 
\draw    (152.78,124.8) -- (162.78,124.8) ;
%Straight Lines [id:da7329361269625565] 
\draw    (207.78,124.8) -- (217.78,124.8) ;
%Shape: Output [id:dp02030880968433757] 
\draw  [fill={rgb, 255:red, 0; green, 0; blue, 0 }  ,fill opacity=1 ] (75,124.96) .. controls (75,127.72) and (72.76,129.96) .. (70,129.96) .. controls (67.24,129.96) and (65,127.72) .. (65,124.96) .. controls (65,122.2) and (67.24,119.96) .. (70,119.96) .. controls (72.76,119.96) and (75,122.2) .. (75,124.96) -- cycle (80,124.96) -- (75,124.96) (60,124.96) -- (65,124.96) ;
%Shape: Output [id:dp9045929943420636] 
\draw   (55,125) .. controls (55,127.76) and (52.76,130) .. (50,130) .. controls (47.24,130) and (45,127.76) .. (45,125) .. controls (45,122.24) and (47.24,120) .. (50,120) .. controls (52.76,120) and (55,122.24) .. (55,125) -- cycle (60,125) -- (55,125) (40,125) -- (45,125) ;

% Text Node
\draw (90,123) node    {$\{$};
% Text Node
\draw    (102.84,116.91) -- (131.84,116.91) -- (131.84,131.91) -- (102.84,131.91) -- cycle  ;
\draw (117.34,124.41) node  [font=\tiny]  {$c'_{13} v_{13}$};
% Text Node
\draw    (162.87,116.91) -- (185.87,116.91) -- (185.87,131.91) -- (162.87,131.91) -- cycle  ;
\draw (174.37,124.41) node  [font=\tiny]  {$c'_{2} v_{2}$};
% Text Node
\draw (137.78,118.2) node [anchor=north west][inner sep=0.75pt]  [font=\scriptsize]  {$+$};
% Text Node
\draw    (217.84,116.76) -- (248.84,116.76) -- (248.84,133.76) -- (217.84,133.76) -- cycle  ;
\draw (233.34,125.26) node  [font=\tiny]  {$c_{13}^{-\prime} v_{13}^{-}$};
% Text Node
\draw (191.78,118.2) node [anchor=north west][inner sep=0.75pt]  [font=\scriptsize]  {$+$};
% Text Node
\draw (256,124) node    {$\}$};
% Text Node
\draw (67,112.4) node [anchor=north west][inner sep=0.75pt]  [font=\tiny]  {$2$};

\end{tikzpicture}} = c_{13}' \vot + \Big(\frac{c_2'}{2} + c^{-\prime}_{13}\Big) \vt + \frac{c'_2}{2} \votm,
\end{align}
which reduces requirement~\eqref{ent_updown_req} to
\begin{align}
\label{ent_updown_req2}
\Big(\frac{c_2'}{2} + c^{-\prime}_{13} \Big) \leq 4 \cdot \Big(\frac{c'_2}{2}\Big) \quad \text{or equivalently} \quad 2 c_{13}^{-\prime} \leq 3c_2'.
\end{align}
To analyse these we check the operation preceding the right hand side of Eq.~\eqref{zx_id8} leading to the constants $c_{13}',c'_2, c^{-\prime}_{13}$, which is either the `up' operation~\eqref{ent_up}, the `down' operation~\eqref{ent_down} or the observable $u_i$ itself: 
\begin{itemize}
 \item If `up', then Eq.~\eqref{ent_matrices} implies $c^{-\prime}_{13} = 0$ which trivially satisfies requirement~\eqref{ent_updown_req2}.
 \item If `down', then we consider the vector $c_{13}^{\prime\prime} \vot + c_2^{\prime\prime} \vt + c_{13}^{-\prime\prime} \votm$ which precedes it:
  \begin{align}
   c_{13}' \vot + c_2' \vt + c_{13}^{-\prime} \votm = \scalebox{0.9}{\tikzset{every picture/.style={line width=0.75pt}} %set default line width to 0.75pt        

\begin{tikzpicture}[x=0.75pt,y=0.75pt,yscale=-1,xscale=1]
%uncomment if require: \path (0,300); %set diagram left start at 0, and has height of 300

%Shape: Output [id:dp3408170284567855] 
\draw  [fill={rgb, 255:red, 0; green, 0; blue, 0 }  ,fill opacity=1 ] (125,105) .. controls (125,107.76) and (122.76,110) .. (120,110) .. controls (117.24,110) and (115,107.76) .. (115,105) .. controls (115,102.24) and (117.24,100) .. (120,100) .. controls (122.76,100) and (125,102.24) .. (125,105) -- cycle (130,105) -- (125,105) (110,105) -- (115,105) ;
%Shape: Output [id:dp7294312557244427] 
\draw   (105,105.04) .. controls (105,107.8) and (102.76,110.04) .. (100,110.04) .. controls (97.24,110.04) and (95,107.8) .. (95,105.04) .. controls (95,102.28) and (97.24,100.04) .. (100,100.04) .. controls (102.76,100.04) and (105,102.28) .. (105,105.04) -- cycle (110,105.04) -- (105,105.04) (90,105.04) -- (95,105.04) ;
%Shape: Output [id:dp057256208346931015] 
\draw   (120,115) .. controls (122.76,115) and (125,117.24) .. (125,120) .. controls (125,122.76) and (122.76,125) .. (120,125) .. controls (117.24,125) and (115,122.76) .. (115,120) .. controls (115,117.24) and (117.24,115) .. (120,115) -- cycle (120,110) -- (120,115) (120,130) -- (120,125) ;
%Straight Lines [id:da18862635100570535] 
\draw    (130,105) -- (143,105) ;
%Shape: Output [id:dp20914042527809062] 
\draw  [fill={rgb, 255:red, 0; green, 0; blue, 0 }  ,fill opacity=1 ] (85,134.96) .. controls (85,137.72) and (82.76,139.96) .. (80,139.96) .. controls (77.24,139.96) and (75,137.72) .. (75,134.96) .. controls (75,132.2) and (77.24,129.96) .. (80,129.96) .. controls (82.76,129.96) and (85,132.2) .. (85,134.96) -- cycle (90,134.96) -- (85,134.96) (70,134.96) -- (75,134.96) ;
%Shape: Output [id:dp427038956698234] 
\draw   (65,135) .. controls (65,137.76) and (62.76,140) .. (60,140) .. controls (57.24,140) and (55,137.76) .. (55,135) .. controls (55,132.24) and (57.24,130) .. (60,130) .. controls (62.76,130) and (65,132.24) .. (65,135) -- cycle (70,135) -- (65,135) (50,135) -- (55,135) ;
%Shape: Circle [id:dp6158979858501483] 
\draw  [fill={rgb, 255:red, 0; green, 0; blue, 0 }  ,fill opacity=1 ] (45.01,133.94) .. controls (45.01,131.46) and (42.99,129.44) .. (40.51,129.44) .. controls (38.02,129.44) and (36.01,131.46) .. (36.01,133.94) .. controls (36.01,136.43) and (38.02,138.44) .. (40.51,138.44) .. controls (42.99,138.44) and (45.01,136.43) .. (45.01,133.94) -- cycle ;
%Straight Lines [id:da6557734418196166] 
\draw    (45.01,134.94) -- (50,135) ;
%Shape: Output [id:dp191663508935507] 
\draw  [fill={rgb, 255:red, 0; green, 0; blue, 0 }  ,fill opacity=1 ] (125,135) .. controls (125,137.76) and (122.76,140) .. (120,140) .. controls (117.24,140) and (115,137.76) .. (115,135) .. controls (115,132.24) and (117.24,130) .. (120,130) .. controls (122.76,130) and (125,132.24) .. (125,135) -- cycle (130,135) -- (125,135) (110,135) -- (115,135) ;
%Shape: Output [id:dp8938092524785841] 
\draw   (105,135) .. controls (105,137.76) and (102.76,140) .. (100,140) .. controls (97.24,140) and (95,137.76) .. (95,135) .. controls (95,132.24) and (97.24,130) .. (100,130) .. controls (102.76,130) and (105,132.24) .. (105,135) -- cycle (110,135) -- (105,135) (90,135) -- (95,135) ;
%Straight Lines [id:da8631135268171464] 
\draw    (130,135) -- (143,135) ;
%Straight Lines [id:da23911778286925922] 
\draw    (152,105) -- (162,105) ;
%Straight Lines [id:da4685834311250041] 
\draw    (212,105) -- (222,105) ;
%Straight Lines [id:da5301706555667431] 
\draw    (267,105) -- (277,105) ;

% Text Node
\draw (117,92.4) node [anchor=north west][inner sep=0.75pt]  [font=\tiny]  {$2$};
% Text Node
\draw (107,117.4) node [anchor=north west][inner sep=0.75pt]  [font=\tiny]  {$4$};
% Text Node
\draw (21,125.4) node [anchor=north west][inner sep=0.75pt]  [font=\tiny]  {$\frac{1}{4}$};
% Text Node
\draw (37,121.4) node [anchor=north west][inner sep=0.75pt]  [font=\tiny]  {$2$};
% Text Node
\draw (77,122.4) node [anchor=north west][inner sep=0.75pt]  [font=\tiny]  {$2$};
% Text Node
\draw (117,142.4) node [anchor=north west][inner sep=0.75pt]  [font=\tiny]  {$2$};
% Text Node
\draw    (143,127) -- (162,127) -- (162,142) -- (143,142) -- cycle  ;
\draw (152.5,134.5) node  [font=\tiny]  {$v_{13}$};
% Text Node
\draw (151.22,102.2) node    {$\{$};
% Text Node
\draw    (161.55,97.1) -- (195.55,97.1) -- (195.55,112.1) -- (161.55,112.1) -- cycle  ;
\draw (178.55,104.6) node  [font=\tiny]  {$c_{13} ''v_{13}$};
% Text Node
\draw    (221.59,97.1) -- (249.59,97.1) -- (249.59,112.1) -- (221.59,112.1) -- cycle  ;
\draw (235.59,104.6) node  [font=\tiny]  {$c_{2} ''v_{2}$};
% Text Node
\draw (197,98.4) node [anchor=north west][inner sep=0.75pt]  [font=\scriptsize]  {$+$};
% Text Node
\draw    (276.55,96.95) -- (310.55,96.95) -- (310.55,113.95) -- (276.55,113.95) -- cycle  ;
\draw (293.55,105.45) node  [font=\tiny]  {$c_{13}^{-\prime\prime} v_{13}^{-}$};
% Text Node
\draw (251,98.4) node [anchor=north west][inner sep=0.75pt]  [font=\scriptsize]  {$+$};
% Text Node
\draw (316.22,104.2) node    {$\}$};

\end{tikzpicture}}
  \end{align}
  Equation~\eqref{ent_matrices2} implies: 
  \begin{align}
   c_{13}' = c_{13}^{\prime\prime}, \quad\quad c_2' = \frac{c_2^{\prime\prime}}{8} + c_{13}^{-\prime\prime}, \quad\quad c_{13}^{-\prime} = \frac{c_2^{\prime\prime}}{8},
  \end{align}
  in which case $2 c_{13}^{-\prime} = \frac{2c_2^{\prime\prime}}{8}$ and $3 c_2' = \frac{3c_2^{\prime\prime}}{8} + 3 c_{13}^{-\prime\prime}$, so that requirement~\eqref{ent_updown_req2} is satisfied. 
 \item Lastly, if the original `down' is connected to the observable, we have
  \begin{align}
   \scalebox{0.9}{\input{tikz_files/zx_id11}}.
  \end{align}
  The identities in~\eqref{zx_identities1} imply that $c_{13}^{-\prime} = c_2' = \frac{1}{2}$ and so requirement~\eqref{ent_updown_req2} is satisfied.
\end{itemize}
This concludes the proof.
\end{proof}

Note that the inequality in the Lemma is just a consequence of the order of the qubits in the construction of the ansatz in Fig.~\ref{fig:ttn_circ}.
We expect that the variance of any pair of observables $X_{i_1}, X_{i_2}$ in the qTTN ansatz vanishes identically modulo a different base, i.e.\ $\Var[\partial_{1,1}\langle X_i\rangle_{\text{qTTN}}]$ vanishes as  $c_i^{\log N}$ for each $i$ and some $c_i > 0$.
Indeed, we could achieve any other register ordering through SWAP gates without affecting the overall trainability of the ansatz.

Using the tools of the proof of \thref{lemma_ttn_app} we now discuss \thref{prop_ttn}.
We have established that observable $X_i$ is causally connected to $(1+\log N)$ qubits and so it follows that a $k$-local observable $X_I$ acting on qubits $I=\{i_1,\ldots,i_k\}$ is causally connected to at most $k(1+\log N)$ qubits.
This bound is usually not tight since there can be some overlap between the qubits causally connected to pairs $X_{i_{j_1}},X_{i_{j_2}}$ for $i_{j_1},i_{j_2} \in I$.
Since \thref{lemma_ttn_app} implies that contributions to the variance $\Var[\partial_{1,1}\CostXi{}]$ are smallest when the observable is $X_N$ (i.e.\ when Eq.~\eqref{ttn_var} contracts to only `up' operations~\eqref{ent_up}), we argue that 
\begin{align}
\label{ttn_kbounds}
 \Var[\partial_{1,1}\langle X_I\rangle_{\text{qTTN}}] \geq \Var[\partial_{1,1}\langle X_{\hat{N}}\rangle_{\text{qTTN}}] = \frac{1}{4} \cdot \Big(\frac{3}{8}\Big)^{k \log N},
\end{align}
where $\hat{N} = 2^{k\log N}$.
Rather than having `up' and `down' operations from $X_{i_{j}}$ for each $i_j\in I$ (cf. Eq.~\eqref{ttn_var} but
with $\mathbf{v_2}$ in registers $I=\{i_1,\ldots,i_k\}$ and $\mathbf{v_{13}}$ elsewhere), the contributions to
the variance are smallest when contracting a network as in Eq.~\eqref{ttn_var_xn} that has $k \log N$ registers
instead where there are only `up' operations from the observable $X_{\hat{N}}$.
Note that in the regime $k \approx N$ the number of qubits causally connected to $X_I$ is $N$ and we obtain exponentially vanishing gradients.

\newpage

\section{Quantum multiscale entanglement renormalization ansatz}
\label{app_mera}

For $8 = 2^3$ qubits the qMERA ansatz in Eq.~\eqref{mera_ans1} can also be represented as
\begin{figure}[h]
\centering
\scalebox{0.85}{\input{tikz_files/mera_circ}}
\caption{\label{fig:mera_circ}
qMERA circuit considered in this article for $8 = 2^3$ qubits and $3$ layers.
For arbitrary $N = 2^n$ qubits and $n$ layers, the $l$-th course-graining layer (CG in the figure) is as in the qTTN ansatz whilst the $l$-th disentangling layer is a composition of the $(l-1)$-th disentangling layer with additional $R^{(j_1,j_2)}_X R^{(j_1,j_2)}_Z CNOT$ acting on adjacent pairs of the newly added qubits $(j_1,j_2)$ within that layer (e.g.\ in the last disentangling operation above the last CNOT gates act on qubits $(2,4)$ and $(6,8)$ which were added on the last layer of the qMERA).
}
\end{figure}

\noindent Note that this circuit is equivalent to the one presented in Eq.~\eqref{mera_ans1} up to a reordering of the qubits.
The qubits in Fig.~\ref{fig:mera_circ} are arranged so that the coarse-graining operations are equivalent to the ones in the qTTN PQC in Fig.~\ref{fig:ttn_circ}.
To that end we redefine the qMERA circuit as a product of course-graining and disentangling layers as
\begin{align}
\label{mera_ans}
 U^{\text{qMERA}} := \prod _{l=n}^1 U^\text{DIS}_l \cdot U^\text{CG}_l.
\end{align}
For all $l \geq 1$ the coarse-graining layers are
\begin{align}
 \scalebox{1}{\tikzset{every picture/.style={line width=0.75pt}} %set default line width to 0.75pt        

\begin{tikzpicture}[x=0.75pt,y=0.75pt,yscale=-1,xscale=1]
%uncomment if require: \path (0,454); %set diagram left start at 0, and has height of 454

%Straight Lines [id:da5317108651760718] 
\draw    (145,130) -- (170,130) ;
%Shape: Rectangle [id:dp4577471512583433] 
\draw   (170,110) -- (205,110) -- (205,150) -- (170,150) -- cycle ;

%Straight Lines [id:da26877991336101514] 
\draw    (205,130) -- (212.56,130) -- (220,130) ;
%Straight Lines [id:da8002999661472217] 
\draw    (250,100) -- (276,100) ;
%Straight Lines [id:da8595246468453102] 
\draw    (250,141) -- (276,140.91) ;
%Straight Lines [id:da3610283492706192] 
\draw    (294,100) -- (305,100) ;
%Straight Lines [id:da6198090458312475] 
\draw    (294,140.91) -- (305,140.91) ;
%Straight Lines [id:da33816189520839246] 
\draw    (324,100) -- (356,100) ;
%Straight Lines [id:da09583505920553459] 
\draw    (340,100) -- (340,138.56) ;
\draw [shift={(340,140.91)}, rotate = 90] [color={rgb, 255:red, 0; green, 0; blue, 0 }  ][line width=0.75]      (0, 0) circle [x radius= 3.35, y radius= 3.35]   ;
%Straight Lines [id:da61164682569786] 
\draw    (340,137.91) -- (340,143.91) ;
%Straight Lines [id:da6647179341943135] 
\draw    (324,140.91) -- (356,140.91) ;
%Straight Lines [id:da7141245118966899] 
\draw    (250,120) -- (356,120) ;
%Straight Lines [id:da3235224882064198] 
\draw    (250,160) -- (356,160) ;
%Straight Lines [id:da22251297256257296] 
\draw    (150,125) -- (160,135) ;
%Straight Lines [id:da06227098560608013] 
\draw    (255,115) -- (265,125) ;
%Straight Lines [id:da5853696542417983] 
\draw    (255,155) -- (265,165) ;
%Shape: Brace [id:dp4786194352226636] 
\draw   (362,168.89) .. controls (366.67,168.84) and (368.98,166.49) .. (368.93,161.82) -- (368.71,139.15) .. controls (368.64,132.48) and (370.94,129.13) .. (375.61,129.09) .. controls (370.94,129.13) and (368.58,125.82) .. (368.51,119.15)(368.54,122.15) -- (368.29,96.49) .. controls (368.24,91.82) and (365.89,89.51) .. (361.22,89.56) ;

% Text Node
\draw (224,120.4) node [anchor=north west][inner sep=0.75pt]    {$=$};
% Text Node
\draw    (276.35,93.09) -- (294.35,93.09) -- (294.35,108.09) -- (276.35,108.09) -- cycle  ;
\draw (285.35,100.59) node  [font=\tiny]  {$R_{X}$};
% Text Node
\draw    (276.35,133) -- (294.35,133) -- (294.35,148) -- (276.35,148) -- cycle  ;
\draw (285.35,140.5) node  [font=\tiny]  {$R_{X}$};
% Text Node
\draw    (306.22,93.09) -- (324.22,93.09) -- (324.22,108.09) -- (306.22,108.09) -- cycle  ;
\draw (315.22,100.59) node  [font=\tiny]  {$R_{Z}$};
% Text Node
\draw    (306.22,133) -- (324.22,133) -- (324.22,148) -- (306.22,148) -- cycle  ;
\draw (315.22,140.5) node  [font=\tiny]  {$R_{Z}$};
% Text Node
\draw (188.89,128.94) node  [font=\footnotesize]  {$U_{l}^{\text{CG}}$};
% Text Node
\draw (146,115.4) node [anchor=north west][inner sep=0.75pt]  [font=\tiny]  {$2^{n}$};
% Text Node
\draw (251,105.4) node [anchor=north west][inner sep=0.75pt]  [font=\tiny]  {$2^{n-l}-1$};
% Text Node
\draw (251,145.4) node [anchor=north west][inner sep=0.75pt]  [font=\tiny]  {$2^{n-l}-1$};
% Text Node
\draw (396.89,128.35) node  [font=\scriptsize] [align=left] {repeat};

\end{tikzpicture}}
\end{align}
where the four registers are composed $2^{l-1}$ times in parallel (vertically).
The disentangling layers are as follows:
For $l = 1$ $U_1^{DIS} = I$ and for $l \geq 2$
\begin{align}
 \scalebox{1}{\tikzset{every picture/.style={line width=0.75pt}} %set default line width to 0.75pt        

\begin{tikzpicture}[x=0.75pt,y=0.75pt,yscale=-1,xscale=1]
%uncomment if require: \path (0,454); %set diagram left start at 0, and has height of 454

%Straight Lines [id:da9933526479609267] 
\draw    (145,130) -- (170,130) ;
%Shape: Rectangle [id:dp9819268492697211] 
\draw   (170,110) -- (205,110) -- (205,150) -- (170,150) -- cycle ;

%Straight Lines [id:da360711461498995] 
\draw    (205,130) -- (212.56,130) -- (220,130) ;
%Straight Lines [id:da23823985950281812] 
\draw    (300,110) -- (336,110) ;
%Straight Lines [id:da7844261715559391] 
\draw    (300,151) -- (336,150.91) ;
%Straight Lines [id:da9425759340551629] 
\draw    (354,110) -- (365,110) ;
%Straight Lines [id:da37319331154329927] 
\draw    (354,150.91) -- (365,150.91) ;
%Straight Lines [id:da29326342318046605] 
\draw    (384,110) -- (416,110) ;
%Straight Lines [id:da8114052839275481] 
\draw    (400,110) -- (400,148.56) ;
\draw [shift={(400,150.91)}, rotate = 90] [color={rgb, 255:red, 0; green, 0; blue, 0 }  ][line width=0.75]      (0, 0) circle [x radius= 3.35, y radius= 3.35]   ;
%Straight Lines [id:da6760248139703029] 
\draw    (400,147.91) -- (400,153.91) ;
%Straight Lines [id:da6859547855224053] 
\draw    (384,150.91) -- (416,150.91) ;
%Straight Lines [id:da2772363742946746] 
\draw    (300,130) -- (416,130) ;
%Straight Lines [id:da5436408977515543] 
\draw    (300,170) -- (416,170) ;
%Straight Lines [id:da44012758018050113] 
\draw    (150,125) -- (160,135) ;
%Straight Lines [id:da9598926372050618] 
\draw    (315,125) -- (325,135) ;
%Straight Lines [id:da8074241036750838] 
\draw    (315,165) -- (325,175) ;
%Shape: Brace [id:dp18291164462345266] 
\draw   (422,178.89) .. controls (426.67,178.84) and (428.98,176.49) .. (428.93,171.82) -- (428.71,149.15) .. controls (428.64,142.48) and (430.94,139.13) .. (435.61,139.09) .. controls (430.94,139.13) and (428.58,135.82) .. (428.51,129.15)(428.54,132.15) -- (428.29,106.49) .. controls (428.24,101.82) and (425.89,99.51) .. (421.22,99.56) ;
%Straight Lines [id:da589340851482268] 
\draw    (300,89) -- (416,89) ;
%Straight Lines [id:da5660876089008953] 
\draw    (315,84) -- (325,94) ;
%Shape: Rectangle [id:dp6185048503360366] 
\draw   (265,80) -- (300,80) -- (300,180) -- (265,180) -- cycle ;

%Straight Lines [id:da2792291347160256] 
\draw    (250,130) -- (265,130) ;

% Text Node
\draw (224,120.4) node [anchor=north west][inner sep=0.75pt]    {$=$};
% Text Node
\draw    (336.35,103.09) -- (354.35,103.09) -- (354.35,118.09) -- (336.35,118.09) -- cycle  ;
\draw (345.35,110.59) node  [font=\tiny]  {$R_{X}$};
% Text Node
\draw    (336.35,143) -- (354.35,143) -- (354.35,158) -- (336.35,158) -- cycle  ;
\draw (345.35,150.5) node  [font=\tiny]  {$R_{X}$};
% Text Node
\draw    (366.22,103.09) -- (384.22,103.09) -- (384.22,118.09) -- (366.22,118.09) -- cycle  ;
\draw (375.22,110.59) node  [font=\tiny]  {$R_{Z}$};
% Text Node
\draw    (366.22,143) -- (384.22,143) -- (384.22,158) -- (366.22,158) -- cycle  ;
\draw (375.22,150.5) node  [font=\tiny]  {$R_{Z}$};
% Text Node
\draw (146,115.4) node [anchor=north west][inner sep=0.75pt]  [font=\tiny]  {$2^{n}$};
% Text Node
\draw (310,118.4) node [anchor=north west][inner sep=0.75pt]  [font=\tiny]  {$2^{n-l+1} -1$};
% Text Node
\draw (310,158.4) node [anchor=north west][inner sep=0.75pt]  [font=\tiny]  {$2^{n-l+1} -1$};
% Text Node
\draw (456.89,138.35) node  [font=\scriptsize] [align=left] {repeat};
% Text Node
\draw (188.89,128.94) node  [font=\footnotesize]  {$U_{l}^{\text{DIS}}$};
% Text Node
\draw (311,76.4) node [anchor=north west][inner sep=0.75pt]  [font=\tiny]  {$2^{n-l}$};
% Text Node
\draw (283.89,127.34) node  [font=\footnotesize]  {$U_{l}^{\text{DIS}}$};

\end{tikzpicture}}
\end{align}
where the rotations and CNOT operation in the last $4$ registers are repeatedly composed in parallel until $N = 2^n$ qubits are reached.

Back to the form in Eq.~\eqref{mera_ans}, \thref{bp_theorem_zx} and App.~\ref{app_zx_meth} imply that $\Var[\partial_{1,1}\langle H\rangle_{\text{qMERA}}]$ for the 16 qubit qMERA PQC in Eq.~\eqref{mera_ans} is given by the tensor network
\begin{align}
\label{mera_var}
 \scalebox{0.8}{\input{tikz_files/mera_tn2v2}}
\end{align}
where the vectors (and registers) are numbered according to the order in which they are added to the qMERA and where the matrix $\Tilde{U}$ is
\begin{align}
\label{mera_bm}
 \scalebox{0.9}{\tikzset{every picture/.style={line width=0.75pt}} %set default line width to 0.75pt        

\begin{tikzpicture}[x=0.75pt,y=0.75pt,yscale=-1,xscale=1]
%uncomment if require: \path (0,351); %set diagram left start at 0, and has height of 351

%Shape: Output [id:dp8529785922183584] 
\draw  [fill={rgb, 255:red, 0; green, 0; blue, 0 }  ,fill opacity=1 ] (284.99,124.92) .. controls (284.99,127.68) and (282.75,129.92) .. (279.99,129.92) .. controls (277.23,129.92) and (274.99,127.68) .. (274.99,124.92) .. controls (274.99,122.16) and (277.23,119.92) .. (279.99,119.92) .. controls (282.75,119.92) and (284.99,122.16) .. (284.99,124.92) -- cycle (289.99,124.92) -- (284.99,124.92) (269.99,124.92) -- (274.99,124.92) ;
%Shape: Output [id:dp13064553088974518] 
\draw   (264.99,124.96) .. controls (264.99,127.72) and (262.75,129.96) .. (259.99,129.96) .. controls (257.23,129.96) and (254.99,127.72) .. (254.99,124.96) .. controls (254.99,122.2) and (257.23,119.96) .. (259.99,119.96) .. controls (262.75,119.96) and (264.99,122.2) .. (264.99,124.96) -- cycle (269.99,124.96) -- (264.99,124.96) (249.99,124.96) -- (254.99,124.96) ;
%Straight Lines [id:da8439266845359406] 
\draw    (245,124.9) -- (249.99,124.96) ;
%Shape: Output [id:dp3189699243850428] 
\draw  [fill={rgb, 255:red, 0; green, 0; blue, 0 }  ,fill opacity=1 ] (325,125) .. controls (325,127.76) and (322.76,130) .. (320,130) .. controls (317.24,130) and (315,127.76) .. (315,125) .. controls (315,122.24) and (317.24,120) .. (320,120) .. controls (322.76,120) and (325,122.24) .. (325,125) -- cycle (330,125) -- (325,125) (310,125) -- (315,125) ;
%Shape: Output [id:dp18767043652946303] 
\draw   (305,125) .. controls (305,127.76) and (302.76,130) .. (300,130) .. controls (297.24,130) and (295,127.76) .. (295,125) .. controls (295,122.24) and (297.24,120) .. (300,120) .. controls (302.76,120) and (305,122.24) .. (305,125) -- cycle (310,125) -- (305,125) (290,125) -- (295,125) ;
%Shape: Output [id:dp9919231665802908] 
\draw   (320,105) .. controls (322.76,105) and (325,107.24) .. (325,110) .. controls (325,112.76) and (322.76,115) .. (320,115) .. controls (317.24,115) and (315,112.76) .. (315,110) .. controls (315,107.24) and (317.24,105) .. (320,105) -- cycle (320,100) -- (320,105) (320,120) -- (320,115) ;
%Shape: Output [id:dp1486717369020405] 
\draw  [fill={rgb, 255:red, 0; green, 0; blue, 0 }  ,fill opacity=1 ] (325.01,95.08) .. controls (325.01,97.84) and (322.77,100.08) .. (320.01,100.08) .. controls (317.25,100.08) and (315.01,97.84) .. (315.01,95.08) .. controls (315.01,92.32) and (317.25,90.08) .. (320.01,90.08) .. controls (322.77,90.08) and (325.01,92.32) .. (325.01,95.08) -- cycle (330.01,95.08) -- (325.01,95.08) (310.01,95.08) -- (315.01,95.08) ;
%Shape: Output [id:dp2756192225194962] 
\draw   (305.01,95.08) .. controls (305.01,97.84) and (302.77,100.08) .. (300.01,100.08) .. controls (297.25,100.08) and (295.01,97.84) .. (295.01,95.08) .. controls (295.01,92.32) and (297.25,90.08) .. (300.01,90.08) .. controls (302.77,90.08) and (305.01,92.32) .. (305.01,95.08) -- cycle (310.01,95.08) -- (305.01,95.08) (290.01,95.08) -- (295.01,95.08) ;
%Shape: Output [id:dp32563789666255216] 
\draw   (345,95.04) .. controls (345,97.8) and (342.76,100.04) .. (340,100.04) .. controls (337.24,100.04) and (335,97.8) .. (335,95.04) .. controls (335,92.28) and (337.24,90.04) .. (340,90.04) .. controls (342.76,90.04) and (345,92.28) .. (345,95.04) -- cycle (350,95.04) -- (345,95.04) (330,95.04) -- (335,95.04) ;
%Straight Lines [id:da6647088857804566] 
\draw    (330,125) -- (370,125) ;
%Straight Lines [id:da24825846015020514] 
\draw    (245,95) -- (290.01,95.08) ;
%Straight Lines [id:da6965937058193503] 
\draw    (200,100) -- (210,100) ;
%Straight Lines [id:da3165502604871202] 
\draw    (200,120) -- (210,120) ;
%Straight Lines [id:da8761427747553698] 
\draw    (156,120) -- (166,120) ;
%Straight Lines [id:da9969366504423518] 
\draw    (156,100) -- (166,100) ;
%Shape: Output [id:dp46418351960844495] 
\draw  [fill={rgb, 255:red, 0; green, 0; blue, 0 }  ,fill opacity=1 ] (365,95) .. controls (365,97.76) and (362.76,100) .. (360,100) .. controls (357.24,100) and (355,97.76) .. (355,95) .. controls (355,92.24) and (357.24,90) .. (360,90) .. controls (362.76,90) and (365,92.24) .. (365,95) -- cycle (370,95) -- (365,95) (350,95) -- (355,95) ;

% Text Node
\draw (277,112.4) node [anchor=north west][inner sep=0.75pt]  [font=\tiny]  {$2$};
% Text Node
\draw (317,132.4) node [anchor=north west][inner sep=0.75pt]  [font=\tiny]  {$2$};
% Text Node
\draw (307,107.4) node [anchor=north west][inner sep=0.75pt]  [font=\tiny]  {$4$};
% Text Node
\draw (317,82.4) node [anchor=north west][inner sep=0.75pt]  [font=\tiny]  {$2$};
% Text Node
\draw (215,104.4) node [anchor=north west][inner sep=0.75pt]  [font=\scriptsize]  {$=$};
% Text Node
\draw    (166,98) -- (200,98) -- (200,122) -- (166,122) -- cycle  ;
\draw (183,110) node    {$\Tilde{U}$};
% Text Node
\draw (356.78,82.8) node [anchor=north west][inner sep=0.75pt]  [font=\tiny]  {$2$};

\end{tikzpicture}}
\end{align}

To calculate the variance for $1$-local operators of the form $X_i$ we replace $\mathbf{u}_i = \vt$ and $\mathbf{u}_{i' \neq i} = \vot$ and contract the resulting tensor network analytically.
In general, this is an inefficient calculation.

We provide an alternative numerical method that exploits the structure of the qMERA and the causal cone of the observable $X_i$.
In MERA the causal cone of a local observable has bounded width~\cite{GEvenbly09}.
To lower-bound the variance of an observable $X_i$ we want to choose the site $i$ that leads to the widest causal cone.
To upper-bound the variance we want to choose $i$ leading to the narrowest causal cone.
This is done to maximize (minimize) the number of qubit registers in the causal cone of $X_i$.
Depending on the chosen site, the tensor network in Eq.~\eqref{mera_var} can have a causal cone of width $2$ or $3$ as illustrated in Fig.~\ref{fig:mera_causal_cones} for sites $2$ and $11$ respectively.
In general, one finds that the wider causal cones are found by choosing the registers that were added in the last course-graining layer of the qMERA (qubits 9 to 16 in the figure). 

\begin{figure}[h]
\centering
\begin{subfigure}[b]{0.45\textwidth}
\centering
 \scalebox{0.5}{\input{tikz_files/mera_tn2_cone1}}
\end{subfigure}
\hfill
\begin{subfigure}[b]{0.45\textwidth}
\centering
 \scalebox{0.5}{\input{tikz_files/mera_tn2_cone2}}
\end{subfigure}
\caption{\label{fig:mera_causal_cones}
Causal cones for observables $X_2$ and $X_{11}$, respectively.
The width of a causal cone is determined by the largest number of $2$-qubit operations at any depth of the ansatz, e.g.\ the third course-graining layer of the causal cone on the right has three $2$-qubit operations whereas the causal cone on the left never has more than two $2$-qubit operations.
}
\end{figure}

Using the arguments of \thref{lemma_ttn_app} we choose observables $X_N$ ($X_1$) for the ansatz in Fig.~\ref{fig:mera_circ} to lower-bound (upper-bound) the gradient variance for one-local observables in qMERA.
The respective variances are calculated by contracting the tensor network in Eq.~\eqref{mera_var} numerically for $2,4,8,16$ qubits, taking care of choosing the correct sites $i = 1, N$ as the ordering of qubits in the networks in Fig.~\ref{fig:mera_circ} and Eq.~\eqref{mera_var} differ.
For the lower bound, $\Var[\partial_{1,1}\langle X_{16}\rangle]$ is given by
\begin{align}
\label{tn_var_simp}
 \scalebox{0.85}{\input{tikz_files/mera_tn2_16q_simp}}
\end{align}
where the vector $w = \frac{1}{2}(\vt + \votm)$. 
In the cases of $8,4,2$ qubits we have
\begin{align}
 \scalebox{0.75}{\input{tikz_files/mera_tn2_842q_simp}}
\end{align}
For the upper bound, $\Var[\partial_{1,1}\langle X_1\rangle_{\text{qMERA}}]$ for $16$ qubits is given by
\begin{align}
 \scalebox{0.85}{\input{tikz_files/mera_tn2_16q_simp_up}}
\end{align}
and for $8,4,2$ qubits by
\begin{align}
\scalebox{0.75}{\input{tikz_files/mera_tn2_842q_simp_up}}
\end{align}
By contracting these networks, we find:
\begin{claim}
\thlabel{th_var_mera_exact}
Let $\langle X_N\rangle_{\text{qMERA}}$ be the cost function associated with the observable $X_N$ and consider the qMERA circuit defined in Eq.~\eqref{mera_ans}, then:
\begin{align}
 \Var[\partial_{1,1}\langle X_N\rangle_{\text{qMERA}}] \approx 
 \begin{cases}
  & 0.09375 \,\, \text{for } N=2\\
  & 0.02477 \,\, \text{for } N=4\\
  & 0.004109 \,\, \text{for } N=8\\
  & 0.000622 \,\, \text{for } N=16
 \end{cases}, \qquad
 \Var[\partial_{1,1}\langle X_1\rangle_{\text{qMERA}}] \approx 
 \begin{cases}
  & 0.1719 \,\, \text{for } N=2\\
  & 0.05242 \,\, \text{for } N=4\\
  & 0.02304 \,\, \text{for } N=8\\
  & 0.00882 \,\, \text{for } N=16
 \end{cases}
\end{align}
\end{claim}
Looking at these results in a log-log plot we find that the data for $N = 4$, $8$ and $16$ lie on straight lines that give us the upper bound scaling like $O(N^{-1.2})$ and the lower bound scaling like $\Omega(N^{-2.7})$.
The numerical results showcase a general brute-force approach to calculating the variances for the proposed qMERA for arbitrary $1$-local observables.

To make a statement for general $N = 2^n$ qubits we argue as in Eq.~\eqref{ttn_kbounds} using the tools from \thref{lemma_ttn_app}. 
\thref{th_mera} states that
\begin{align}
\label{mera_lb}
 \Var[\partial_{1,1}\langle X_N\rangle_{\text{qMERA}}] \geq \Var[\partial_{1,1}\langle X_{\hat{N}}\rangle_{\text{qTTN}}] = \frac{1}{4}\cdot\Big(\frac{3}{8}\Big)^{2\log N}
\end{align}
where $\hat{N} = 2^{2 \log N}$.
Indeed the $1$-local observable $X_i$ in the qMERA circuit is causally connected to at most $2 \log N$ qubits.
\thref{lemma_ttn_app} suggests that, in the form of Fig.~\ref{fig:mera_circ}, the contributions are smallest when carried by the `up' operation~\eqref{ent_up}.
Hence $\Var[\partial_{1,1}\langle X_i\rangle_{\text{qMERA}}]$ is lower-bounded by a circuit analogous to the one in Eq.~\eqref{ttn_var_xn} but with $2 \log N$ qubits instead.

We use the same arguments for $k$-local observables of the form $X_I$ acting on qubits $I = \{i_1,\ldots,i_k\}$ for $k \ll N$ as in \thref{prop_mera}.
The observable $X_I$ is causally connected to an upper bound of $2 k \log N$ qubits and so:
\begin{align}
\label{mera_kbound}
\Var[\partial_{1,1}\langle X_I\rangle_{\text{qMERA}}] \geq \Var[\partial_{1,1}\langle X_{\hat{N}}\rangle_{\text{qTTN}}] = \frac{1}{4}\cdot\Big(\frac{3}{8}\Big)^{2 k \log N}
\end{align}
where $\hat{N} = 2^{2 k \log N}$ by similar arguments as used at the end of App.~\ref{app_ttn}.
These bounds are not tight, but as long as $k \ll N$ they still suggest that the qMERA avoids the barren plateau problem for $k$-local Hamiltonians.
In the limit $k \approx N$ we obtain exponentially vanishing gradients as all qubits are in the causal cone of the observable.

\end{document}